\RequirePackage{snapshot}
\documentclass[11pt,twoside,english,british]{article}
\usepackage[utf8]{inputenc}
\usepackage[a4paper]{geometry}
\usepackage{fancyhdr}
\usepackage{color}
\usepackage{babel}
\usepackage{array}
\usepackage{verbatim}
\usepackage{refstyle}
\usepackage{booktabs}
\usepackage{mathrsfs}
\usepackage{mathtools}
\usepackage{enumitem}
\usepackage{amsmath}
\usepackage{amsthm}
\usepackage{amssymb}
\usepackage{graphicx}
\usepackage{setspace}
\usepackage[authoryear,longnamesfirst]{natbib}
\usepackage[unicode=true,
 bookmarks=true,bookmarksnumbered=false,bookmarksopen=false,
 breaklinks=true,pdfborder={0 0 0},pdfborderstyle={},backref=false,colorlinks=false]
 {hyperref}

\makeatletter

\AtBeginDocument{\providecommand\subsecref[1]{\ref{subsec:#1}}}
\AtBeginDocument{\providecommand\secref[1]{\ref{sec:#1}}}
\AtBeginDocument{\providecommand\eqref[1]{\ref{eq:#1}}}
\AtBeginDocument{\providecommand\assref[1]{\ref{ass:#1}}}
\AtBeginDocument{\providecommand\appref[1]{\ref{app:#1}}}
\AtBeginDocument{\providecommand\lemref[1]{\ref{lem:#1}}}
\AtBeginDocument{\providecommand\figref[1]{\ref{fig:#1}}}
\AtBeginDocument{\providecommand\enuref[1]{\ref{enu:#1}}}
\AtBeginDocument{\providecommand\exaref[1]{\ref{exa:#1}}}
\AtBeginDocument{\providecommand\thmref[1]{\ref{thm:#1}}}
\AtBeginDocument{\providecommand\propref[1]{\ref{prop:#1}}}
\providecommand{\tabularnewline}{\\}
\RS@ifundefined{subsecref}
  {\newref{subsec}{name = \RSsectxt}}
  {}
\RS@ifundefined{thmref}
  {\def\RSthmtxt{theorem~}\newref{thm}{name = \RSthmtxt}}
  {}
\RS@ifundefined{lemref}
  {\def\RSlemtxt{lemma~}\newref{lem}{name = \RSlemtxt}}
  {}

\numberwithin{equation}{section}
\theoremstyle{remark}
\newtheorem*{notation*}{\protect\notationname}
\theoremstyle{plain}
\newtheorem{assumption}{\protect\assumptionname}
\theoremstyle{plain}
\newtheorem{prop}{\protect\propositionname}[section]
\theoremstyle{definition}
\newtheorem{example}{\protect\examplename}[section]
\theoremstyle{plain}
\newtheorem{thm}{\protect\theoremname}[section]
\theoremstyle{plain}
\newtheorem{lem}{\protect\lemmaname}[section]

\setlist[enumerate,1]{label=\upshape{(\roman*)}, ref=(\roman*)}
\setlist[enumerate,2]{label=\upshape{(\alph*)}, ref=(\alph*)}
\setlist[enumerate,3]{label=\upshape{\roman*.}, ref=\roman*}

\makeatletter
  \@ifpackageloaded{refstyle}{}{\usepackage{refstyle}}
\makeatother

\newref{thm}{name = Theorem~}
\newref{prop}{name = Proposition~}
\newref{lem}{name = Lemma~}
\newref{cor}{name = Corollary~}
\newref{ass}{name = }
\newref{exa}{name = Example~}
\newref{rem}{name = Remark~}
\newref{def}{name = Definition~}

\newref{eq}{refcmd = (\ref{#1})}

\newref{sec}{name = Section~}
\newref{sub}{name = Section~}
\newref{subsec}{name = Section~}
\newref{app}{name = Appendix~}

\newref{fig}{name = Figure~}
\newref{tbl}{name = Table~}

\makeatletter
  \@ifpackageloaded{mathrsfs}{}{\usepackage{mathrsfs}}
\makeatother

\date{}

\usepackage{nextpage}

\allowdisplaybreaks[1]

\usepackage{tikz}
\usepackage{scalefnt}
\usepackage[group-separator={,}]{siunitx}
\newcommand\smaller[2][0.85]{{\scalefont{#1}#2}}
\newcommand{\ass}[1]{{\upshape{\smaller[0.76]{#1}}}}

\usepackage{alphalph}

\usepackage{changepage}
\usepackage{threeparttable}

\newcommand{\proofpart}[1]{\textbf{({\romannumeral #1}).}}

\usepackage[usestackEOL]{stackengine}
\setstackgap{S}{5pt}
\newcommand{\authaffil}[2]{\Shortunderstack{#1\\\small{#2}}}

\theoremstyle{definition}
\newtheorem{examplex}{Example}
\renewenvironment{example}
  {\pushQED{\qed}\examplex}
  {\popQED\endexamplex}

\numberwithin{examplex}{section}

\usepackage[justification=centering]{caption}

\makeatother

\addto\captionsbritish{\renewcommand{\assumptionname}{Assumption}}
\addto\captionsbritish{\renewcommand{\examplename}{Example}}
\addto\captionsbritish{\renewcommand{\lemmaname}{Lemma}}
\addto\captionsbritish{\renewcommand{\notationname}{Notation}}
\addto\captionsbritish{\renewcommand{\propositionname}{Proposition}}
\addto\captionsbritish{\renewcommand{\theoremname}{Theorem}}
\addto\captionsenglish{\renewcommand{\assumptionname}{Assumption}}
\addto\captionsenglish{\renewcommand{\examplename}{Example}}
\addto\captionsenglish{\renewcommand{\lemmaname}{Lemma}}
\addto\captionsenglish{\renewcommand{\notationname}{Notation}}
\addto\captionsenglish{\renewcommand{\propositionname}{Proposition}}
\addto\captionsenglish{\renewcommand{\theoremname}{Theorem}}
\providecommand{\assumptionname}{Assumption}
\providecommand{\examplename}{Example}
\providecommand{\lemmaname}{Lemma}
\providecommand{\notationname}{Notation}
\providecommand{\propositionname}{Proposition}
\providecommand{\theoremname}{Theorem}

\begin{document}
\selectlanguage{english}%


\global\long\def\uwrite#1#2{\underset{#2}{\underbrace{#1}} }%

\global\long\def\blw#1{\ensuremath{\underline{#1}}}%

\global\long\def\abv#1{\ensuremath{\overline{#1}}}%

\global\long\def\vect#1{\mathbf{#1}}%


\global\long\def\smlseq#1{\{#1\} }%

\global\long\def\seq#1{\left\{  #1\right\}  }%

\global\long\def\smlsetof#1#2{\{#1\mid#2\} }%

\global\long\def\setof#1#2{\left\{  #1\mid#2\right\}  }%


\global\long\def\goesto{\ensuremath{\rightarrow}}%

\global\long\def\ngoesto{\ensuremath{\nrightarrow}}%

\global\long\def\uto{\ensuremath{\uparrow}}%

\global\long\def\dto{\ensuremath{\downarrow}}%

\global\long\def\uuto{\ensuremath{\upuparrows}}%

\global\long\def\ddto{\ensuremath{\downdownarrows}}%

\global\long\def\ulrto{\ensuremath{\nearrow}}%

\global\long\def\dlrto{\ensuremath{\searrow}}%


\global\long\def\setmap{\ensuremath{\rightarrow}}%

\global\long\def\elmap{\ensuremath{\mapsto}}%

\global\long\def\compose{\ensuremath{\circ}}%

\global\long\def\cont{C}%

\global\long\def\cadlag{D}%

\global\long\def\Ellp#1{\ensuremath{\mathcal{L}^{#1}}}%


\global\long\def\naturals{\ensuremath{\mathbb{N}}}%

\global\long\def\reals{\mathbb{R}}%

\global\long\def\complex{\mathbb{C}}%

\global\long\def\rationals{\mathbb{Q}}%

\global\long\def\integers{\mathbb{Z}}%


\global\long\def\abs#1{\ensuremath{\left|#1\right|}}%

\global\long\def\smlabs#1{\ensuremath{\lvert#1\rvert}}%
 
\global\long\def\bigabs#1{\ensuremath{\bigl|#1\bigr|}}%
 
\global\long\def\Bigabs#1{\ensuremath{\Bigl|#1\Bigr|}}%
 
\global\long\def\biggabs#1{\ensuremath{\biggl|#1\biggr|}}%

\global\long\def\norm#1{\ensuremath{\left\Vert #1\right\Vert }}%

\global\long\def\smlnorm#1{\ensuremath{\lVert#1\rVert}}%
 
\global\long\def\bignorm#1{\ensuremath{\bigl\|#1\bigr\|}}%
 
\global\long\def\Bignorm#1{\ensuremath{\Bigl\|#1\Bigr\|}}%
 
\global\long\def\biggnorm#1{\ensuremath{\biggl\|#1\biggr\|}}%


\global\long\def\Union{\ensuremath{\bigcup}}%

\global\long\def\Intsect{\ensuremath{\bigcap}}%

\global\long\def\union{\ensuremath{\cup}}%

\global\long\def\intsect{\ensuremath{\cap}}%

\global\long\def\pset{\ensuremath{\mathcal{P}}}%

\global\long\def\clsr#1{\ensuremath{\overline{#1}}}%

\global\long\def\symd{\ensuremath{\Delta}}%

\global\long\def\intr{\operatorname{int}}%

\global\long\def\cprod{\otimes}%

\global\long\def\Cprod{\bigotimes}%


\global\long\def\smlinprd#1#2{\ensuremath{\langle#1,#2\rangle}}%

\global\long\def\inprd#1#2{\ensuremath{\left\langle #1,#2\right\rangle }}%

\global\long\def\orthog{\ensuremath{\perp}}%

\global\long\def\dirsum{\ensuremath{\oplus}}%


\global\long\def\spn{\operatorname{sp}}%

\global\long\def\rank{\operatorname{rk}}%

\global\long\def\proj{\operatorname{proj}}%

\global\long\def\tr{\operatorname{tr}}%


\global\long\def\smpl{\ensuremath{\Omega}}%

\global\long\def\elsmp{\ensuremath{\omega}}%

\global\long\def\sigf#1{\mathcal{#1}}%

\global\long\def\sigfield{\ensuremath{\mathcal{F}}}%
\global\long\def\sigfieldg{\ensuremath{\mathcal{G}}}%

\global\long\def\flt#1{\mathcal{#1}}%

\global\long\def\filt{\mathcal{F}}%
\global\long\def\filtg{\mathcal{G}}%

\global\long\def\Borel{\ensuremath{\mathcal{B}}}%

\global\long\def\cyl{\ensuremath{\mathcal{C}}}%

\global\long\def\nulls{\ensuremath{\mathcal{N}}}%

\global\long\def\gauss{\mathfrak{g}}%

\global\long\def\leb{\mathfrak{m}}%


\global\long\def\prob{P}%

\global\long\def\Prob{\ensuremath{\mathbb{P}}}%

\global\long\def\Probs{\mathcal{P}}%

\global\long\def\PROBS{\mathcal{M}}%

\global\long\def\expect{\ensuremath{\mathbb{E}}}%

\global\long\def\probspc{\ensuremath{(\smpl,\filt,\Prob)}}%


\global\long\def\iid{\ensuremath{\textnormal{i.i.d.}}}%

\global\long\def\as{\ensuremath{\textnormal{a.s.}}}%

\global\long\def\asp{\ensuremath{\textnormal{a.s.p.}}}%

\global\long\def\io{\ensuremath{\ensuremath{\textnormal{i.o.}}}}%

\newcommand\independent{\protect\mathpalette{\protect\independenT}{\perp}}
\def\independenT#1#2{\mathrel{\rlap{$#1#2$}\mkern2mu{#1#2}}}

\global\long\def\indep{\independent}%

\global\long\def\distrib{\ensuremath{\sim}}%

\global\long\def\distiid{\ensuremath{\sim_{\iid}}}%

\global\long\def\asydist{\ensuremath{\overset{a}{\distrib}}}%

\global\long\def\inprob{\ensuremath{\overset{p}{\goesto}}}%

\global\long\def\inprobu#1{\ensuremath{\overset{#1}{\goesto}}}%

\global\long\def\inas{\ensuremath{\overset{\as}{\goesto}}}%

\global\long\def\eqas{=_{\as}}%

\global\long\def\inLp#1{\ensuremath{\overset{\Ellp{#1}}{\goesto}}}%

\global\long\def\indist{\ensuremath{\overset{d}{\goesto}}}%

\global\long\def\eqdist{=_{d}}%

\global\long\def\wkc{\ensuremath{\rightsquigarrow}}%

\global\long\def\wkcu#1{\overset{#1}{\ensuremath{\rightsquigarrow}}}%

\global\long\def\plim{\operatorname*{plim}}%


\global\long\def\var{\operatorname{var}}%

\global\long\def\lrvar{\operatorname{lrvar}}%

\global\long\def\cov{\operatorname{cov}}%

\global\long\def\corr{\operatorname{corr}}%

\global\long\def\bias{\operatorname{bias}}%

\global\long\def\MSE{\operatorname{MSE}}%

\global\long\def\med{\operatorname{med}}%


\global\long\def\simple{\mathcal{R}}%

\global\long\def\sring{\mathcal{A}}%

\global\long\def\sproc{\mathcal{H}}%

\global\long\def\Wiener{\ensuremath{\mathbb{W}}}%

\global\long\def\sint{\bullet}%

\global\long\def\cv#1{\left\langle #1\right\rangle }%

\global\long\def\smlcv#1{\langle#1\rangle}%

\global\long\def\qv#1{\left[#1\right]}%

\global\long\def\smlqv#1{[#1]}%


\global\long\def\trans{\ensuremath{\prime}}%

\global\long\def\indic{\ensuremath{\mathbf{1}}}%

\global\long\def\Lagr{\mathcal{L}}%

\global\long\def\grad{\nabla}%

\global\long\def\pmin{\ensuremath{\wedge}}%
\global\long\def\Pmin{\ensuremath{\bigwedge}}%

\global\long\def\pmax{\ensuremath{\vee}}%
\global\long\def\Pmax{\ensuremath{\bigvee}}%

\global\long\def\sgn{\operatorname{sgn}}%

\global\long\def\argmin{\operatorname*{argmin}}%

\global\long\def\argmax{\operatorname*{argmax}}%

\global\long\def\Rp{\operatorname{Re}}%

\global\long\def\Ip{\operatorname{Im}}%

\global\long\def\deriv{\ensuremath{\mathrm{d}}}%

\global\long\def\diffnspc{\ensuremath{\deriv}}%

\global\long\def\diff{\ensuremath{\,\deriv}}%

\global\long\def\i{\ensuremath{\mathrm{i}}}%

\global\long\def\e{\mathrm{e}}%

\global\long\def\sep{,\ }%

\global\long\def\defeq{\coloneqq}%

\global\long\def\eqdef{\eqqcolon}%

\selectlanguage{british}%

\selectlanguage{english}%
\global\long\def\eigs{\mathcal{L}}%

\global\long\def\spect{\mathcal{\lambda}}%

\global\long\def\trans{\mathsf{T}}%

\global\long\def\qcs{\textnormal{QCS}}%

\global\long\def\cs{\textnormal{CS}}%

\global\long\def\lu{{\scriptscriptstyle \textnormal{LU}}}%

\global\long\def\ur{{\scriptscriptstyle \textnormal{UR}}}%

\global\long\def\st{{\scriptscriptstyle \textnormal{ST}}}%

\global\long\def\sl{{\scriptscriptstyle \textnormal{SL}}}%

\global\long\def\ls{{\scriptscriptstyle \textnormal{LS}}}%

\global\long\def\lds{{\scriptscriptstyle \textnormal{L\ensuremath{\bullet}S}}}%

\global\long\def\sdl{{\scriptscriptstyle \textnormal{S\ensuremath{\bullet}L}}}%

\global\long\def\vek{\operatorname{vec}}%

\global\long\def\diag{\operatorname{diag}}%

\global\long\def\col{\operatorname{col}}%

\global\long\def\chol{\text{ch}}%

\global\long\def\proj{\mathcal{Q}}%

\global\long\def\smlfloor#1{\lfloor#1\rfloor}%

\global\long\def\sm{\mathfrak{sm}}%

\global\long\def\mg{\mathfrak{mg}}%

\global\long\def\M{\mathcal{M}}%

\global\long\def\N{\mathcal{N}}%

\global\long\def\wald{\mathcal{W}}%

\global\long\def\PHI{\boldsymbol{\Phi}}%

\global\long\def\like{\mathcal{\ell}}%

\global\long\def\likens{\like^{\ast}}%

\global\long\def\mn{\mathrm{MN}}%

\global\long\def\normdist{\mathrm{N}}%

\global\long\def\lr{\mathcal{LR}}%

\global\long\def\spc#1{\mathcal{#1}}%

\global\long\def\set#1{\mathscr{#1}}%

\global\long\def\err{\varepsilon}%

\global\long\def\serr{w}%

\global\long\def\radius{\rho}%

\global\long\def\xdet{\bar{x}}%

\global\long\def\ydet{\bar{y}}%

\global\long\def\zdet{\bar{z}}%

\global\long\def\errdet{\bar{\err}}%

\global\long\def\Zdet{\bar{Z}}%

\global\long\def\Ld{\set L_{\mathrm{d}}}%

\global\long\def\Ln{\set L_{\mathrm{n}}}%

\global\long\def\Ls{\set L_{\mathrm{s}}}%

\global\long\def\ci{\mathcal{C}}%

\global\long\def\cisimple{\ci_{\textnormal{S}}}%

\global\long\def\cibonf{\ci_{\textnormal{B}}}%

\global\long\def\cijoh{\ci_{\textnormal{J}}}%

\global\long\def\cinp{\ci_{\textnormal{NP}}}%

\global\long\def\largedec#1{\mathbf{#1}}%

\global\long\def\R{\largedec R}%

\global\long\def\G{\largedec G}%

\global\long\def\L{\largedec L}%

\global\long\def\y{\largedec y}%

\global\long\def\x{\largedec x}%

\global\long\def\BM{\mathrm{BM}}%

\global\long\def\ginv{\#}%

\global\long\def\irf{\mathrm{IRF}}%

\global\long\def\f{\varphi}%

\global\long\def\spcf{\mathcal{P}}%

\global\long\def\cone{\mathcal{K}}%

\global\long\def\snum{{\scriptscriptstyle \#}}%

\global\long\def\locest{f}%

\global\long\def\dist{\operatorname{dist}}%

\global\long\def\THETA{\boldsymbol{\Theta}}%

\global\long\def\GAMMA{\boldsymbol{\Gamma}}%

\global\long\def\PI{\boldsymbol{\Pi}}%

\global\long\def\Z{\mathcal{Z}}%

\global\long\def\dElta{\boldsymbol{\delta}}%

\global\long\def\vLambda{\boldsymbol{\lambda}}%

\global\long\def\vA{\boldsymbol{a}}%

\global\long\def\vall{\boldsymbol{\pi}}%

\global\long\def\vAll{\boldsymbol{\Pi}}%

\global\long\def\pr{{\scriptscriptstyle \textnormal{PR}}}%

\global\long\def\crit{\mathrm{cv}}%

\global\long\def\test{\text{\ensuremath{\phi}}}%

\global\long\def\np{\mathcal{NP}}%

\global\long\def\ret{\varrho}%

\global\long\def\lfst{\mathcal{ST}}%

\global\long\def\Err{\boldsymbol{\varepsilon}}%
\selectlanguage{british}%

\title{Cointegration without Unit Roots}
\author{\authaffil{James A.\ Duffy\footnotemark[1]{}}{University of Oxford}\hspace{2cm}
\authaffil{Jerome R.\ Simons\footnotemark[2]{}\vphantom{y}}{University
of Cambridge}}
\date{\vspace*{0.3cm}April 2023}

\maketitle
\renewcommand*{\thefootnote}{\fnsymbol{footnote}}

\footnotetext[1]{Department of Economics and Corpus Christi College;
\texttt{james.duffy@economics.ox.ac.uk}}

\footnotetext[2]{Faculty of Economics and St Edmund's College; \texttt{jrs89@cam.ac.uk}}

\renewcommand*{\thefootnote}{\arabic{footnote}}

\setcounter{footnote}{0}
\begin{abstract}
\noindent It has been known since Elliott (1998) that standard methods
of inference on cointegrating relationships break down entirely when
autoregressive roots are near but not exactly equal to unity. We consider
this problem within the framework of a structural VAR, arguing this
it is as much a problem of identification failure as it is of inference.
We develop a characterisation of cointegration based on the impulse
response function, which allows long-run equilibrium relationships
to remain identified even in the absence of exact unit roots. Our
approach also provides a framework in which the structural shocks
driving the common persistent components continue to be identified
via long-run restrictions, just as in an SVAR with exact unit roots.
We show that inference on the cointegrating relationships is affected
by nuisance parameters, in a manner familiar from predictive regression;
indeed the two problems are asymptotically equivalent. By adapting
the approach of Elliott, Müller and Watson (2015) to our setting,
we develop tests that robustly control size while sacrificing little
power (relative to tests that are efficient in the presence of exact
unit roots).
\end{abstract}
\vfill

\noindent The authors thank G.~Bårdsen, V.\ Berenguer-Rico, S.\ Bond,
P.\ Boswijk, V.\ Carvalho, G.\ Chevillon, A.\ Harvey, O.\ Linton,
S.\ Mavroeidis, U.\ Müller, B.\ Nielsen, A.\ Onatskiy and participants
at seminars at Amsterdam, NTNU (Trondheim), ESSEC (Cergy), Essex,
Southampton, Sydney, UNSW and Oxford for helpful comments on earlier
drafts of this work.

\thispagestyle{plain}

\pagenumbering{roman}

\newpage{}

\thispagestyle{plain}

\setcounter{tocdepth}{2}

\tableofcontents{}

\newpage{}

\pagenumbering{arabic}

\section{Introduction}

Since its development in the late 1980s, the cointegrated vector autoregressive
model has been widely applied to the modelling of macroeconomic time
series -- a testament to its ability to account for both the short-
and long-run dynamics of these series in a unified way (\citealp{hendry1986econometric};
\citealp{EG87Ecta}; \citealp{Joh95}). By allowing for one or more
autoregressive roots at unity, the model is able to match two key
features of these series: their high degree of persistence, which
gives rise to their characteristically `random wandering' behaviour,
and the tendency for economically related series to move together,
such that certain linear combinations of these series -- given by
the cointegrating relationships -- are markedly less persistent than
the series themselves. Dual to this, the model provides a framework
for identifying the structural shocks whose permanent effects generate
these patterns of co-movement (\citealp{BlanchardQuah89}; \citealp{KPSW91AER}),
which has been widely used in empirical studies.

Cointegrating relationships are often of intrinsic interest because
macroeconomic theories make definite predictions about the existence
and magnitude of the long-run equilibrium relationships that these
embody. (See e.g.\ \citealp{FM95QJE}; \citealp{MMS04AER}; \citealp{SK07AER}.)
For the purposes of estimating these relationships, and thus of testing
the predictions of such theories, a variety of efficient methods exist,
such as FM-OLS (\citealp{PH90RES}), DOLS (\citealp{SW93Ecta}), and
(rank-imposed) maximum likelihood estimation of the VAR itself (\citealp{Joh95}).
However, all these methods rely on a common assumption that the data
is generated by a VAR with a certain number of exact unit roots. \citet{Ell98Ecta}
showed that should this assumption fail only slightly -- such that
some roots are merely `close' but not exactly equal to unity --
then inferences based on these methods can suffer from severe size
distortions. His findings are particularly disturbing because this
problem arises even in a VAR with roots that are `nearly' unity,
in the sense of lying within an $O(n^{-1})$ neighbourhood thereof,
which is practically indistinguishable from the same model with exact
unit roots.

The present work addresses the problem posed by \citet{Ell98Ecta}:
how can one perform valid inference on the cointegrating relationships
in an (S)VAR, when the dominant roots in that model may not be exactly
unity? In view of the significance of Elliott's findings, it is perhaps
surprising that only a few previous contributions have also attempted
to respond to them: most notably \citet{Wri00JBES}, \citet{MP09},
\citet{MW13JoE}, \citet{FJ17Ect}, and \citet{HV23JBES}. The approach
taken in this paper is quite different from that taken in those previous
works, which have largely followed \citet{Ell98Ecta} in framing the
problem as an inferential one, which might be solved merely by using
appropriately modified estimators and tests. Instead, our view is
that the problem is at least as much one of identification failure
as it is of inference. Indeed, the usual definition of cointegration
-- in terms of linear combinations of series that eliminate their
common integrated components -- becomes meaningless as soon as the
largest characteristic root in a VAR departs even slightly from unity.
(See \subsecref{literature} below for a further discussion of how
our contribution relates to those previous works.)

Our first task is thus to develop a characterisation of cointegration,
based on the impulse response function implied by the VAR, that remains
meaningful in a model with some roots near but not necessarily equal
to unity. In a $p$-dimensional VAR with $q$ roots near unity, one
can always identify a $p-q=r$-dimensional subspace $S_{r}$, such
that the decay of the impulse response function in the directions
contained in $S_{r}$ is more rapid than it is in all other directions
(Sections\ \ref{subsec:model}--\ref{subsec:qcs}). We term $S_{r}$
the \emph{quasi-cointegrating space} (QCS) of the VAR. When those
$q$ roots are exactly unity, the QCS coincides exactly with the cointegrating
space -- and when they are modelled as being local to unity, i.e.\ lying
within a $O(n^{-1})$ neighbourhood of unity, the quasi-cointegrating
relations are exactly those that eliminate the common near stochastic
trends from the system.

While quasi-cointegration is not the only conceivable way of extending
cointegration to a wider domain, our approach has the further advantage
of maintaining the duality, that exists in an SVAR with exact unit
roots, between the identification of the long-run equilibrium relationships
between the series, and of the subvector of structural shocks whose
common permanent effects underpin those relationships. In this way,
we simultaneously extend both cointegration, and the use of long-run
identifying restrictions, to an SVAR without exact unit roots, by
allowing that a subset of the structural shocks may have effects that
are highly persistent, rather than permanent, where persistence is
understood in terms of the (relative) decay rate of the impact of
those shocks (\subsecref{lr-restrictions}). We thereby show how these
long-run restrictions, which are often thought to be available only
in the case of exact unit roots (as noted e.g.\ in \citealp{KL17book},
Sec.~10.5.1), remain a viable approach to identification even without
this auxiliary assumption. 

Inference on the QCS is complicated by the presence of nuisance parameters,
which measure the proximity of the dominant roots of the VAR to unity
(\secref{estimation}). This problem is similar to that which arises
in predictive regressions, when the regressors have an unknown but
possibly high degree of persistence, such as has been studied e.g.\ by
\citet{CES95ET}, \citet{CY06JFE}, \citet{JM06Ecta}, \citet{PL13JoE},
\citet{Phi14Ecta}, and \citet{KMS15RFS}. In fact, we show that the
problem of inference on the QCS is asymptotically equivalent to that
of inference in a predictive regression: both converge to a common
limiting experiment under an appropriate local parametrisation. This
equivalence permits methods that have been developed for predictive
regression -- for which there are a great many -- to be transposed
the present setting. Our problem also fits within the general framework
of \citet{EMW15Ecta}, and by adapting their approach to the present
setting, we obtain tests and confidence intervals that are effectively
free of any size distortions, while sacrificing little power relative
to the efficient estimators, even when the data is generated with
exact unit roots. We also extend the mixed normality of the maximum-likelihood
estimator, when the correct number of unit roots are imposed, to the
case where the correct values of the dominant roots are imposed. This
provides efficient likelihood-based tests and confidence intervals
for cases where one is willing to take a stand on the values of these
parameters.

The finite-sample performance of our procedure is evaluated through
a series of simulation exercises, and illustrated with an empirical
application to the expectations theory of the term structure (\secref{finite-sample}).
Proofs of all technical results appear in the appendices.
\begin{notation*}
All limits are taken as $n\goesto\infty$ unless otherwise stated.
$\inprob$ and $\wkc$ respectively denote convergence in probability
and in distribution (weak convergence). We write `$X_{n}(\lambda)\wkc X(\lambda)$
on $D[0,1]$' to denote that $\{X_{n}\}$ converges weakly to $X$,
where these are considered as random elements of $D[0,1]$, the space
of cadlag functions $[0,1]\setmap\reals^{m}$, equipped with the uniform
topology. $\smlnorm{\cdot}$ denotes the Euclidean norm on $\reals^{m}$;
all matrix norms are induced by the corresponding vector norms. For
$X$ a random variable and $p\geq1$, $\smlnorm X_{p}\defeq(\expect\smlabs X^{p})^{1/p}$.
$M^{1/2}$ denotes the principal square root of a positive semidefinite
matrix $M$.
\end{notation*}

\section{`Cointegration' in a VAR without unit roots}

\label{sec:cointegration}

\subsection{Model and assumptions}

\label{subsec:model}

The data generating process (DGP) for the observed series $\{y_{t}\}_{t=1}^{n}$
is a $k$th order vector autoregressive (VAR) model, written in unobserved
components form as
\begin{align}
y_{t} & =\mu+\delta t+x_{t} & x_{t} & =\sum_{i=1}^{k}\Phi_{i}x_{t-i}+\err_{t}\label{eq:dgp}
\end{align}
where $\err_{t}$, $x_{t}$ and $y_{t}$ are $p$-dimensional random
vectors. The reduced-form shocks $\{\err_{t}\}$ depend on an underlying
($p$-dimensional) vector of i.i.d.\ and mutually uncorrelated structural
shocks $\{\serr_{t}\}$, via
\begin{equation}
\err_{t}=\Upsilon\serr_{t},\label{eq:dgp-err}
\end{equation}
so that \eqref{dgp}--\eqref{dgp-err} comprise a structural VAR
(or SVAR). Let $\Phi(\lambda)\defeq I\lambda^{k}-\sum_{i=1}^{k}\Phi_{i}\lambda^{k-i}$
denote the characteristic polynomial associated to \eqref{dgp}; we
shall refer to any $\lambda$ for which $\det\Phi(\lambda)=0$ as
a `root of $\Phi$'. Let $\PHI\defeq(\Phi_{1},\Phi_{2},\ldots,\Phi_{k})\in\reals^{p\times kp}$.
We generally maintain the following.

\setcounter{assumption}{2901}
\begin{assumption}
\label{ass:DGP}$\{y_{t}\}_{t=1}^{n}$ and $\{x_{t}\}_{t=1}^{n}$
are generated under \eqref{dgp}--\eqref{dgp-err}, where:
\begin{enumerate}[label=\textnormal{\ass{DGP\arabic*}},leftmargin=1.5cm]
\item $\det\Phi(\lambda)\neq0$ for all $\smlabs{\lambda}>1$; 
\item $\{\serr_{t}\}$ is i.i.d.\ with $\expect\serr_{t}=0$ and $\expect\serr_{t}\serr_{t}^{\trans}=I_{p}$,
and $\Sigma\defeq\Upsilon\Upsilon^{\trans}$ is positive definite;
\item \label{enu:DGP:init}$x_{0}=x_{-1}=\cdots=x_{-k+1}=0$.
\end{enumerate}
\end{assumption}
We say that a $d_{z}$-dimensional process $\{z_{t}\}$ is integrated
of order zero, denoted $z_{t}\sim I(0)$, if there exists a deterministic
process $\{\mu_{t}\}$ such that $n^{-1/2}\sum_{s=1}^{\smlfloor{nr}}(z_{s}-\mu_{s})\wkc B(r)$,
for $B$ a $d_{z}$-dimensional Brownian motion. Letting $\Delta^{d}$
denote the $d$th order temporal differencing operator, we say that
$z_{t}$ is \emph{integrated of order $d$}, denoted $z_{t}\sim I(d)$,
if $\Delta^{d}z_{t}\sim I(0)$. We say $\{z_{t}\}$ is \emph{nearly
integrated} if $n^{-1/2}(z_{\smlfloor{nr}}-\mu_{\smlfloor{nr}})\wkc\int_{0}^{r}\e^{C(r-s)}\deriv B(s)$
for some $C\in\reals^{d_{z}\times d_{z}}$.

\subsection{Cointegration: the model with unit roots}

\label{subsec:coint-ur}

Cointegration analysis is concerned with how linear combinations of
$I(d)$ processes can yield processes that are themselves only $I(d-b)$
for some $0<b\leq d$. The reduced persistence and more rapid mean
reversion of the latter is interpreted as evidence of a long-run equilibrium
relationship between the original processes. Here, we focus exclusively
on the special but practically important case of $I(1)$ processes
having linear combinations that are $I(0)$, reserving the term \emph{cointegration}
exclusively for this case. As is well-known, the VAR model \eqref{dgp}
is able to generate cointegrated $I(1)$ processes under the following
assumptions, which define the $I(0)/I(1)$ cointegrated VAR (CVAR)
model.

\setcounter{assumption}{99}
\begin{assumption}
\label{ass:J}~
\begin{enumerate}[label=\textnormal{\ass{CV\arabic*}},leftmargin=1.5cm]
\item \label{enu:J:roots}$\Phi$ has $q$ roots at (real) unity, and all
others strictly inside the unit circle.
\item \label{enu:J:rank}$\rank\Phi(1)=p-q\eqdef r$
\end{enumerate}
\end{assumption}
By the Granger--Johansen representation theorem (GJRT; see e.g.\ \citealt[Thm~4.2 and Cor.~4.3]{Joh95}),
the preceding is necessary and sufficient for $y_{t}\sim I(1)$, and
for there to exist a rank $r$ matrix $\beta\in\reals^{p\times r}$
of cointegrating relationships, such that $\beta^{\trans}y_{t}\sim I(0)$.
The matrix $\beta$ is identified only up to its column space, $\cs\defeq\spn\beta$,
termed the \emph{cointegrating space} (CS). Two equivalent characterisations
of the cointegrating space, the first of which is definitional and
the second of which follows immediately from the GJRT, are
\begin{enumerate}[label=(C.\roman*),leftmargin=1.5cm]
\item \label{enu:coint:i0}$b^{\trans}y_{t}\sim I(0)$ if and only if $b\in\cs$;
and
\item \label{enu:coint:Phi}$\cs=\spn\Phi(1)^{\trans}=\{\ker\Phi(1)\}^{\perp}$.
\end{enumerate}

Our objective in this paper is to estimate the CS, or a space sharing
its key properties, in a setting more general than that of \assref{J}.
For this purpose, we next recall two further characterisations of
the CS that extend beyond the setting of \assref{J}, in a way that
the preceding do not.\footnote{While the arguments that lead to these characterisations may be familiar
to the reader, for completeness formal statements and proofs of the
results underlying the discussion that follows appear in \appref{qcs-results}.} The third characterisation is in terms of the impulse response function
of $\{y_{t}\}$ with respect to the reduced-form or structural disturbances
(i.e.\ $\{\err_{s}\}$ or $\{\serr_{s}\}$), denoted
\begin{align*}
\irf_{s}^{\err} & \defeq\frac{\partial y_{t+s}}{\partial\err_{t}}=\frac{\partial x_{t+s}}{\partial\err_{t}} & \irf_{s}^{\serr} & \defeq\frac{\partial y_{t+s}}{\partial\serr_{t}}=\irf_{s}^{\err}\frac{\partial\serr_{t}}{\partial\err_{t}}=\irf_{s}^{\err}\Upsilon,
\end{align*}
For a given $b\in\reals^{p}$, the product $b^{\trans}\irf_{s}^{\serr}$
gives the response of the linear combination $b^{\trans}y_{t+s}$
to $\serr_{s}$. The rate at which $b^{\trans}\irf_{s}^{\serr}$ (or
$b^{\trans}\irf_{s}^{\err}$) decays as the horizon $s$ diverges
provides a measure of the persistence of the series $\{b^{\trans}y_{t}\}$.
Now let $m<p$, and define $S_{m}\subset\reals^{p}$ to be an $m$-dimensional
linear subspace such that for every $b\in S_{m}$ and $c\notin S_{m}$,
\begin{equation}
\lim_{s\goesto\infty}\frac{\smlnorm{b^{\trans}\irf_{s}^{\serr}}}{\smlnorm{c^{\trans}\irf_{s}^{\serr}}}=0.\label{eq:relative-irf}
\end{equation}
When it exists, $S_{m}$ collects those $m$ linear combinations of
$y_{t}$ that are, in the sense of \eqref{relative-irf}, the least
persistent. While $\irf_{s}^{\serr}$ evidently depends on $\Upsilon$,
the subspace $S_{m}$ itself is \emph{invariant} to $\Upsilon$, and
hence to the scheme used to identify the structural shocks; indeed,
we may equivalently characterise $S_{m}$ in terms of the reduced-form
impulse responses, with $\irf_{s}^{\err}$ taking the place of $\irf_{s}^{\serr}$
in \eqref{relative-irf}. Under \assref{J}, $S_{m}$ with $m=r$
exists and is unique, and moreover
\begin{enumerate}[resume, resume*]
\item \label{enu:coint:irf}$\cs=S_{r}$
\end{enumerate}
(see \lemref{qcs}). In other words, the cointegrating space is spanned
by the vectors giving the $r$ least persistent linear combinations
of $y_{t}$.

\begin{figure}
\begin{adjustwidth}{-2cm}{-2cm}
\noindent \begin{centering}
\begin{tabular}{cc}
\includegraphics[viewport=100bp 250bp 458bp 539bp,clip,scale=0.7]{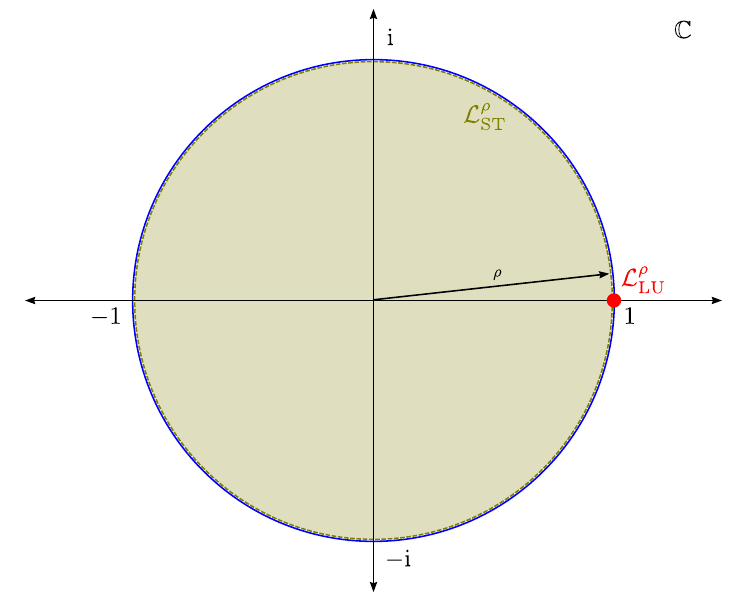} & \includegraphics[viewport=100bp 250bp 458bp 539bp,clip,scale=0.7]{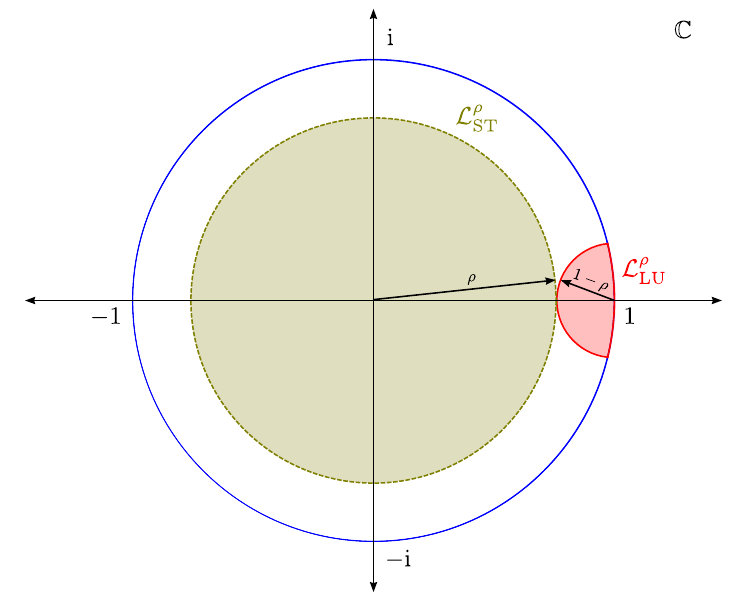}\tabularnewline
(a) Cointegration ($\radius=1$) & (b) Quasi-cointegration ($\radius<1$)\tabularnewline
\end{tabular}
\par\end{centering}
\end{adjustwidth}

\caption{$\protect\eigs_{\protect\lu}^{\rho}$ and $\protect\eigs_{\protect\st}^{\protect\radius}$,
from \eqref{spect_sets}, in the complex plane}

\label{fig:roots}
\end{figure}

Our final characterisation of the cointegrating space provides the
basis for its estimation in settings more general than \assref{J};
it derives from the application of a spectral decomposition to the
companion form representation of \eqref{dgp} (see \lemref{GLR}).
Define the disjoint sets
\begin{align}
\eigs_{\lu}^{\radius} & \defeq\{z\in\complex\mid\smlabs z\leq1\text{ and }\smlabs{z-1}\leq1-\radius\} & \eigs_{\st}^{\radius} & \defeq\{z\in\complex\mid\smlabs z<\radius\},\label{eq:spect_sets}
\end{align}
so that for a given $\radius\leq1$ (but close to unity), $\eigs_{\lu}^{\radius}$
defines a closed set of points on or inside the unit complex circle,
within a distance $1-\rho$ of real unity, and $\eigs_{\st}^{\radius}$
defines an open ball of radius $\radius$ centred at zero, as depicted
in \figref{roots} separately for the cases where $\rho<1$ and $\rho=1$.
Now suppose that $\Phi$ has $q$ roots in $\eigs_{\lu}^{\radius}$
and all others in $\eigs_{\st}^{\radius}$ for some $\radius\leq1$.
Under \enuref{J:roots} this setup holds with $\radius=1$, so that
these sets are as in \figref{roots}(a). Since $\eigs_{\lu}^{\radius}$
and $\eigs_{\st}^{\radius}$ are disjoint, there exist real matrices
\begin{align}
\underset{(p\times kp)}{R\vphantom{R_{\lu}}} & \defeq[\begin{array}[t]{cc}
\underset{(p\times q)}{R_{\lu}\vphantom{\sum_{1}}} & \underset{\phantom{p\times q}}{R_{\st}}\end{array}] & \underset{(p\times kp)}{L\vphantom{R_{\lu}}} & \defeq[\begin{array}[t]{cc}
\underset{(p\times q)}{L_{\lu}\vphantom{\sum_{1}}} & \underset{\phantom{p\times q}}{L_{\st}}\end{array}] & \underset{(kp\times kp)}{\Lambda\vphantom{\Lambda_{\lu}}} & \defeq\diag\{\underset{(q\times q)}{\Lambda_{\lu}}\sep\underset{\phantom{q\times q}}{\Lambda_{\st}}\}\label{eq:RL}
\end{align}
such that: (a) the eigenvalues of $\Lambda_{\lu}$ and $\Lambda_{\st}$
correspond to the roots of $\Phi$, and lie in $\eigs_{\lu}^{\radius}$
and $\eigs_{\st}^{\radius}$ respectively; (b) the triple $(R_{\lu},\Lambda_{\lu},L_{\lu})$
satisfies 
\begin{align}
R_{\lu}\Lambda_{\lu}^{k}-\sum_{i=1}^{k}\Phi_{i}R_{\lu}\Lambda_{\lu}^{k-i} & =0 & \Lambda_{\lu}^{k}L_{\lu}^{\trans}-\sum_{i=1}^{k}\Lambda_{\lu}^{k-i}L_{\lu}^{\trans}\Phi_{i} & =0;\label{eq:eig-eig}
\end{align}
and (c) the (reduced-form) impulse response function of $y_{t}$ is
can be written as
\begin{equation}
\frac{\partial y_{t+s}}{\partial\err_{t}}=\irf_{s}^{\err}=R\Lambda^{k-1+s}L^{\trans}=R_{\lu}\Lambda_{\lu}^{k-1+s}L_{\lu}^{\trans}+R_{\st}\Lambda_{\st}^{k-1+s}L_{\st}^{\trans},\label{eq:irfdecomp}
\end{equation}
(see \lemref{GLR} and the subsequent remarks). Under \assref{J},
we have $\Lambda_{\lu}=I_{q}$ and $\rank R_{\lu}=\rank L_{\lu}=q$
(see \lemref{qcs}); it follows that the limits of the structural
and reduced-form IRFs as the horizon $s\goesto\infty$ are
\begin{align}
\lim_{s\goesto\infty}\irf_{s}^{\serr} & =R_{\lu}L_{\lu}^{\trans}\Upsilon, & \lim_{s\goesto\infty}\irf_{s}^{\err} & =R_{\lu}L_{\lu}^{\trans},\label{eq:limiting-irf}
\end{align}
yielding our final characterisation of the cointegrating space as
\begin{enumerate}[resume, resume*]
\item \label{enu:coint:Rlu}$\cs=(\spn R_{\lu})^{\perp}$.
\end{enumerate}
Consistent with the discussion of $S_{r}$ above, the matrix $\Upsilon$
-- and thus the identification of the structural shocks -- plays
no role here; the cointegrating space depends only on the column space
of the matrices appearing on the r.h.s.\ of \eqref{limiting-irf},
and so is invariant to post-multiplication by any full-rank matrix.\footnote{As \exaref{state-space} below illustrates, if $\{\err_{t}\}$ follows
a finite-order MA process -- such as typically arises when $\{y_{t}\}$
is generated by a linear state-space model -- then the autoregressive
coefficients, and the associated characteristic polynomial, alone
carry all the information required to recover $\cs$ (and similarly
for the quasi-cointegrating space introduced in the next section).}

\subsection{Cointegration without unit roots}

\label{subsec:qcs}

Having thoroughly characterised cointegration in a VAR with $q$ exact
unit roots, we may now return to our motivating problem: that of inference
on the cointegrating relationships when those $q$ roots are allowed
to be merely `near' to (real) unity. As shown by \citet{Ell98Ecta},
even if we consider the apparently favourable case of a sequence of
models whose roots drift towards unity as per $\Lambda_{\lu}=I+n^{-1}C$,
standard efficient estimators of the cointegrating relationships (such
as FM-OLS, DOLS and ML) are in general asymptotically biased, and
the associated inferences severely size distorted. This lack of robustness
is particularly disturbing because it arises in VARs that cannot be
consistently distinguished from those with exact unit roots, preventing
the extent of this problem from being empirically evaluated.

Our view is that this problem is fundamentally one of identification,
whose resolution demands a characterisation of cointegration that
retains its meaning over a wider domain than merely a VAR with exact
unit roots. Neither \enuref{coint:i0} nor \enuref{coint:Phi}, which
are implicitly utilised by the standard estimators, are fit for this
purpose. For if the largest $q$ roots of $\Phi$ were strictly inside
the unit circle -- as would now be permitted -- then \emph{all}
linear combinations of $y_{t}$ would be $I(0)$, $\Phi(1)$ would
have full rank, and hence both characterisations would identify the
cointegrating space with the whole of $\reals^{p}$. In contrast,
both \enuref{coint:irf} and \enuref{coint:Rlu} would continue to
identify that ($r$-dimensional) subspace of linear combinations of
$\{y_{t}\}$ having the least persistence, and thereby continue to
capture the long-run equilibrium relationships between these series.
Accordingly, they provide a sound basis on which to extend `cointegration'
to VARs without exact unit roots.

To this end, we now consider relaxing \assref{J} above as follows

\setcounter{assumption}{444}
\begin{assumption}
\label{ass:QC}Let $\radius\leq1$ be given.
\begin{enumerate}[label=\textnormal{\ass{QC\arabic*}},leftmargin=1.5cm]
\item \label{enu:Q:roots}$\Phi$ has $q$ roots in $\eigs_{\lu}^{\radius}$,
and all others in $\eigs_{\st}^{\radius}$.
\end{enumerate}
Let $\Lambda_{\lu}$ denote a real $(q\times q)$ matrix whose eigenvalues
correspond to the roots of $\Phi$ that are in $\eigs_{\lu}^{\radius}$,
and let $R_{\lu}$ and $L_{\lu}$ be $p\times q$ matrices that satisfy
\eqref{RL}--\eqref{irfdecomp}.
\begin{enumerate}[resume*]
\item \label{enu:Q:rank}$\rank R_{\lu}=\rank L_{\lu}=q$ and $\Lambda_{\lu}$
is diagonalisable.
\end{enumerate}
\end{assumption}
\enuref{Q:roots} is plainly the analogue of \enuref{J:roots}: whereas
we previously assumed $q$ roots at unity, we now allow for $q$ roots
in the vicinity of unity; indeed \assref{J} is a special case of
\assref{QC} with $\radius=1$ (\lemref{qcs}).\footnote{The construction of $\eigs_{\lu}^{\radius}$ requires the $q$ roots
to be on or inside the unit circle. Though the theory developed here
could also accommodate explosive roots -- simply by redefining $\eigs_{\lu}^{\radius}$
as $\{z\in\complex\mid\smlabs{z-1}\leq1-\radius\}$ -- we have deliberately
excluded such roots to be consistent with the greater empirical relevance
of stationary departures from unit roots for most applications.} By allowing for the possibility that $\radius<1$, we move from panel~(a)
of \figref{roots} to panel~(b). The requirement that $\Lambda_{\lu}$
be diagonalisable purposely rules out series that are integrated
of order two or higher (\citealp{dAut92}), since these are also excluded
under \assref{J} (which as noted above implies $\Lambda_{\lu}=I_{q}$).
For $\radius<1$ but `close' to unity, a model satisfying \assref{QC}
will thus inherit the main qualitative features of the cointegrated
VAR model: the high persistence of $\{y_{t}\}$, and the lesser persistence
of $r$ linear combinations of $\{y_{t}\}$, understood in terms of
\eqref{relative-irf} above.

Accordingly, the subspace $S_{r}$ spanned by the $r$ `least persistent'
linear combinations of $y_{t}$ retains an interpretation akin to
that of the cointegrating space. These two objects coincide exactly
in a VAR with unit roots (recall \enuref{coint:irf} above), but only
the former remains meaningfully interpretable when these roots are
allowed to be merely near to unity, as entertained by \assref{QC}.
That $S_{r}$ is always well defined under our assumptions is guaranteed
by the following, the proof of which appears in \appref{qcs-results}.
\begin{prop}
\label{prop:reptheorybody}Suppose \assref{DGP} and \assref{QC}
hold. Then $S_{r}=(\spn R_{\lu})^{\perp}$.
\end{prop}
We henceforth term the elements of $S_{r}$ the \emph{quasi-cointegrating
relationships}, and refer to $S_{r}$ itself as the \emph{quasi-cointegrating
space} (QCS), denoted 
\[
\qcs\defeq S_{r}=(\spn R_{\lu})^{\perp}.
\]
We continue to reserve the term \emph{cointegration} for the VAR with
$q$ exact unit roots. We also let $\beta\in\reals^{p\times r}$ denote
a matrix of rank $r$ whose columns span the QCS, and which therefore
has the property that $\beta^{\trans}R_{\lu}=0$.

There remains the question of how $\radius$ might be chosen in practice
as opposed to merely fixing it at unity. To build intuition on the
choice of $\rho$, we consider the reduced-form IRF \eqref{irfdecomp},
and the allied decomposition
\begin{equation}
y_{t}-\mu-\delta t=x_{t}=\Phi_{\lu}z_{\lu,t-1}+\Phi_{\st}z_{\st,t-1}+\err_{t}\label{eq:GJ-type-decomp}
\end{equation}
where $\Phi_{\lu}=R_{\lu}\Lambda_{\lu}^{k}$ and $\Phi_{\st}=R_{\st}\Lambda_{\st}^{k}$,
and the `common trend' $z_{\lu,t}\in\reals^{q}$ and `transitory'
$z_{\st,t}\in\reals^{kp-q}$ components follow\begin{subequations}\label{eq:zproc}
\begin{align}
z_{\lu,t} & =\Lambda_{\lu}z_{\lu,t-1}+\err_{\lu,t} & \err_{\lu,t} & \defeq L_{\lu}^{\trans}\err_{t}\label{eq:zLU}\\
z_{\st,t} & =\Lambda_{\st}z_{\st,t-1}+\err_{\st,t} & \err_{\st,t} & \defeq L_{\st}^{\trans}\err_{t}\label{eq:zST}
\end{align}
\end{subequations}under \assref{QC} (see \lemref{repasAR}). By
imposing a lower bound on the eigenvalues of $\Lambda_{\lu}$, the
value $\radius$ regulates the persistence of $z_{\lu,t}$. Equivalently,
via \eqref{irfdecomp}, $\radius$ is interpretable in terms of the
minimum half-life of the most persistent reduced-form shocks, $\err_{\lu,t}=L_{\lu}^{\trans}\err_{t}$,
that drive $y_{t}$, as being $\blw h\defeq-\log2/\log\radius$ periods.
As discussed in the next section, $L_{\lu}^{\trans}\err_{t}$ may
be given a structural interpretation, in terms of the subset of the
structural shocks $w_{t}$ that are identified (by the relevant macroeconomic
theory) as having highly persistent, and possibly permanent, effects.

In the extreme case where $\radius=1$, $z_{\lu,t}$ is an integrated
process and these shocks will have permanent effects, i.e.\ $\blw h=\infty$;
but in general reasonable finite choices for $\blw h$ will be available,
with smaller values of $\blw h$ (and hence of $\radius$) affording
greater robustness to departures from exact unit roots. This choice
will itself depend on the application at hand. For example, in a macroeconomic
context it would be appropriate to allow that the most persistent
shocks to $\{y_{t}\}$ may not have permanent effects, but still have
a half-life longer than the average duration of the business cycle:
with postwar US data of annual frequency, this corresponds to setting
$\blw h=8$ and thus $\radius=2^{-1/\blw h}=0.917$; or for quarterly
data, $\radius=0.979$.

\subsection{Implications for long-run identifying restrictions}

\label{subsec:lr-restrictions}

Up to this point, we have motivated quasi-cointegration in essentially
descriptive terms, as a means of extending the key time-series properties
of cointegrated systems to a wider domain. Of course, other characterisations
of `cointegration' or `long-run equilibria' in time series models
might be developed, and used to extend the concept in alternative
directions. From the point of view of empirical macroeconomics, however,
a particularly advantageous feature of quasi-cointegration is that
it is grounded in the SVAR, what is arguably the workhorse model in
this field (\citealp{KL17book}). Accordingly, as we shall now discuss,
quasi-cointegration both provides a framework for identifying structural
shocks via long-run restrictions, and a means of extracting testable
long-run predictions from macroeconomic theories, that does not require
one to take a stand on the presence or absence of exact unit roots
in the underlying SVAR -- a matter on which economic theory is largely
silent. Indeed, it does so more generally in the VARMA, or approximate
VAR, representations implied by linear state-space models, and is
thus relevant to a broad class of (linearised) structural macroeconomic
models.

\subsubsection{Structural impulse response functions}

A long-standing approach to the identification of structural IRFs
involves `long-run restrictions', which demarcate shocks according
to whether they are permitted to have permanent effects (e.g.\ \citealp{BlanchardQuah89};
\citealp{KPSW91AER}; \citealp{Gali99AER}; \citealp{CEV06NBER}).
These typically derive from an underlying theoretical model in which
one or more state variables -- such as total factor productivity
or the natural rate of interest -- are assumed to have a stochastic
trend. In equilibrium, these trends are imparted to some of the endogenous
variables, upon which the driving structural shocks must therefore
have permanent effects. When formulated in the setting of an SVAR,
$q$ stochastic trends manifest as $q$ unit roots and limiting impulse
response matrices \eqref{limiting-irf} of reduced rank, which identify
the relevant subset of the structural shocks as $L_{\lu}^{\trans}\err_{t}$.

Being expressed in terms of limiting impulse response matrices, these
restrictions are generally thought to require the presence of exact
unit roots to provide a viable approach to identification (see e.g.\ the
discussion in \citealp{KL17book}, Sec.~10.5.1). However, since these
restrictions are dual to the cointegrating relations -- the former
relates to the column span, the latter to the row span, of the long-run
IRF \eqref{limiting-irf} -- the former are just as amenable to being
extended beyond the setting of exact unit roots as is the latter.
As the examples below illustrate, if certain state variables are permitted
to, more generally, follow a highly persistent but not exactly integrated
autoregressive process, we obtain a SVAR(MA) process with $q$ roots
near unity (see also \citealp{campbell1994inspecting}, for a discussion
with $q=1$). The decomposition \eqref{irfdecomp}, which is generally
available under \assref{QC}, then provides a means of isolating the
driving structural shocks from those having comparatively transient
effects, with $L_{\lu}^{\trans}\err_{t}$ yielding $q$ linear combinations
of the former.
\begin{example}
\citet{Gali99AER} develops a stylised DSGE model of labour market
dynamics in the presence of a nominal rigidity, with two (i.i.d.\ and
mutually independent) structural shocks: one, $\eta_{t}$, to the
underlying technology process $\{Z_{t}\}$, which evolves as
\[
\log Z_{t}=\rho_{z}\log Z_{t-1}+\eta_{t}
\]
with $\rho_{z}=1$, i.e.\ as a random walk, and the other, $\xi_{t}$,
to the growth rate of the money supply. The model implies the following
VAR(1) representation for log productivity $a_{t}$ and hours $n_{t}$,
\begin{align}
x_{t}\defeq\begin{bmatrix}a_{t}\\
n_{t}
\end{bmatrix} & =c+\begin{bmatrix}\rho_{z} & \rho_{z}(1-\varphi)\\
0 & 0
\end{bmatrix}\begin{bmatrix}a_{t-1}\\
n_{t-1}
\end{bmatrix}+\varphi^{-1}\begin{bmatrix}\varphi-1 & \gamma(\varphi-1)+1\\
1 & -(1-\gamma)
\end{bmatrix}\begin{bmatrix}\xi_{t}\\
\eta_{t}
\end{bmatrix}\label{eq:GaliVAR}\\
 & \eqdef c+\Phi x_{t-1}+\err_{t}.\nonumber 
\end{align}
When $\rho_{z}=1$, only $\eta_{t}$ has a permanent effect on the
(log) level of productivity, which justifies an empirical strategy
of identifying the technology shock, in a bivariate SVAR of productivity
and hours, from the restriction that only it may have a permanent
effect on productivity. However, such a strategy is also viable when
$\rho_{z}$ is merely near unity, as can be seen by applying the decomposition
\eqref{irfdecomp} to the VAR \eqref{GaliVAR}, which for $h\geq1$
yields
\begin{equation}
\frac{\partial x_{t+h}}{\partial\err_{t}}=\Phi^{h}=\begin{bmatrix}1\\
0
\end{bmatrix}\rho_{z}^{h}\begin{bmatrix}1 & 1-\varphi\end{bmatrix}=\lambda_{\lu}^{h}r_{\lu}l_{\lu}^{\trans}\label{eq:gali-irf}
\end{equation}
with $\lambda_{\lu}=\rho_{z}$. In particular, irrespective of whether
$\rho_{z}=1$, $l_{\lu}^{\trans}$ recovers the technology shock,
since
\[
l_{\lu}^{\trans}\err_{t}=\varphi^{-1}\begin{bmatrix}1 & 1-\varphi\end{bmatrix}\begin{bmatrix}\varphi-1 & \gamma(\varphi-1)-1\\
1 & -(1-\gamma)
\end{bmatrix}\begin{bmatrix}\xi_{t}\\
\eta_{t}
\end{bmatrix}=\eta_{t}.\qedhere
\]
From \eqref{gali-irf}, we see that the implied quasi-cointegrating
relationship is $\beta=(0,1)^{\trans}\in(\spn r_{\lu})^{\perp}$,
i.e.\ that $\beta^{\trans}x_{t}=n_{t}$, consistent with the implication
of the model that technology shocks only have a long-lived effects
on productivity, and not on hours.
\end{example}
The preceding example yields a VAR(1) with a reduced rank autoregressive
matrix (with eigenvalues at zero and $\rho_{z}$). However, the same
points may be made more generally for linear structural models that
can be written in state-space form, which nests a wide range of structural
macroeconomic models, as follows. For simplicity, we consider a model
with a first-order state equation, but the conclusions carry over
straightforwardly to higher-order processes.
\begin{example}
\label{exa:state-space}Consider the state-space model\begin{subequations}\label{eq:state-space}
\begin{align}
x_{t} & =Ax_{t-1}+Bw_{t}\label{eq:state}\\
y_{t} & =Cx_{t-1}+Dw_{t}\label{eq:measure}
\end{align}
\end{subequations}in which each of $y_{t}$, $x_{t}$ and $w_{t}$
are $p$-dimensional (cf.\ \citealp{FRSW07AER}). The dynamics are
governed by the state equation \eqref{state}: if one or more of the
state variables $x_{t}$ are integrated, then $A$ will have (say)
$q$ unit eigenvalues and the long-run IRF for $x_{t}$, with respect
to $w_{t}$, will have rank $q$. Let us partition $w_{t}=(w_{1t}^{\trans},w_{2t}^{\trans})^{\trans}$,
where $w_{1t}$ is the $q$-dimensional subvector of shocks that have
permanent effects on $x_{t}$. In light of the preceding discussion,
it is not necessary to maintain that the persistence in the state
variables is generated by $q$ exact unit roots, but only that the
weaker requirement of \assref{QC} should hold, with $A$ having $q$
eigenvalues in $\eigs_{\lu}^{\radius}$. Regardless of the specific
values of these roots, by the analogue of decomposition \eqref{RL}--\eqref{irfdecomp}
for the state equation, we have
\[
\frac{\partial x_{t+s}}{\partial w_{t}}=R_{A,\lu}\Lambda_{A,\lu}^{k-1+s}L_{A,\lu}^{\trans}+R_{A,\st}\Lambda_{A,\st}^{k-1+s}L_{A,\st}^{\trans},
\]
and hence $L_{A,\lu}^{\trans}Bw_{t}=Mw_{1t}$, for some $M\in\reals^{q\times q}$
having full rank, since only the impact of these shocks decay at the
slower rate regulated by the eigenvalues of $\Lambda_{A,\lu}$.

Under this weaker assumption, the implied VAR(MA) representation of
the model yields the same long-run identifying restrictions as when
exact unit roots are present. Provided $C$ and $D$ are invertible,
\eqref{state-space} implies that
\begin{equation}
y_{t}=\Phi y_{t-1}+\err_{t}-\Psi\err_{t-1}\label{eq:VARMA}
\end{equation}
where $\Phi=CAC^{-1}$, $\Psi=\Phi-CBD^{-1}$, and $\err_{t}\defeq Dw_{t}$
are the reduced-form shocks. Because $\Phi$ and $A$ are similar,
the characteristic roots in the state equation \eqref{state} coincide
exactly with those in \eqref{VARMA}; in particular, both systems
are characterised by $q$ roots in $\eigs_{\lu}^{\radius}$, and $L_{\lu}=(C^{-1})^{\trans}L_{A,\lu}$.
Because of the MA component, the reduced-form IRF takes the modified
form $\frac{\partial y_{t+h}}{\partial\err_{t}}=\Phi^{h-1}(\Phi-\Psi)$,
and it is no longer the case that $L_{\lu}^{\trans}\err_{t}$ recovers
the shocks $w_{1t}$ driving the common persistent components; instead,
we have
\begin{equation}
L_{\lu}^{\trans}(\Phi-\Psi)\err_{t}=L_{A,\lu}^{\trans}C^{-1}[CBD^{-1}]Dw_{t}=L_{A,\lu}^{\trans}Bw_{t}=Mw_{1t},\label{eq:MAident}
\end{equation}
where $L_{\lu}^{\trans}(\Phi-\Psi)$ depends only on the reduced-form
parameters.

In practice, \eqref{VARMA} is rarely estimated; the usual approach
is to approximate by a finite-order VAR, truncating the l.h.s.\ of
the representation
\begin{equation}
\sum_{i=0}^{\infty}\Psi^{i}(y_{t-i}-\Phi y_{t-1-i})=\err_{t}\label{eq:VARinfty}
\end{equation}
at some finite $k$.\footnote{If $\Psi$, or equivalently $A-BD^{-1}C$, is a nilpotent matrix,
then there exists an exact finite order VAR representation for the
system, of some order $k^{\ast}$, and the following claims -- in
particular \eqref{approxL} -- hold for any $k\geq k^{\ast}$, rather
than merely in the limit (cf.\ \citealp{Rav07JME}, Cor.~2.2).} Suppose that the eigenvalues of $\Psi$ lie in $\eigs_{\st}^{\radius}$,
so that these may be distinguished from the $q$ dominant eigenvalues
of $\Phi$. Letting $\Psi_{k-1}(\lambda)\defeq I_{p}\lambda^{k-1}+\Psi\lambda^{k-2}+\cdots+\Psi^{k-1}$,
the truncated VAR($k$) has characteristic polynomial $\Gamma(\lambda)\defeq\Psi_{k}(\lambda)(I_{p}\lambda-\Phi)$,
whose roots are the eigenvalues of $\Phi$ (and therefore of $A$),
and otherwise complex rotations of the eigenvalues of $\Psi$. Then
the truncated VAR satisfies \assref{QC}, and if we apply the decomposition
\eqref{irfdecomp} to the truncated VAR, then
\begin{equation}
L_{k,\lu}^{\trans}\err_{t}\goesto L_{\lu}^{\trans}(\Phi-\Psi)\err_{t}=Mw_{1t}\label{eq:approxL}
\end{equation}
as $k\goesto\infty$, where the equality holds by \eqref{MAident},
and thus the correct shocks are recovered in the limit, as the order
of the VAR approximation grows. Remarkably, $R_{k,\lu}=R_{\lu}$ for
all $k$, so the quasi-cointegrating relations are carried \emph{exactly}
by the truncated VAR.
\end{example}

\subsubsection{Long-run predictions of macroeconomic theories}

\label{subsec:long-run-pred}

A second respect in which quasi-contegration is empirically useful
is in testing what might be termed the long-run predictions of economic
theories, which in an (S)VAR with exact unit roots would be embodied
in the cointegrating relations between the series. In a range of structural
models, the dependence of the endogenous variables on a common set
of state variables, combined with an elevated degree of persistence
in the mechanisms generating some of those variables, manifests as
a collection of long-run equilibrium relationships between those variables.
If the theory underlying the structural model makes definite predictions
about the coefficients parametrising these relationships, this provides
a means of testing the theory, and cointegration analysis has been
widely applied to this end (such as to testing business cycle models,
theories of purchasing power parity, and the expectations theory of
the term structure). Quasi-cointegration provides a means of continuing
to conduct tests of this kind without having to maintain the auxiliary
assumption of exact unit roots, by providing a means of expressing
these long-run equilibrium relationships in a way that is robust to
departures from this assumption. The following example illustrates
this in detail, and we shall return to it in in our empirical application
in \secref{empirical}.
\begin{example}
In its simplest form, the expectations theory of the term structure
holds that the yield on a (zero-coupon) bond should be equal to the
sum of the expected future yields on a shorter-dated bond (see e.g.\ \citealp{LS04book},
Ch.~13, for a textbook derivation from a dynamic asset pricing model
under risk neutrality). If the reduced-form process followed by the
yields on these two bonds follows a bivariate VAR with a single unit
root, then it has long been known that a major implication of this
theory is that the the (annualised) rate of return on these bonds
should be cointegrated, with $\beta=(1,-1)^{\trans}$ (see e.g.\ \citealp{10.2307/1833129};
\citealp{CARRIERO2006339}). However, as we shall now show, that implied
one-for-one long-run equilibrium relationship is contingent on the
assumption of an exact unit root -- a contingency also noted previously,
albeit in a simplified setting in which one of the yields follows
a reduced-form AR(1) process, by \citet{MW13JoE}.

Let $\ret_{i,t}$ denote the (annualised) yield on a zero-coupon bond
with $i$ years to maturity (in year $t$), and suppose that we observe
data on both a 1-year and an $m$-year bond, generated by a VAR($k$)
that satisfies \assref{QC} for some $\radius\leq1$, with $\delta=0$
and reduced form errors $\{\err_{t}\}$. The loglinearised form the
expectations theory implies that
\begin{equation}
\ret_{m,t}=\frac{1}{m}\sum_{i=0}^{m-1}\expect_{t}\ret_{1,t+i}+\xi_{t},\label{eq:exptheory}
\end{equation}
where $\{\xi_{t}\}$ captures the term premium. To keep \eqref{exptheory}
consistent with a VAR for $\ret_{t}\defeq(\ret_{m,t},\ret_{1,t})$,
$\{\xi_{t}\}$ must be a linear process of the form $\xi_{t}=\sum_{i=0}^{\infty}\psi_{i}\err_{t-i}$.
In a setting with exact unit roots ($\radius=1$), one would assume
that $\xi_{t}\sim I(0)$. To allow the theory to retain predictive
content in our more general setting ($\radius\leq1$), we make the
analogous assumption that $\xi_{t}$ is strictly less persistent than
$\ret_{t}$ itself, in the sense that $\radius^{-h}\psi_{h}\goesto0$
as $h\goesto\infty$. 

Recognising that $\ret_{1,t+i}=\sum_{k=0}^{i-1}\lambda^{k}\Delta_{\lambda}\ret_{1,t+i-k}+\lambda^{i}\ret_{1,t}$,
we may rewrite the preceding for any $\lambda\in[0,1]$ as
\begin{equation}
\ret_{m,t}-a_{m}(\lambda)\ret_{1,t}\defeq\ret_{m,t}-\frac{1}{m}\frac{1-\lambda^{m}}{1-\lambda}\ret_{1,t}=\frac{1}{m}\sum_{i=0}^{m-1}\sum_{j=0}^{i-1}\lambda^{j}\expect_{t}\Delta_{\lambda}\ret_{1,t+i-j}+\xi_{t}.\label{eq:am}
\end{equation}
In particular, if we take $\lambda=\lambda_{\lu}$ (where $\lambda_{\lu}$
corresponds to the root nearest to unity in the bivariate VAR representation
of the yields), then by \lemref{quasi-VECM} in \appref{qcs-results},
\[
\frac{1}{m}\sum_{i=0}^{m-1}\sum_{j=0}^{i-1}\lambda_{\lu}^{j}\expect_{t}\Delta_{\lambda_{\lu}}\ret_{1,t+i-j}=\gamma_{0}\beta^{\trans}\ret_{t}+\sum_{i=0}^{p-1}\gamma_{i}\Delta_{\lambda_{\lu}}\ret_{t-i}\eqdef\varsigma_{t}
\]
for some $\{\gamma_{i}\}_{i=0}^{p-1}$ that depends on $m$, $\lambda_{\lu}$
and the VAR coefficients. By that same result, $(\beta^{\trans}\ret_{t},\Delta_{\lambda_{\lu}}\ret_{t})$
also follows a VAR, whose roots lie in $\eigs_{\st}^{\radius}$. Thus
$\radius^{-h}\frac{\partial\varsigma_{t+h}}{\partial\err_{t}}\goesto0$
as $h\goesto\infty$, and hence the same is true of $\ret_{m,t}-a_{m}(\lambda_{\lu})\ret_{1,t}$.
In this manner, one of the principal implications of the expectations
theory may be generalised. Rather than merely implying that yields
$\ret_{t}\defeq(\ret_{m,t},\ret_{1,t})$ should be cointegrated (with
unit coefficient), in a VAR with an exact unit root, the theory more
generally entails that these are quasi-cointegrated, with
\[
\beta(\lambda_{\lu})^{\trans}=[\begin{matrix}1 & -a_{m}(\lambda_{\lu})\end{matrix}].\qedhere
\]
\end{example}

\subsection{Connections to the literature}

\label{subsec:literature}

As noted in the introduction, there have been relatively few attempts
to address the problem identified by \citet{Ell98Ecta}, most notably
\citet{Wri00JBES}, \citet{MP09}, \citet{MW13JoE}, \citet{FJ17Ect},
and \citet{HV23JBES}. Having now outlined our own approach to this
problem, we may briefly explain how our work is situated relative
to those contributions.

Insofar as they also consider a VAR model with some characteristic
roots near unity, \citet{FJ17Ect} relates closely to the present
study. Their setting is a VAR model with one lag, written in error
correction form as
\begin{equation}
\Delta x_{t}=(\alpha\beta^{\trans}+\alpha_{1}\Gamma\beta_{1}^{\trans})x_{t-1}+\err_{t}\eqdef\Pi x_{t-1}+\err_{t},\label{eq:FJvar}
\end{equation}
where $\alpha,\beta\in\reals^{p\times r}$ and $\alpha_{1},\beta_{1}\in\reals^{p\times q}$
have full column rank, and $\Gamma\in\reals^{q\times q}$. When $\Gamma=0$,
the model specialises exactly to the CVAR model of \subsecref{coint-ur}
with $q$ unit roots and $\cs=\spn\beta$. If some elements of $\Gamma$
are non-zero, the cointegrated model becomes one with some roots near
but not equal to unity, and $\Pi$ need no longer be of reduced rank.
As the authors acknowledge, if each of $\alpha_{1}$, $\beta_{1}$
and $\Gamma$ are freely varying, then $\beta$ is not identified.
They accordingly treat $\alpha_{1}$ and $\beta_{1}$ as known, which
restores identification and facilitates likelihood-based inference
on each of $\alpha$, $\beta$ and $\Gamma$. But while a priori knowledge
of $\alpha_{1}$ and $\beta_{1}$ may indeed be available in certain
situations, this seems unlikely to be the case in general; whereas
in our approach, the criterion \eqref{relative-irf} ensures that
$\beta$ remains identified even as the dominant roots depart from
unity. Moreover, with $\Gamma$ fixed (and nonzero) in this model,
it is unclear how $\beta$ in \eqref{FJvar} could be interpreted
in terms of long-run relationships between the elements of $x_{t}$.

\citet{MP09} consider \citeauthor{Ell98Ecta}'s \citeyearpar{Ell98Ecta}
problem in the context of the triangular model\begin{subequations}\label{eq:triangular}
\begin{align}
x_{1t} & =Ax_{2t}+u_{1t}\label{eq:triangular-coint}\\
x_{2t} & =R_{n}x_{2,t-1}+u_{2t}
\end{align}
\end{subequations}where $u_{t}=(u_{1t}^{\trans},u_{2t}^{\trans})^{\trans}$
is a weakly dependent linear process. When $R_{n}=I_{q}$, this encompasses
the $I(0)/I(1)$ CVAR model with $q$ unit roots, but allows for a
more general semiparametric treatment of the model's short-run dynamics.
If $R_{n}$ instead merely drifts towards $I_{q}$ (possibly at a
slower rate than $n^{-1}$) as $n\goesto\infty$, then the authors
show that it is still possible to obtain an asymptotically mixed normal
estimate of $A$, by using instruments that are constructed by filtering
$x_{2t}$ (what they term the `IVX' estimator of $A$). However,
the greater generality afforded by the triangular model comes at the
price that $R_{n}\goesto I_{q}$ is now necessary for identification
of $A$; if on the other hand $R_{n}$ were fixed with eigenvalues
strictly less than unity, then all linear combinations of $x_{t}$
would be weakly dependent, leaving $A$ unidentified. This is true
generally of approaches that rely on the triangular form, because
of its agnosticism about the dynamics of $u_{t}$; thus the same point
may be made in the context of \citet{HV23JBES}, who (when $R_{n}=I_{q}+n^{-1}C$)
consider an augmented regression estimator of \eqref{triangular-coint}
using low frequency transforms of the original data, and a Bonferroni-based
approach to correct for the effect of $C$ on its limiting distribution.

Finally, \citet{MW13JoE} consider a very general setting in which
the `common trends' in $x_{t}$ are permitted to belong to a broad
family of processes. A consequence of this generality is that the
authors conceptualise `cointegration' in terms different from quasi-cointegration,
and the two definitions do not always agree. Essentially, \citeauthor{MW13JoE}
define $x_{t}$ to be `cointegrated' with cointegrating relations
$\beta\in\reals^{p\times r}$, if $n^{-1/2}\sum_{t=1}^{\smlfloor{nr}}\beta^{\trans}x_{t}$
converges weakly to a Brownian motion, while the common trends $n^{-1/2}\beta_{\perp}^{\trans}x_{\smlfloor{nr}}$
converge weakly to a some cadlag process (where $\beta_{\perp}\in\reals^{p\times q}$
has $\rank\beta_{\perp}=q$ and $\beta_{\perp}^{\trans}\beta=0$).
In the context of our VAR model, where \assref{QC} holds for some
$\radius<1$, $n^{-1/2}\sum_{t=1}^{\smlfloor{nr}}x_{t}$ converges
weakly to a Brownian motion if all the roots are strictly inside the
unit circle; so in such a case there is no `cointegration' in the
foregoing sense, even though quasi-cointegrating relationships would
be defined. (On the other hand, if the largest $q$ roots of $\Phi$
are localised to unity at rate $n^{-1}$, though not more slowly,
then their `cointegrating' vectors would coincide with our quasi-cointegrating
vectors.) Regarding inference, the authors construct a confidence
set for $\beta$ by inverting a stationarity test for $\beta^{\trans}x_{t}$,
extending a idea originally due to \citet{Wri00JBES}. For a comparison
of their tests with ours, in terms of size and power, see the simulations
in \secref{simulations} below: these indicate that the price paid,
in terms of power, for robustness to a broader class of trend generating
mechanisms (than are permitted by the VAR), may be considerable.

Thus in relation to these papers, one of the major distinguishing
contributions of our work is to provide a means of identifying long-run
equilibrium relationships that is well-defined for a \emph{fixed}
parametrisation of the underlying (S)VAR; i.e.\ we do not rely on
that VAR drifting towards a model with exact unit roots (as $n\goesto\infty$).
Not only is our identifying criterion \eqref{relative-irf} readily
interpretable in terms of the relative persistence of either structural
or reduced-form impulse responses, but it also maintains the duality
that exists, both with and without exact unit roots, between the identification
of long-run equilibrium relationships, and of those structural shocks
whose common persistent effects give rise to those relationships.
As we will discuss subsequently, an added benefit is that it reduces
\citeauthor{Ell98Ecta}'s \citeyearpar{Ell98Ecta} problem to one
that is asymptotically equivalent to inference in a multivariate predictive
regression, a canonical problem that has given rise to a rich literature.

\section{Estimation and inference}

\label{sec:estimation}

\subsection{Formulation of the problem}

\subsubsection{Model likelihood}

\label{subsec:likelihood}

As with the CS in a cointegrated VAR model, inference on the QCS in
our more general setting will be based on the normal model likelihood
(or quasi-likelihood, if $\err_{t}$ is not in fact normally distributed).
Recall that the model \eqref{dgp} may be rendered as
\begin{equation}
y_{t}=m+dt+\sum_{i=1}^{k}\Phi_{i}y_{t-i}+\err_{t}.\label{eq:redform}
\end{equation}
To facilitate the exposition, we focus on the case where the intercept
and trend parameters ($m,d$) are unrestricted in \eqref{redform},
while maintaining that the data is generated under \eqref{dgp}, so
as to exclude the possibility of a quadratic trend in $y_{t}$. (For
a discussion of alternative potential treatments of the deterministic
terms, see \subsecref{deterministics} below.)

The loglikelihood with $(m,d)$ concentrated out -- or equivalently,
expressed in terms of a maximal invariant for transformations of the
form $\{y_{t}\}\elmap\{m+dt+y_{t}\}$ -- may be written as
\[
\like_{n}(\PHI,\Sigma)\defeq-\frac{n}{2}\log(2\pi\det\Sigma)-\min_{m,d}\frac{1}{2}\sum_{t=1}^{n}\norm{y_{t}-m-dt-\sum_{i=1}^{k}\Phi_{i}y_{t-i}}_{\Sigma^{-1}}^{2},
\]
where $\smlnorm x_{W}^{2}\defeq x^{\trans}Wx$ for $x\in\reals^{p}$
and $W\in\reals^{p\times p}$ positive semidefinite. The QCS depends
only on $\PHI$, and the main (asymptotic) results of this paper are
not sensitive to the choice of estimator for $\Sigma$, provided that
it is consistent. To simplify our arguments, we shall therefore generally
assume that the unrestricted ML estimator $\hat{\Sigma}_{n}$, i.e.\ the
OLS variance estimator, is used. Henceforth, let $\likens_{n}(\PHI)\defeq\like_{n}(\PHI,\hat{\Sigma}_{n})$;
for convenience we shall refer to maximisers of $\likens_{n}$ as
`maximum likelihood estimators'.

\subsubsection{QCS as a functional of the VAR coefficients}

\label{subsec:functionals}

Under \assref{QC}, the $\qcs$ is well-defined and has dimension
$q$. Since any basis $\beta\in\reals^{p\times q}$ for the QCS is
only identified up to its column space, and has rank $q$, it is convenient
to maintain the normalisation
\begin{equation}
\beta^{\trans}=[\begin{matrix}I_{r} & -A\end{matrix}],\label{eq:qcsnorm}
\end{equation}
so that inference on the QCS reduces to inference on the elements
of the matrix $A\in\reals^{r\times q}$. \eqref{qcsnorm} is not restrictive
-- i.e.\ it is indeed merely a normalisation of $\beta$ -- if
the QCS does not contain any nonzero vectors whose first $r$ elements
are all zero, as will be the case if the elements of $y_{t}$ are
ordered appropriately. Since $R_{\lu}$ has rank $q$ and $\beta^{\trans}R_{\lu}=0$,
\eqref{qcsnorm} is equivalent to
\begin{equation}
R_{\lu}=\begin{bmatrix}A\\
I_{q}
\end{bmatrix}.\label{eq:Rlurnom}
\end{equation}

Since the $q$ roots in $\eigs_{\lu}^{\radius}$ are separated from
the $kp-q$ roots in $\eigs_{\st}^{\radius}$, the column space of
$R_{\lu}$ depends smoothly on the VAR coefficients. To express this
formally, let $\lambda_{i}(\PHI)$ denote the $i$th root of the characteristic
polynomial associated to the VAR with coefficients $\PHI$, when these
are placed in descending order of modulus; and $G^{\trans}\defeq[0_{q\times r},I_{q}]$.
Define $\set P\subset\reals^{p\times kp}$\label{subsec:Psetdef}
to be the set of VAR coefficients such that: (i) $\smlabs{\lambda_{q+1}(\PHI)}<\smlabs{\lambda_{q}(\PHI)}$;
(ii) there exist $R_{\lu}\in\reals^{p\times q}$ and $\Lambda_{\lu}\in\reals^{q\times q}$
such that the eigenvalues of $\Lambda_{\lu}$ are $\{\lambda_{1}(\PHI),\ldots,\lambda_{q}(\PHI)\}$,
\begin{equation}
R_{\lu}\Lambda_{\lu}^{k}-\sum_{i=1}^{k}\Phi_{i}R_{\lu}\Lambda_{\lu}^{k-i}=0;\label{eq:PsetRlu}
\end{equation}
and (iii) $\rank\{G^{\trans}R_{\lu}\}=q$. Then $\set P$ is open,
and since $G^{\trans}R_{\lu}$ has full rank, we may choose $(R_{\lu},\Lambda_{\lu})$
to be consistent with the normalisation \eqref{Rlurnom}. The conditions
defining $\set P$, together with \eqref{Rlurnom}, implicitly define
smooth (i.e.\ infinitely differentiable) maps $R_{\lu}(\PHI)$, $A(\PHI)$,
and $\Lambda_{\lu}(\PHI)$ on $\set P$ (\lemref{implicit-maps}).
In light of this, inference on the QCS may be rephrased in terms of
inference on parameters $A=A(\PHI)$ defined by a smooth transformation
of the VAR coefficients.

\subsubsection{Parameter space for the near-unit roots}

\label{subsec:LambdaLUparspc}

In a model with exact unit roots, efficient estimation of $A$ requires
the model to be estimated under some of the restrictions implied by
\assref{J}: for example, ML estimation of $A$ proceeds under the
assumption that $\rank\Phi(1)=r$, as per \enuref{J:rank}. In a setting
with only $I(1)$ series, this is equivalent to maintaining $q=p-r$
roots at real unity. Transposed to the present setting, with \assref{QC}
taking the place of \assref{J}, we now require the model be estimated
with $q$ roots lying in $\eigs_{\lu}^{\radius}$, as per \enuref{Q:roots}.
This entails both a choice of $\radius$, and the specification of
an appropriate parameter space $\set L\subset\reals^{q\times q}$
for $\Lambda_{\lu}$, such that it is diagonalisable (as per \enuref{Q:rank}),
with eigenvalues lying in $\eigs_{\lu}^{\radius}$.

When $q=1$, $\Lambda_{\lu}$ is scalar, so $\set L=[\radius,1]$
and only a lower bound $\radius$ on the largest root of $\Phi$ needs
to be specified. A discussion of the considerations that might inform
$\radius$ were given at the end of \subsecref{qcs} (see also the
application in \secref{empirical} below). When $q\geq2$, $\set L$
is instead a set of matrices with eigenvalues lying in the interval
$[\radius,1]$. While \enuref{Q:rank} might suggest taking $\set L$
to be the subset $\Ld$ of real, diagonalisable $q\times q$ matrices,
a potential difficulty with $\Ld$ is that is that some non-diagonalisable
matrices are in its closure, as can be seen e.g.\ by taking the limit
of $[\begin{smallmatrix}\lambda+\epsilon & 1\\
0 & \lambda-\epsilon
\end{smallmatrix}]$ as $\epsilon\goesto0$. This would effectively permit departures
from the $I(0)/I(1)$ cointegrated VAR model in the direction of a
model with some $I(2)$ components, something that we wish to avoid
here. The set of either normal ($\Ln$) or symmetric ($\Ls$) matrices
with eigenvalues in $[\radius,1]$ would thus be more appropriate
choices for $\set L$, since each give a closed subset of the set
of diagonalisable matrices, with the principal difference between
the two being that the former allows for complex eigenvalues.

\subsubsection{Local-to-unity asymptotics}

\label{subsec:loc-to-unity}

The $\qcs$, and the associated coefficient matrix $A$, remain identified
so long as the roots of $\Phi$ separate as prescribed by \assref{QC}.
In particular, there is no requirement that the roots in $\eigs_{\lu}^{\radius}$
should drift towards unity at any rate, as $n\goesto\infty$. However,
as will become evident below, the proximity of those $q$ largest
roots to unity affects the distributions of estimators and test statistics,
even in very large samples. We therefore need to work with a sequence
of models that preserves this dependence in the limit, and avoids
the discontinuities in the asymptotics that would otherwise arise
at exact unit roots. We shall accordingly study the large-sample behaviour
of the likelihood, and of derived estimators and test statistics,
under\setcounter{assumption}{8504}
\begin{assumption}
\label{ass:LOC} $\{(y_{t},x_{t})\}_{t=1}^{n}$ is generated per \assref{DGP}
with $\PHI=\PHI_{n}$, for $\{\PHI_{n}\}\subset\set P$ such that:
\begin{enumerate}[label=\textnormal{\ass{LOC\arabic*}},leftmargin=1.5cm]
\item for some $C\in\reals^{q\times q}$ with non-positive eigenvalues,
\begin{equation}
\Lambda_{\lu}(\PHI_{n})=\Lambda_{n,\lu}\defeq I_{q}+n^{-1}C;\label{eq:loctounity}
\end{equation}
\item $R_{\lu}(\PHI_{n})=[\begin{smallmatrix}A\\
I_{q}
\end{smallmatrix}]$ for some $A\in\reals^{r\times q}$ ; and
\end{enumerate}
letting $R_{n,\st}$, $\Lambda_{n,\st}$ and $L_{n}=[L_{n,\lu},L_{n,\st}]$
be such that \eqref{RL}--\eqref{irfdecomp} hold for each $n$:
\begin{enumerate}[resume, resume*]
\item \label{enu:LOC:stat}$R_{n,\st}=R_{\st}$ and $\Lambda_{n,\st}=\Lambda_{\st}$
are fixed, and the eigenvalues of the latter lie strictly inside the
complex unit circle; and
\item $L_{\lu}\defeq\lim_{n\goesto\infty}L_{n,\lu}$ has full column rank.
\end{enumerate}
\end{assumption}
Under \assref{LOC}, we may choose a $\radius<1$ such that \assref{QC}
holds for all $n$ sufficiently large; \assref{LOC} may thus be regarded
as capturing (sequences of) VAR models that are both in the immediate
vicinity of a cointegrated VAR with unit roots, while also satisfying
those regularity conditions that ensure the QCS is well-defined. The
localisation \eqref{loctounity} moreover entails a sharper delineation
between the common trend and transitory components appearing in the
implied decomposition of $y_{t}$ in \eqref{GJ-type-decomp}, since
we now have the joint weak convergences
\begin{align}
n^{-1/2}\sum_{t=1}^{n}\err_{t} & \wkc E(r) & n^{-1/2}z_{\lu,\smlfloor{nr}}\wkc\int_{0}^{r}\e^{C(r-s)}L_{\lu}^{\trans}\deriv E(s) & \eqdef Z_{C}(r),\label{eq:Zproc}
\end{align}
on $D[0,1]$, for $E$ a $p$-dimensional Brownian motion with variance
$\Sigma$; and thus
\begin{equation}
n^{-1/2}x_{\smlfloor{nr}}=\Phi_{n,\lu}n^{-1/2}z_{\lu,\smlfloor{nr}}+o_{p}(1)=_{d}\Phi_{n,\lu}Z_{C}(r)+o_{p}(1)\label{eq:x-nearint}
\end{equation}
so that $z_{\lu,t}$ and $x_{t}$ are nearly integrated. Although
$\Phi_{n,\lu}=R_{\lu}\Lambda_{n,\lu}^{k}$ depends on $n$, its column
space does not, and
\begin{equation}
\beta^{\trans}x_{t}=\beta^{\trans}\Phi_{\st}z_{\st,t-1}+\beta^{\trans}\err_{t}\sim I(0).\label{eq:betax_t}
\end{equation}
Thus, analogously to the GJRT, \eqref{GJ-type-decomp} decomposes
$x_{t}$, and therefore also $y_{t}$ (upon detrending), into the
sum of a nearly integrated component and an $I(0)$ component; the
quasi-cointegrating relations are precisely those that eliminate the
nearly integrated common trends from $y_{t}$.

\subsection{Asymptotics of the loglikelihood ratio process}

\label{subsec:likelihoodasymp}

We first consider the asymptotic behaviour of the loglikelihood ratio
process, under a local reparametrisation of the VAR given by
\begin{align}
\vall & \defeq\vall_{n}(\PHI)\defeq n\vek\begin{bmatrix}A(\PHI)-A(\PHI_{n})\\
\Lambda_{\lu}(\PHI)-\Lambda_{\lu}(\PHI_{n})
\end{bmatrix}, & f\defeq f_{n}(\PHI) & \defeq n^{1/2}\vek\{(\PHI-\PHI_{n})\R_{n,\st}\};\label{eq:reparm-thm}
\end{align}
where $\{\PHI_{n}\}$ is as in \assref{LOC}, and $\R_{n,\st}$ is
a $kp\times(kp-q)$ matrix defined in \appref{RnSTdef}. \eqref{reparm-thm}
effectively isolates the signal from the nearly integrated and $I(0)$
components of $x_{t}$, as given in \eqref{x-nearint}--\eqref{betax_t},
with only the former carrying (asymptotically) information relevant
to the estimation of $A$ and $\Lambda_{\lu}$. These parameters will
thus enjoy elevated rate of convergence, relative to the other components
of the VAR, just as is familiar from the VAR with exact unit roots.
In connection with \eqref{reparm-thm}, let
\begin{equation}
K\defeq\begin{bmatrix}\beta^{\trans}R_{\st}(I-\Lambda_{\st})^{-1}L_{\st}^{\trans}\\
L_{\lu}^{\trans}
\end{bmatrix}\eqdef\begin{bmatrix}{\cal J}\\
L_{\lu}^{\trans}
\end{bmatrix}.\label{eq:Kdef}
\end{equation}
denote a component of limiting Jacobian matrix for $\vall_{n}(\PHI)$
at $\PHI=\PHI_{n}$.

The asymptotics reveal the close correspondence that exists between
the problem of inference on $(A,\Lambda_{\lu})$ in a quasi-cointegrated
VAR, and of inference in a predictive regression model with highly
persistent regressors. This is of interest because the latter is a
canonical problem that has received considerable attention in the
literature, where a variety of inferential procedures with good size
and power properties have been developed. A major implication of our
next result is that such procedures should enjoy similar properties,
once transposed to our setting. While we do not develop this possibility
further here, the good properties found for the \citet{EMW15Ecta}
procedure in \secref{simulations} are of little surprise, in view
of the similar performance enjoyed by that method in a predictive
regression.

By a \emph{predictive regression model}, we mean a model of the form\begin{subequations}\label{eq:predreg-body}
\begin{align}
y_{\pr,t} & =m_{y}+d_{y}t+A_{\pr}z_{\pr,t-1}+\xi_{y,t}\label{eq:y-pr}\\
z_{\pr,t} & =m_{z}+d_{z}t+\Lambda_{\pr}z_{\pr,t-1}+\xi_{z,t}\label{eq:z-pr}
\end{align}
\end{subequations}where $\{y_{\pr,t}\}$ and $\{z_{\pr,t}\}$ respectively
take values in $\reals^{r}$ and $\reals^{q}$. To complete the specification
of the model, suppose the following:

\setcounter{assumption}{433}
\begin{assumption}
\label{ass:PR}$\{y_{\pr,t}\}_{t=1}^{n}$ and $\{z_{\pr,t}\}_{t=1}^{n}$
are generated under \eqref{predreg-body}, where
\begin{enumerate}[label=\textnormal{\ass{PR\arabic*}},leftmargin=1.5cm]
\item $A_{\pr}\in\reals^{r\times q}$ and $\Lambda_{\pr}=I_{q}+n^{-1}C$,
for $C$ as in \assref{LOC};
\item $z_{\pr,0}=0$, $m_{y}=d_{y}=0$ and $m_{z}=d_{z}=0$;
\item \label{enu:PR:normal}$\xi_{t}\defeq(\xi_{y,t}^{\trans},\xi_{z,t}^{\trans})^{\trans}\distiid N[0,\Omega]$.
\end{enumerate}
\end{assumption}
To simplify the discussion, the covariance matrix $\Omega$ is assumed
to be known. Under \enuref{PR:normal}, the model loglikelihood is
\[
\like_{n}^{\pr}(A,\Lambda)\defeq-\frac{n}{2}\log(2\pi\det\Omega)-\min_{m,d}\frac{1}{2}\sum_{t=1}^{n}\norm{\begin{bmatrix}y_{\pr,t}\\
z_{\pr,t}
\end{bmatrix}-\begin{bmatrix}m_{y}\\
m_{z}
\end{bmatrix}-\begin{bmatrix}d_{y}\\
d_{t}
\end{bmatrix}t-\begin{bmatrix}A\\
\Lambda
\end{bmatrix}z_{\pr,t-1}}_{\Omega^{-1}}^{2}.
\]
Let $\Zdet_{C}(r)$ denote the residual of an $L^{2}[0,1]$ projection
of each sample path of $Z_{C}$ in \eqref{Zproc} onto a constant
and linear trend. The proof of the next result, and of all other theorems,
appears in \appref{theoremproofs}.
\begin{thm}
\label{thm:emw}Suppose that:
\begin{enumerate}
\item $\{y_{t}\}$ is generated under \assref{LOC}. Let $\like_{n}(\vall,f)\defeq\like_{n}(\PHI,\Sigma)$,
where $\vall=\vall_{n}(\PHI)$ and $f=f_{n}(\PHI)$ as in \eqref{reparm-thm},
and let $\likens_{n}(\vall)\defeq\max_{\vall(\PHI)=\vall}\likens_{n}(\PHI)$.
\item $\{(y_{\pr,t},z_{\pr,t})\}$ is generated under \assref{PR}, with
$\Omega=K\Sigma K^{\trans}$. Let $\like_{n}^{\pr}(\vall)\defeq\like_{n}^{\pr}(A,\Lambda)$,
where
\[
\vall=n\vek\begin{bmatrix}A-A_{\pr}\\
\Lambda-(I_{q}+n^{-1}C)
\end{bmatrix}.
\]
\end{enumerate}
Then the finite-dimensional distributions of each of $\{\like_{n}(\vall,f)-\like_{n}(0,f)\}$,
$\{\likens_{n}(\vall)-\likens_{n}(0)\}$ and $\{\like_{n}^{\pr}(\vall)-\like_{n}^{\pr}(0)\}$
converge to those of
\begin{equation}
S_{\vall}^{\trans}\vall-\tfrac{1}{2}\vall^{\trans}H_{\vall}\vall\label{eq:lim-exp}
\end{equation}
where, for $W$ a $p$-dimensional standard Brownian motion on $[0,1]$,
\begin{align}
S_{\vall} & \defeq\int_{0}^{1}[\Zdet_{C}(r)\otimes(K\Sigma K^{\trans})^{-1/2}\deriv W(r)] & H_{\vall} & \defeq\int\Zdet_{C}\Zdet_{C}^{\trans}\otimes(K\Sigma K^{\trans})^{-1}.\label{eq:SHpi}
\end{align}
\end{thm}
Up to a term that depends only on $f$, the likelihood ratio processes
$\{\like_{n}(\vall,f)-\like_{n}(0,0)\}$ and $\{\like_{n}^{\pr}(\vall)-\like_{n}^{\pr}(0)\}$
thus share the same distributional limit; since this limit obtains
under all local-to-unity sequences permitted by \assref{LOC} and
\assref{PR}, it is evident that these models converge to the same
limiting experiment, in the sense of \citet[Def.~9.1]{VDV98}. This
substantiates the asymptotic equivalence between the problem of inference
on $(A,\Lambda_{\lu})$ in the VAR \eqref{dgp}--\eqref{dgp-err},
and of inference on $(A_{\pr},\Lambda_{\pr})$ in a predictive regression,
in the vicinity of unit roots. That the weak limit in \eqref{lim-exp}
is also shared by $\{\likens_{n}(\vall)-\like_{n}(0)\}$ is of practical
importance for the likelihood-based tests discussed below.

\subsection{Likelihood-based inference}

\subsubsection{\label{subsec:ML-estimators}ML estimators}

We next turn to the ML estimators of $A$ and $\Lambda_{\lu}$. Let
the unrestricted and restricted estimators of $\PHI$, the latter
with $\Lambda_{\lu}(\PHI)=\Lambda_{0}\in\reals^{q\times q}$ imposed,
be denoted as
\begin{align*}
\hat{\PHI}_{n}\defeq & \argmax_{\PHI\in\reals^{p\times kp}}\likens_{n}(\PHI) & \hat{\PHI}_{n\mid\Lambda_{0}} & \defeq\argmax_{\{\PHI\in\set P\mid\Lambda_{\lu}(\PHI)=\Lambda_{0}\}}\likens_{n}(\PHI).
\end{align*}
(For details of how to compute $\hat{\PHI}_{n\mid\Lambda_{0}}$ in
practice, see \appref{numerical}.) Then $\hat{A}_{n}\defeq A(\hat{\PHI}_{n})$
and $\hat{\Lambda}_{n,\lu}\defeq\Lambda_{\lu}(\hat{\PHI}_{n})$ are
the associated unrestricted MLEs of $A$ and $\Lambda_{\lu}$, and
$\hat{A}_{n\mid\Lambda_{0}}\defeq A(\hat{\PHI}_{n\mid\Lambda_{0}})$
the MLE of $A$ under $\Lambda_{\lu}(\PHI)=\Lambda_{0}$.

To state our main result on these estimators, let $L_{\lu,\perp}$
be any $p\times r$ matrix spanning $(\spn L_{\lu})^{\perp}$; one
possible choice is $\alpha\defeq\lim_{n\goesto\infty}\Phi_{n}(1)\beta(\beta^{\trans}\beta)^{-1}$.
Recall that a random vector $\eta$ is \emph{mixed normal} with mean
zero and conditional variance $V$, denoted $\eta\sim\mn[0,V]$, if
$\expect\e^{\i\tau^{\trans}\eta}=\expect\e^{-\frac{1}{2}\tau^{\trans}V\tau}$.
\begin{thm}
\label{thm:estimators}Suppose \assref{LOC} holds. Then
\begin{enumerate}
\item \label{enu:estimators:unres}$\hat{A}_{n}\defeq A(\hat{\PHI}_{n})$
and $\hat{\Lambda}_{n,\lu}\defeq\Lambda_{\lu}(\hat{\PHI}_{n})$ satisfy\footnote{Recall $\Zdet_{C}(r)=Z_{C}(r)-\mu_{0}-\mu_{1}r$, for $\mu_{0}\defeq\int_{0}^{1}(4-6s)B(s)\diff s$
and $\mu_{1}=\int_{0}^{1}(-6+12s)B(s)\diff s$ (see e.g.\ \citealp{Ell98Ecta},
p.\ 151). Since $\Zdet_{C}$ is not adapted, an expression such as
$\int\Zdet_{C}(\deriv E)^{\trans}$ should be understood as a convenient
shorthand for $\int Z_{C}(\deriv E)^{\trans}-\mu_{0}\int(\deriv E)^{\trans}-\mu_{1}\int r(\deriv E)^{\trans}$.}
\[
n\begin{bmatrix}\hat{A}_{n}-A\\
\hat{\Lambda}_{n,\lu}-\Lambda_{n,\lu}
\end{bmatrix}\wkc K\int(\deriv E)\Zdet_{C}^{\trans}\left(\int\Zdet_{C}\Zdet_{C}^{\trans}\right)^{-1}
\]
\item $n\vek(\hat{A}_{n\mid\Lambda_{n,\lu}}-A)\wkc\mn[0,V_{zz}\otimes V_{\err\err}]$,
where
\begin{align}
V_{zz}\otimes V_{\err\err} & \defeq\left(\int\Zdet_{C}\Zdet_{C}^{\trans}\right)^{-1}\otimes\mathcal{J}L_{\lu,\perp}(L_{\lu,\perp}^{\trans}\Sigma^{-1}L_{\lu,\perp})^{-1}L_{\lu,\perp}^{\trans}\mathcal{J}^{\trans}\label{eq:Arstr}\\
 & =\left(\int\Zdet_{C}\Zdet_{C}^{\trans}\right)^{-1}\otimes(\alpha^{\trans}\Sigma^{-1}\alpha)^{-1}.\label{eq:johvar}
\end{align}
\end{enumerate}
\end{thm}
The limiting distribution of the unrestricted ML estimator of $A$
thus depends on $C$, which cannot be consistently estimated. However,
if the \emph{correct} value of $\Lambda_{\lu}$ is imposed, then the
restricted ML estimator $\hat{A}_{n\mid\Lambda_{n,\lu}}$ is asymptotically
mixed normal. In the special case where $\Lambda_{\lu}=I_{q}$, this
exactly replicates the mixed normality of the ML estimates of the
cointegrating relations, when the correct cointegrating rank is imposed
(see e.g.\ \citealp{Joh95}, Thm.~13.3). In this manner, the preceding
theorem generalises that mixed normality beyond the setting of a VAR
with exact unit roots. Though we shall not give the proof here, it
may be shown that with the correct $\Lambda_{\lu}$ imposed, the model
loglikelihood is locally asymptotically mixed normal (LAMN), so that
$\hat{A}_{n\mid\Lambda_{n,\lu}}$ also inherits the large-sample efficiency
properties familiar from the case of exact unit roots (\citealp{Phi91Ecta}).

\subsubsection{Likelihood ratio tests}

Though part~(i) of the preceding provides a basis for inference on
$A$ using Wald-type statistics, there are some difficulties with
this approach in practice, because there is no guarantee that the
characteristic roots of the unrestrictedly estimated VAR will separate
in the manner prescribed by \assref{QC}. Since these roots come in
conjugate pairs, it may well be the case that when ordered by their
complex modulus (or proximity to real unity), the $q$th and $(q+1)$th
roots will be complex conjugates, preventing us from isolating the
first $q$ roots from the rest -- a problem exacerbated by typically
imprecise estimation of these roots \citep{OU12ET}. A superior approach
therefore utilises (quasi-) likelihood ratio (LR) tests to perform
inference on both $\Lambda_{\lu}$ and $A$; specifically the statistics\begin{subequations}\label{eq:lrstats}
\begin{align}
\lr_{n}(\Lambda_{0}) & \defeq2\left[\max_{\{\PHI\in\set P\mid\Lambda_{\lu}(\PHI)\in\set L\}}\likens_{n}(\PHI)-\max_{\{\Phi\in\set P\mid\Lambda_{\lu}(\PHI)=\Lambda_{0}\}}\likens_{n}(\PHI)\right]\label{eq:LRroot}\\
\lr_{n}(a_{0};\Lambda_{0}) & \defeq2\left[\max_{\{\PHI\in\set P\mid\Lambda_{\lu}(\PHI)=\Lambda_{0}\}}\likens_{n}(\PHI)-\max_{\{\Phi\in\set P\mid\Lambda_{\lu}(\PHI)=\Lambda_{0},a_{ij}(\PHI)=a_{0}\}}\likens_{n}(\PHI)\right],
\end{align}
\end{subequations}where $\set L$ is the parameter space for $\Lambda_{\lu}$
(see \subsecref{LambdaLUparspc} above). $\lr_{n}(\Lambda_{0})$ is
the usual likelihood ratio test for $H_{0}:\Lambda_{\lu}(\PHI)=\Lambda_{0}$,
while $\lr_{n}(a_{0};\Lambda_{0})$ corresponds to the likelihood
ratio test of $H_{0}:a_{ij}(\PHI)=a_{0}$, when $\Lambda_{\lu}(\PHI)=\Lambda_{0}$
is maintained under both the null and the alternative. Our next result
provides the asymptotic distributions of these test statistics; for
given $C\in\reals^{q\times q}$, let
\[
C_{\ast}\defeq(L_{\lu}^{\trans}\Sigma L_{\lu})^{-1/2}C(L_{\lu}^{\trans}\Sigma L_{\lu})^{1/2}.
\]

\begin{thm}
\label{thm:lrstats}Suppose \assref{LOC} holds. Then
\begin{gather}
\lr_{n}(\Lambda_{n,\lu})\wkc\tr\left\{ \int(\deriv W_{\ast})\Zdet_{C_{\ast}}^{\trans}\left(\int\Zdet_{C_{\ast}}\Zdet_{C_{\ast}}^{\trans}\right)^{-1}\int\Zdet_{C_{\ast}}(\deriv W_{\ast})^{\trans}\right\} \label{eq:multDF}
\end{gather}
where $W_{\ast}\sim\BM(I_{q})$, $\Zdet_{C_{\ast}}$ is the residual
from an $L^{2}[0,1]$ projection of the sample paths of $Z_{C_{\ast}}(r)\defeq\int_{0}^{r}\e^{C_{\ast}(r-s)}\diff W_{\ast}(s)$
onto a constant and linear trend; and
\begin{equation}
\lr_{n}[a_{ij}(\PHI_{n});\Lambda_{n,\lu}]\wkc\chi_{1}^{2}.\label{eq:chisqlim}
\end{equation}
\end{thm}

\subsubsection{Nearly optimal tests}

\citet{EMW15Ecta} consider hypothesis tests that are affected by
nuisance parameters, in settings where the limiting experiment is
not a Gaussian shift experiment (with an unrestricted parameter space),
such the usual asymptotic optimality enjoyed by ML-based inference
does not hold. In view of \thmref{emw}, our problem falls within
their framework: the correspondence can be most easily seen when,
analogously to their equation (1), we seek to test hypotheses of the
form
\[
H_{0}:A=A_{0}\sep\Lambda_{\lu}\in\set L\qquad\text{against}\qquad H_{1}:A\neq A_{0}\sep\Lambda_{\lu}\in\set L,
\]
so that $A$ is the parameter of interest, and $\Lambda_{\lu}$ the
nuisance parameter.\footnote{More precisely, this correspondence emerges asymptotically, under
the local parametrisation \eqref{reparm-thm}, with $\vall$ playing
the role of $(A,\Lambda_{\lu})$. Though the other components of the
VAR, represented by $f$, are technically also nuisance parameters,
the separability of the limiting experiment in $\vall$ and $f$ entails
that these asymptotically play no role in the testing problem. Concentrating
these parameters out of the likelihood thus leads to test statistics
with the same limiting distribution as if these other components were
known \emph{a priori}.} Their tests have the following Neyman--Pearson form,
\begin{equation}
\np_{n}(A_{0})\defeq\indic\left\{ \int_{\reals^{r\times q}\times\set L}\e^{\likens_{n}(A,\Lambda)}F_{1}(\deriv A,\deriv\Lambda)>\crit_{\alpha}\int_{\set L}\e^{\likens_{n}(A_{0},\Lambda)}F_{0}(A_{0},\deriv\Lambda)\right\} \label{eq:NP}
\end{equation}
where $\likens_{n}(A,\Lambda)\defeq\max_{\{\Phi\in\set P\mid\Lambda_{\lu}(\PHI)=\Lambda,A(\PHI)=A\}}\likens_{n}(\PHI)$,
and $F_{0}$ and $F_{1}$ are distributions that respectively concentrate
on those subsets of the parameter space for $(A,\Lambda)$ consistent
with the null and the alternative.

As discussed by \citet{EMW15Ecta}, the level $\alpha$ test that
maximises weighted average power (WAP; against the $F_{1}$-weighted
alternative) results when $F_{0}$ and $\crit_{\alpha}$ are such
that
\begin{align}
\int_{\set L}\Prob_{(A_{0},\Lambda)}\{\np_{n}(A_{0})=1\}F_{0}(\deriv\Lambda) & =\alpha, & \Prob_{(A_{0},\Lambda)}\{\np_{n}(A_{0})=1\} & \leq\alpha\sep\forall\Lambda\in\set L;\label{eq:LFD}
\end{align}
in which case $F_{0}$ is the least favourable distribution (LFD)
for the testing problem (for the given $F_{1}$-weighed alternative).
While exact calculation of the LFD is infeasible, the authors provide
an algorithm that delivers an $F_{0}$ (and associated $\crit_{\alpha}$)
that approximates the properties of the LFD, in the sense of their
Definition~1; the associated test is thus `nearly optimal' in the
WAP sense against $F_{1}$.

Their methodology may be applied to the present setting, with \thmref{emw}
indicating the manner in which the limiting experiment depends on
the model parameters $(\PHI,\Sigma)$. In particular, \thmref{emw}
justifies simulating from an appropriately parametrised VAR(1) model
to approximate the distribution of the likelihood ratio process, and
particularly the probabilities in \eqref{LFD}, which are the key
input for the determination of $F_{0}$ and $\crit_{\alpha}$ (see
\appref{numerical} for details).\footnote{Unlike \citet{EMW15Ecta}, we have phrased the testing problem in
terms of the `original' model parameters, rather than the local
parameters that appear in the limiting experiment. Since the approximate
LFD and the weighted alternative are defined in terms of the local
parameters, $F_{0}$ and $F_{1}$ in \eqref{NP} should be indexed
by $n$ so as to correspond to (sample-size independent) distributions
on the local parameter space. In our implementation, where $F_{0}$
is obtained by simulating from the finite-sample distribution of a
suitably-chosen model, such details are handled implicitly. That this
construction yields a test that is asymptotically of size $\alpha$
follows from \thmref{emw} and the continuous mapping theorem (with
the finiteness of the supports of $F_{0}$ and $F_{1}$ implying that
the finite-dimensional convergence obtained in that result is sufficient
here): a rigorous development is not given here, to keep the paper
to a manageable length.} While the application to our problem is in many respects similar
to their application to a predictive regression model (\citealp{EMW15Ecta},
Sec.\ 5.3), with both problems sharing a common limiting experiment,
an important difference arises in that we impose a lower bound $\radius$
for the $q$ largest roots, for reasons discussed in \subsecref{qcs}
above. Accordingly, we do not develop a `switching test' of the
form that would be needed to accommodate a parametrisation of the
model lying deeper in the stationary region.

More generally, the same approach may be taken to testing hypotheses
that restrict only a subset of the elements of $A$ under the null,
with the remaining elements being concentrated out. While the optimality
properties of the resultant test are less clear in that case, the
procedure still yields a test of asymptotic level $\alpha$. For example,
a test of $H_{0}:a_{ij}(\PHI)=a_{0},\Lambda_{\lu}\in\set L$, for
a given $(i,j)$, could be constructed as
\begin{equation}
\np_{n}(a_{0})\defeq\indic\left\{ \int_{\reals\times\set L}\e^{\likens_{n}(a,\Lambda)}F_{1}(\deriv a,\deriv\Lambda)>\crit_{\alpha}\int_{\set L}\e^{\likens_{n}(a_{0},\Lambda)}F_{0}(a_{0},\deriv\Lambda)\right\} \label{eq:NPa}
\end{equation}
where now $\likens_{n}(a,\Lambda)\defeq\max_{\{\Phi\in\set P\mid\Lambda_{\lu}(\PHI)=\Lambda,a_{ij}(\PHI)=a\}}\likens_{n}(\PHI)$.
$F_{0}$ and $F_{1}$ concentrate respectively on subsets of the parameter
space for $a_{ij}$ and $\Lambda_{\lu}$ consistent with the null
and the alternative, with the remaining elements of $A$ replaced
(in a finite sample) by consistent estimators.

\subsection{Confidence intervals}

\label{subsec:confidence-Intervals}

Suppose now that interest centres on a given element $a_{ij}$ of
$A$. The preceding suggests a number of possible ways for constructing
asymptotically valid $1-\alpha$ confidence intervals for $a_{ij}$.
Firstly, if $\Lambda_{\lu}$ were known to be some $\Lambda_{0}\in\set L$,
then by \thmref{lrstats} a confidence interval based on the efficient
(likelihood ratio) test may be constructed conditionally on that $\Lambda_{0}$,
as
\begin{equation}
\ci_{a_{ij}\mid\Lambda_{0}}(\alpha)\defeq\{a_{0}\in\reals\mid\lr_{n}(a_{0};\Lambda_{0})\leq\chi_{1,1-\alpha}^{2}\}.\label{eq:condci}
\end{equation}
$\ci_{a_{ij}\mid\Lambda_{0}}$ may also be of interest in cases where
$\Lambda_{\lu}$ is not plausibly known a priori, insofar as a plot
of these intervals illustrates the potential sensitivity (or robustness)
of inferences on $a_{ij}$ to departures from the assumption of exact
unit roots. (This is particularly feasible when $q=1$: see \secref{empirical}
for an illustration).

In the more realistic case that $\Lambda_{\lu}$ is unknown, a well-established
approach to inference is based on Bonferroni's inequality, which
involves constructing a first-stage confidence interval for $\Lambda_{\lu}$
on the basis of the LR test in \thmref{lrstats}, as per
\[
\ci_{\Lambda}(\alpha_{1})\defeq\{\Lambda_{0}\in\set L\mid\lr_{n}(\Lambda_{0})\leq c_{1-\alpha_{1}}[n(\Lambda_{0}-I_{q})]\}
\]
where $c_{\tau}(C)$ denotes the $\tau$th quantile of the distribution
of \eqref{multDF}, under the local parameter $C$. By construction,
for any $\alpha_{1}+\alpha_{2}\leq\alpha$, the set
\[
\cibonf(\alpha_{1},\alpha_{2})\defeq\Union_{\Lambda_{0}\in\ci_{\Lambda}(\alpha_{1})}\ci_{a_{ij}\mid\Lambda_{0}}(\alpha_{2}),
\]
has asymptotic level $\alpha$. Since this yields inferences on $a_{ij}$
that are necessarily conservative, refinements along lines proposed
by \citet{CES95ET} and \citet{CY06JFE} in the context of predictive
regression (an approach that has since been further extended by \citealp{McC17JoE})
may also be considered here.

However, we have found in practice that the likelihood ratio test
of $\Lambda_{\lu}$ generally lacks power, with the associated confidence
set $\ci_{\Lambda}(\alpha_{1})$ being unsuitably wide. In general,
much tighter intervals can be found on the basis of the nearly optimal
test of \citet{EMW15Ecta},
\[
\cinp(\alpha)\defeq\{a_{0}\in\reals\mid\np_{n}(a_{0})=0\}.
\]
where $\np_{n}(a_{0})$ is defined in \eqref{NPa} above.

\subsection{Deterministic terms}

\label{subsec:deterministics}

For the cointegrated VAR with exact unit roots, \citet[Sec.~5.7]{Joh95}
develops a hierarchy of models (in his notation, $H_{2}\subset H_{1}^{\ast}\subset H_{1}\subset H^{\ast}\subset H$)
ordered according to their treatment of the deterministic terms in
form \eqref{redform} of the model. In our more general setting where
$\Lambda_{\lu}=I_{q}$ is not required, these models take on an altered
expression, and not all are realisable through restrictions on the
model parameters. To discuss how we treat might handle deterministic
terms in our setting, and their implications for inference, we first
recall that the mapping from the DGP \eqref{dgp} to form \eqref{redform}
of the VAR implies
\begin{align*}
m & =\Phi(1)\mu+\Psi\delta, & d & =\Phi(1)\delta,
\end{align*}
where $\Psi\defeq\sum_{i=1}^{k}i\Phi_{i}$. Three important cases
are the following:\footnote{There is a fourth case, which sits in between the first two, in which
a linear trend is present in $y_{t}$ but is assumed to be eliminated
by the quasi-cointegrating relationships, whence $\beta^{\trans}\delta=0$.
Since $\beta\in(\spn R_{\lu})^{\perp}$, this is equivalent to requiring
$d\in\spn\Phi(1)R_{\lu}$. If we assume exact unit roots, then $\Phi(1)R_{\lu}=0$
(from \eqref{eig-eig} above) and this restriction can be imposed
simply by estimating the VAR \eqref{redform} without a trend (as
in Johansen's model $H_{1}$). However, in our setting with non-unit
roots this restriction cannot be so simply expressed, because $\Phi(1)$
may have full rank; all that can be said is that $d\in\spn\Phi(1)R_{\lu}$.
Estimation under this restriction is accordingly more involved, and
we leave the development of the asymptotics of our procedure in this
case for future work.}
\begin{enumerate}
\item Both $\mu$ and $\delta$ are unrestricted. The reduced form VAR \eqref{redform}
should be estimated with $(m,d)$ unrestricted (as per Johansen's
model $H$). Our asymptotics assume that the DGP is the structural
VAR \eqref{dgp}, so that $d=\Phi(1)\delta$ holds even though this
is not imposed in estimation. Indeed, it would not be possible to
impose the restriction $d\in\spn\Phi(1)$ (as per Johansen's $H^{\ast}$)
in the present setting, because whenever the largest roots of $\Phi$
are not exactly unity, $\Phi(1)$ has full rank. Thus $d=\Phi(1)\delta$
would be effectively unrestricted, and a model with exact unit roots
and $d\notin\Phi(1)$ would lie in the closure of the parameter space.
\item $\delta=0$, but $\mu$ is unrestricted. The VAR \eqref{redform}
should be estimated with only a constant (as in Johansen's model
$H_{1}$). Under the assumption that the DGP is the structural VAR
with $\delta=0$, $y_{t}$ has no drift. The limit theorems given
in Theorems~\ref{thm:emw}--\ref{thm:lrstats} must be amended in
this case, by replacing each instance of $\Zdet_{C}$ with the demeaned
diffusion process $Z_{C}(r)-\int_{0}^{1}Z_{C}(s)\diff s$. (Imposing
the restriction that $m\in\spn\Phi(1)$, as per Johansen's model
$H_{1}^{\ast}$, is impossible in our setting.)
\item $\mu=\delta=0$. The VAR \eqref{redform} should be estimated with
$m=d=0$ (as per Johansen's model $H_{2}$); in Theorems~\ref{thm:emw}--\ref{thm:lrstats},
$\Zdet_{C}$ is replaced by $Z_{C}$.
\end{enumerate}

In light of this, our recommendation is to estimate the model with
an unrestricted intercept and trend if there is a discernable linear
trend in the data, and to otherwise estimate the model with only an
intercept.

\section{Finite-sample performance}

\label{sec:finite-sample}

\subsection{Simulations}

\label{sec:simulations}

We conducted simulations to evaluate the finite-sample performance
of $\np_{n}(a)$, in terms of size and power. For this exercise, natural
comparisons are with the likelihood ratio test $\lr_{n}(a;I_{q})$
corresponding to the efficient rank-imposed MLE (\citealp{Joh95}),
and with the low-frequency stationarity test $\lfst_{n}(a)$ of \citet{MW13JoE},
constructed with $b=10/r^{1/2}$ in their equation (25), as per their
recommendations. 

The DGP in the simulation design is a bivariate ($p=2$) VAR(2) with
one root near unity ($q=1$), is parametrised in terms of the underlying
$(R,\Lambda)$ matrices (see \appref{qcs-results}) as
\begin{align*}
R_{\lu} & =\begin{bmatrix}1\\
1
\end{bmatrix} & R_{\st} & =\begin{bmatrix}0 & 1 & 1\\
1 & 0 & 2
\end{bmatrix} & \lambda_{\lu} & =\lambda_{\lu,0} & \Lambda_{\st} & =\diag\{0.5,0.4,0.3\}
\end{align*}
from which the implied VAR coefficients may be recovered (via \lemref{GLR}\enuref{GLR:jordanF}).
The implied quasi-cointegrating vector is $\beta_{0}=[1,-a_{0}]^{\trans}$
with $a_{0}=1$. Across the simulation designs, we vary:
\begin{enumerate}
\item the largest root $\lambda_{\lu,0}$ over $\{0.96,0.98,1.00\}$; and
\item the matrix $\Sigma$ such that $\Omega=K\Sigma K^{\trans}=[\begin{smallmatrix}1 & \omega_{\sl}\\
\omega_{\sl} & 1
\end{smallmatrix}]$, for $\omega_{\sl}\in\{-0.9,-0.8,\ldots,0.9\}$;
\end{enumerate}
for each of which we generate $20,000$ samples of length $n=200$.
In implementing $\np_{n}(a)$, we set a lower bound of $\radius=0.9$
on $\set L=[\radius,1]$, the parameter space for $\lambda_{\lu}$.
When estimating the VAR, we include a constant, and select the lag
length by the Akaike information criterion (AIC; with a maximum lag
order of 10). The nominal level of all tests is $\alpha=0.05$.

\begin{figure}
\begin{adjustwidth}{-2cm}{-2cm}\centering

\begin{tabular}{ccc}
\begin{tikzpicture}[x=1pt,y=1pt]
\definecolor{fillColor}{RGB}{255,255,255}
\path[use as bounding box,fill=fillColor,fill opacity=0.00] (0,0) rectangle (202.36,202.36);
\begin{scope}
\path[clip] (  0.00,  0.00) rectangle (202.36,202.36);
\definecolor{drawColor}{RGB}{255,255,255}
\definecolor{fillColor}{RGB}{255,255,255}

\path[draw=drawColor,line width= 0.6pt,line join=round,line cap=round,fill=fillColor] (  0.00,  0.00) rectangle (202.36,202.36);
\end{scope}
\begin{scope}
\path[clip] ( 42.06, 33.48) rectangle (196.36,196.36);
\definecolor{fillColor}{RGB}{255,255,255}

\path[fill=fillColor] ( 42.06, 33.48) rectangle (196.36,196.36);
\definecolor{drawColor}{RGB}{248,118,109}

\path[draw=drawColor,line width= 0.6pt,line join=round] ( 49.07, 40.93) --
	( 64.66, 40.89) --
	( 80.24, 40.89) --
	( 95.83, 40.92) --
	(111.42, 40.88) --
	(127.00, 40.89) --
	(142.59, 40.89) --
	(158.17, 40.96) --
	(173.76, 40.93) --
	(189.34, 40.93);
\definecolor{drawColor}{RGB}{0,186,56}

\path[draw=drawColor,line width= 0.6pt,dash pattern=on 1pt off 3pt ,line join=round] ( 49.07,188.95) --
	( 64.66,134.54) --
	( 80.24, 84.11) --
	( 95.83, 58.10) --
	(111.42, 46.03) --
	(127.00, 42.64) --
	(142.59, 44.26) --
	(158.17, 52.36) --
	(173.76, 75.04) --
	(189.34,147.39);
\definecolor{drawColor}{RGB}{97,156,255}

\path[draw=drawColor,line width= 0.6pt,dash pattern=on 7pt off 3pt ,line join=round] ( 49.07, 41.50) --
	( 64.66, 42.02) --
	( 80.24, 41.79) --
	( 95.83, 41.61) --
	(111.42, 41.79) --
	(127.00, 42.09) --
	(142.59, 41.77) --
	(158.17, 41.97) --
	(173.76, 41.69) --
	(189.34, 42.11);
\definecolor{drawColor}{gray}{0.30}

\path[draw=drawColor,line width= 0.6pt,line join=round,line cap=round] ( 42.06, 33.48) rectangle (196.36,196.36);
\end{scope}
\begin{scope}
\path[clip] (  0.00,  0.00) rectangle (202.36,202.36);
\definecolor{drawColor}{gray}{0.30}

\node[text=drawColor,anchor=base east,inner sep=0pt, outer sep=0pt, scale=  0.96] at ( 36.66, 30.50) {0.00};

\node[text=drawColor,anchor=base east,inner sep=0pt, outer sep=0pt, scale=  0.96] at ( 36.66, 69.74) {0.25};

\node[text=drawColor,anchor=base east,inner sep=0pt, outer sep=0pt, scale=  0.96] at ( 36.66,108.97) {0.50};

\node[text=drawColor,anchor=base east,inner sep=0pt, outer sep=0pt, scale=  0.96] at ( 36.66,148.20) {0.75};

\node[text=drawColor,anchor=base east,inner sep=0pt, outer sep=0pt, scale=  0.96] at ( 36.66,187.44) {1.00};
\end{scope}
\begin{scope}
\path[clip] (  0.00,  0.00) rectangle (202.36,202.36);
\definecolor{drawColor}{gray}{0.30}

\path[draw=drawColor,line width= 0.6pt,line join=round] ( 39.06, 33.81) --
	( 42.06, 33.81);

\path[draw=drawColor,line width= 0.6pt,line join=round] ( 39.06, 73.04) --
	( 42.06, 73.04);

\path[draw=drawColor,line width= 0.6pt,line join=round] ( 39.06,112.28) --
	( 42.06,112.28);

\path[draw=drawColor,line width= 0.6pt,line join=round] ( 39.06,151.51) --
	( 42.06,151.51);

\path[draw=drawColor,line width= 0.6pt,line join=round] ( 39.06,190.74) --
	( 42.06,190.74);
\end{scope}
\begin{scope}
\path[clip] (  0.00,  0.00) rectangle (202.36,202.36);
\definecolor{drawColor}{gray}{0.30}

\path[draw=drawColor,line width= 0.6pt,line join=round] ( 80.24, 30.48) --
	( 80.24, 33.48);

\path[draw=drawColor,line width= 0.6pt,line join=round] (119.21, 30.48) --
	(119.21, 33.48);

\path[draw=drawColor,line width= 0.6pt,line join=round] (158.17, 30.48) --
	(158.17, 33.48);
\end{scope}
\begin{scope}
\path[clip] (  0.00,  0.00) rectangle (202.36,202.36);
\definecolor{drawColor}{gray}{0.30}

\node[text=drawColor,anchor=base,inner sep=0pt, outer sep=0pt, scale=  0.96] at ( 80.24, 21.46) {-0.5};

\node[text=drawColor,anchor=base,inner sep=0pt, outer sep=0pt, scale=  0.96] at (119.21, 21.46) {0.0};

\node[text=drawColor,anchor=base,inner sep=0pt, outer sep=0pt, scale=  0.96] at (158.17, 21.46) {0.5};
\end{scope}
\begin{scope}
\path[clip] (  0.00,  0.00) rectangle (202.36,202.36);
\definecolor{drawColor}{RGB}{0,0,0}

\node[text=drawColor,anchor=base,inner sep=0pt, outer sep=0pt, scale=  1.20] at (119.21,  8.33) {$\omega_{\sl}$};
\end{scope}
\end{tikzpicture} &  & 
\begin{tikzpicture}[x=1pt,y=1pt]
\definecolor{fillColor}{RGB}{255,255,255}
\path[use as bounding box,fill=fillColor,fill opacity=0.00] (0,0) rectangle (202.36,202.36);
\begin{scope}
\path[clip] (  0.00,  0.00) rectangle (202.36,202.36);
\definecolor{drawColor}{RGB}{255,255,255}
\definecolor{fillColor}{RGB}{255,255,255}

\path[draw=drawColor,line width= 0.6pt,line join=round,line cap=round,fill=fillColor] (  0.00,  0.00) rectangle (202.36,202.36);
\end{scope}
\begin{scope}
\path[clip] ( 42.06, 33.48) rectangle (196.36,196.36);
\definecolor{fillColor}{RGB}{255,255,255}

\path[fill=fillColor] ( 42.06, 33.48) rectangle (196.36,196.36);
\definecolor{drawColor}{RGB}{248,118,109}

\path[draw=drawColor,line width= 0.6pt,line join=round] ( 70.65,181.39) --
	( 75.94,175.08) --
	( 81.35,165.04) --
	( 86.86,150.25) --
	( 92.48,128.65) --
	( 98.22,103.06) --
	(104.08, 76.90) --
	(110.06, 56.39) --
	(116.16, 44.33) --
	(122.38, 40.88) --
	(128.73, 49.90) --
	(135.21, 81.22) --
	(141.82,127.90) --
	(148.57,164.08) --
	(155.45,181.53) --
	(162.48,187.12) --
	(169.65,188.59) --
	(176.96,188.82) --
	(184.43,188.92);
\definecolor{drawColor}{RGB}{0,186,56}

\path[draw=drawColor,line width= 0.6pt,dash pattern=on 1pt off 3pt ,line join=round] ( 70.65,188.73) --
	( 75.94,188.32) --
	( 81.35,187.31) --
	( 86.86,184.35) --
	( 92.48,176.64) --
	( 98.22,160.64) --
	(104.08,131.23) --
	(110.06, 93.32) --
	(116.16, 60.92) --
	(122.38, 46.28) --
	(128.73, 52.39) --
	(135.21, 79.96) --
	(141.82,122.98) --
	(148.57,159.76) --
	(155.45,179.26) --
	(162.48,186.38) --
	(169.65,188.28) --
	(176.96,188.82) --
	(184.43,188.95);
\definecolor{drawColor}{RGB}{97,156,255}

\path[draw=drawColor,line width= 0.6pt,dash pattern=on 7pt off 3pt ,line join=round] ( 70.65, 59.08) --
	( 75.94, 57.85) --
	( 81.35, 56.41) --
	( 86.86, 54.71) --
	( 92.48, 52.60) --
	( 98.22, 50.43) --
	(104.08, 48.20) --
	(110.06, 46.45) --
	(116.16, 45.00) --
	(122.38, 44.38) --
	(128.73, 44.52) --
	(135.21, 45.81) --
	(141.82, 47.47) --
	(148.57, 49.95) --
	(155.45, 52.41) --
	(162.48, 54.98) --
	(169.65, 56.87) --
	(176.96, 58.54) --
	(184.43, 60.03);
\definecolor{drawColor}{gray}{0.30}

\path[draw=drawColor,line width= 0.6pt,line join=round,line cap=round] ( 42.06, 33.48) rectangle (196.36,196.36);
\end{scope}
\begin{scope}
\path[clip] (  0.00,  0.00) rectangle (202.36,202.36);
\definecolor{drawColor}{gray}{0.30}

\node[text=drawColor,anchor=base east,inner sep=0pt, outer sep=0pt, scale=  0.96] at ( 36.66, 34.16) {0.00};

\node[text=drawColor,anchor=base east,inner sep=0pt, outer sep=0pt, scale=  0.96] at ( 36.66, 72.04) {0.25};

\node[text=drawColor,anchor=base east,inner sep=0pt, outer sep=0pt, scale=  0.96] at ( 36.66,109.91) {0.50};

\node[text=drawColor,anchor=base east,inner sep=0pt, outer sep=0pt, scale=  0.96] at ( 36.66,147.78) {0.75};

\node[text=drawColor,anchor=base east,inner sep=0pt, outer sep=0pt, scale=  0.96] at ( 36.66,185.65) {1.00};
\end{scope}
\begin{scope}
\path[clip] (  0.00,  0.00) rectangle (202.36,202.36);
\definecolor{drawColor}{gray}{0.30}

\path[draw=drawColor,line width= 0.6pt,line join=round] ( 39.06, 37.47) --
	( 42.06, 37.47);

\path[draw=drawColor,line width= 0.6pt,line join=round] ( 39.06, 75.34) --
	( 42.06, 75.34);

\path[draw=drawColor,line width= 0.6pt,line join=round] ( 39.06,113.22) --
	( 42.06,113.22);

\path[draw=drawColor,line width= 0.6pt,line join=round] ( 39.06,151.09) --
	( 42.06,151.09);

\path[draw=drawColor,line width= 0.6pt,line join=round] ( 39.06,188.96) --
	( 42.06,188.96);
\end{scope}
\begin{scope}
\path[clip] (  0.00,  0.00) rectangle (202.36,202.36);
\definecolor{drawColor}{gray}{0.30}

\path[draw=drawColor,line width= 0.6pt,line join=round] ( 55.17, 30.48) --
	( 55.17, 33.48);

\path[draw=drawColor,line width= 0.6pt,line join=round] ( 85.67, 30.48) --
	( 85.67, 33.48);

\path[draw=drawColor,line width= 0.6pt,line join=round] (116.16, 30.48) --
	(116.16, 33.48);

\path[draw=drawColor,line width= 0.6pt,line join=round] (146.65, 30.48) --
	(146.65, 33.48);

\path[draw=drawColor,line width= 0.6pt,line join=round] (177.15, 30.48) --
	(177.15, 33.48);
\end{scope}
\begin{scope}
\path[clip] (  0.00,  0.00) rectangle (202.36,202.36);
\definecolor{drawColor}{gray}{0.30}

\node[text=drawColor,anchor=base,inner sep=0pt, outer sep=0pt, scale=  0.96] at ( 55.17, 21.46) {-0.10};

\node[text=drawColor,anchor=base,inner sep=0pt, outer sep=0pt, scale=  0.96] at ( 85.67, 21.46) {-0.05};

\node[text=drawColor,anchor=base,inner sep=0pt, outer sep=0pt, scale=  0.96] at (116.16, 21.46) {0.00};

\node[text=drawColor,anchor=base,inner sep=0pt, outer sep=0pt, scale=  0.96] at (146.65, 21.46) {0.05};

\node[text=drawColor,anchor=base,inner sep=0pt, outer sep=0pt, scale=  0.96] at (177.15, 21.46) {0.10};
\end{scope}
\begin{scope}
\path[clip] (  0.00,  0.00) rectangle (202.36,202.36);
\definecolor{drawColor}{RGB}{0,0,0}

\node[text=drawColor,anchor=base,inner sep=0pt, outer sep=0pt, scale=  1.20] at (119.21,  8.33) {$a-a_{0}$};
\end{scope}
\end{tikzpicture}\tabularnewline
(a) Size: $\lambda_{\lu}=0.96$. &  & (d) Power: $\lambda_{\lu}=0.96$ and $\omega_{\sl}=0.3$.\tabularnewline\addlinespace
\begin{tikzpicture}[x=1pt,y=1pt]
\definecolor{fillColor}{RGB}{255,255,255}
\path[use as bounding box,fill=fillColor,fill opacity=0.00] (0,0) rectangle (202.36,202.36);
\begin{scope}
\path[clip] (  0.00,  0.00) rectangle (202.36,202.36);
\definecolor{drawColor}{RGB}{255,255,255}
\definecolor{fillColor}{RGB}{255,255,255}

\path[draw=drawColor,line width= 0.6pt,line join=round,line cap=round,fill=fillColor] (  0.00,  0.00) rectangle (202.36,202.36);
\end{scope}
\begin{scope}
\path[clip] ( 42.06, 33.48) rectangle (196.36,196.36);
\definecolor{fillColor}{RGB}{255,255,255}

\path[fill=fillColor] ( 42.06, 33.48) rectangle (196.36,196.36);
\definecolor{drawColor}{RGB}{248,118,109}

\path[draw=drawColor,line width= 0.6pt,line join=round] ( 49.07, 40.93) --
	( 64.66, 40.88) --
	( 80.24, 40.95) --
	( 95.83, 40.88) --
	(111.42, 40.89) --
	(127.00, 40.89) --
	(142.59, 40.89) --
	(158.17, 40.94) --
	(173.76, 40.93) --
	(189.34, 40.88);
\definecolor{drawColor}{RGB}{0,186,56}

\path[draw=drawColor,line width= 0.6pt,dash pattern=on 1pt off 3pt ,line join=round] ( 49.07,188.95) --
	( 64.66,101.30) --
	( 80.24, 65.41) --
	( 95.83, 51.17) --
	(111.42, 44.61) --
	(127.00, 42.64) --
	(142.59, 43.74) --
	(158.17, 48.12) --
	(173.76, 59.86) --
	(189.34,112.22);
\definecolor{drawColor}{RGB}{97,156,255}

\path[draw=drawColor,line width= 0.6pt,dash pattern=on 7pt off 3pt ,line join=round] ( 49.07, 41.32) --
	( 64.66, 41.74) --
	( 80.24, 41.83) --
	( 95.83, 41.37) --
	(111.42, 41.64) --
	(127.00, 41.86) --
	(142.59, 41.91) --
	(158.17, 41.99) --
	(173.76, 41.41) --
	(189.34, 41.71);
\definecolor{drawColor}{gray}{0.30}

\path[draw=drawColor,line width= 0.6pt,line join=round,line cap=round] ( 42.06, 33.48) rectangle (196.36,196.36);
\end{scope}
\begin{scope}
\path[clip] (  0.00,  0.00) rectangle (202.36,202.36);
\definecolor{drawColor}{gray}{0.30}

\node[text=drawColor,anchor=base east,inner sep=0pt, outer sep=0pt, scale=  0.96] at ( 36.66, 74.81) {0.25};

\node[text=drawColor,anchor=base east,inner sep=0pt, outer sep=0pt, scale=  0.96] at ( 36.66,120.23) {0.50};

\node[text=drawColor,anchor=base east,inner sep=0pt, outer sep=0pt, scale=  0.96] at ( 36.66,165.65) {0.75};
\end{scope}
\begin{scope}
\path[clip] (  0.00,  0.00) rectangle (202.36,202.36);
\definecolor{drawColor}{gray}{0.30}

\path[draw=drawColor,line width= 0.6pt,line join=round] ( 39.06, 78.12) --
	( 42.06, 78.12);

\path[draw=drawColor,line width= 0.6pt,line join=round] ( 39.06,123.54) --
	( 42.06,123.54);

\path[draw=drawColor,line width= 0.6pt,line join=round] ( 39.06,168.96) --
	( 42.06,168.96);
\end{scope}
\begin{scope}
\path[clip] (  0.00,  0.00) rectangle (202.36,202.36);
\definecolor{drawColor}{gray}{0.30}

\path[draw=drawColor,line width= 0.6pt,line join=round] ( 80.24, 30.48) --
	( 80.24, 33.48);

\path[draw=drawColor,line width= 0.6pt,line join=round] (119.21, 30.48) --
	(119.21, 33.48);

\path[draw=drawColor,line width= 0.6pt,line join=round] (158.17, 30.48) --
	(158.17, 33.48);
\end{scope}
\begin{scope}
\path[clip] (  0.00,  0.00) rectangle (202.36,202.36);
\definecolor{drawColor}{gray}{0.30}

\node[text=drawColor,anchor=base,inner sep=0pt, outer sep=0pt, scale=  0.96] at ( 80.24, 21.46) {-0.5};

\node[text=drawColor,anchor=base,inner sep=0pt, outer sep=0pt, scale=  0.96] at (119.21, 21.46) {0.0};

\node[text=drawColor,anchor=base,inner sep=0pt, outer sep=0pt, scale=  0.96] at (158.17, 21.46) {0.5};
\end{scope}
\begin{scope}
\path[clip] (  0.00,  0.00) rectangle (202.36,202.36);
\definecolor{drawColor}{RGB}{0,0,0}

\node[text=drawColor,anchor=base,inner sep=0pt, outer sep=0pt, scale=  1.20] at (119.21,  8.33) {$\omega_{\sl}$};
\end{scope}
\end{tikzpicture} &  & 
\begin{tikzpicture}[x=1pt,y=1pt]
\definecolor{fillColor}{RGB}{255,255,255}
\path[use as bounding box,fill=fillColor,fill opacity=0.00] (0,0) rectangle (202.36,202.36);
\begin{scope}
\path[clip] (  0.00,  0.00) rectangle (202.36,202.36);
\definecolor{drawColor}{RGB}{255,255,255}
\definecolor{fillColor}{RGB}{255,255,255}

\path[draw=drawColor,line width= 0.6pt,line join=round,line cap=round,fill=fillColor] (  0.00,  0.00) rectangle (202.36,202.36);
\end{scope}
\begin{scope}
\path[clip] ( 42.06, 33.48) rectangle (196.36,196.36);
\definecolor{fillColor}{RGB}{255,255,255}

\path[fill=fillColor] ( 42.06, 33.48) rectangle (196.36,196.36);
\definecolor{drawColor}{RGB}{248,118,109}

\path[draw=drawColor,line width= 0.6pt,line join=round] ( 70.65,184.25) --
	( 75.94,179.83) --
	( 81.35,172.57) --
	( 86.86,160.88) --
	( 92.48,142.90) --
	( 98.22,117.24) --
	(104.08, 86.11) --
	(110.06, 56.72) --
	(116.16, 40.88) --
	(122.38, 46.99) --
	(128.73, 83.81) --
	(135.21,130.87) --
	(141.82,162.92) --
	(148.57,178.98) --
	(155.45,185.61) --
	(162.48,187.95) --
	(169.65,188.60) --
	(176.96,188.86) --
	(184.43,188.94);
\definecolor{drawColor}{RGB}{0,186,56}

\path[draw=drawColor,line width= 0.6pt,dash pattern=on 1pt off 3pt ,line join=round] ( 70.65,188.67) --
	( 75.94,188.34) --
	( 81.35,187.08) --
	( 86.86,183.35) --
	( 92.48,173.87) --
	( 98.22,152.32) --
	(104.08,113.28) --
	(110.06, 68.95) --
	(116.16, 44.06) --
	(122.38, 47.64) --
	(128.73, 82.19) --
	(135.21,133.00) --
	(141.82,167.13) --
	(148.57,182.45) --
	(155.45,187.49) --
	(162.48,188.60) --
	(169.65,188.88) --
	(176.96,188.94) --
	(184.43,188.95);
\definecolor{drawColor}{RGB}{97,156,255}

\path[draw=drawColor,line width= 0.6pt,dash pattern=on 7pt off 3pt ,line join=round] ( 70.65, 80.73) --
	( 75.94, 77.73) --
	( 81.35, 73.88) --
	( 86.86, 69.06) --
	( 92.48, 63.35) --
	( 98.22, 56.92) --
	(104.08, 50.21) --
	(110.06, 44.83) --
	(116.16, 41.52) --
	(122.38, 42.14) --
	(128.73, 46.14) --
	(135.21, 52.61) --
	(141.82, 60.34) --
	(148.57, 67.54) --
	(155.45, 73.77) --
	(162.48, 78.69) --
	(169.65, 82.68) --
	(176.96, 85.59) --
	(184.43, 87.80);
\definecolor{drawColor}{gray}{0.30}

\path[draw=drawColor,line width= 0.6pt,line join=round,line cap=round] ( 42.06, 33.48) rectangle (196.36,196.36);
\end{scope}
\begin{scope}
\path[clip] (  0.00,  0.00) rectangle (202.36,202.36);
\definecolor{drawColor}{gray}{0.30}

\node[text=drawColor,anchor=base east,inner sep=0pt, outer sep=0pt, scale=  0.96] at ( 36.66, 30.58) {0.00};

\node[text=drawColor,anchor=base east,inner sep=0pt, outer sep=0pt, scale=  0.96] at ( 36.66, 69.35) {0.25};

\node[text=drawColor,anchor=base east,inner sep=0pt, outer sep=0pt, scale=  0.96] at ( 36.66,108.12) {0.50};

\node[text=drawColor,anchor=base east,inner sep=0pt, outer sep=0pt, scale=  0.96] at ( 36.66,146.89) {0.75};

\node[text=drawColor,anchor=base east,inner sep=0pt, outer sep=0pt, scale=  0.96] at ( 36.66,185.65) {1.00};
\end{scope}
\begin{scope}
\path[clip] (  0.00,  0.00) rectangle (202.36,202.36);
\definecolor{drawColor}{gray}{0.30}

\path[draw=drawColor,line width= 0.6pt,line join=round] ( 39.06, 33.89) --
	( 42.06, 33.89);

\path[draw=drawColor,line width= 0.6pt,line join=round] ( 39.06, 72.65) --
	( 42.06, 72.65);

\path[draw=drawColor,line width= 0.6pt,line join=round] ( 39.06,111.42) --
	( 42.06,111.42);

\path[draw=drawColor,line width= 0.6pt,line join=round] ( 39.06,150.19) --
	( 42.06,150.19);

\path[draw=drawColor,line width= 0.6pt,line join=round] ( 39.06,188.96) --
	( 42.06,188.96);
\end{scope}
\begin{scope}
\path[clip] (  0.00,  0.00) rectangle (202.36,202.36);
\definecolor{drawColor}{gray}{0.30}

\path[draw=drawColor,line width= 0.6pt,line join=round] ( 55.17, 30.48) --
	( 55.17, 33.48);

\path[draw=drawColor,line width= 0.6pt,line join=round] ( 85.67, 30.48) --
	( 85.67, 33.48);

\path[draw=drawColor,line width= 0.6pt,line join=round] (116.16, 30.48) --
	(116.16, 33.48);

\path[draw=drawColor,line width= 0.6pt,line join=round] (146.65, 30.48) --
	(146.65, 33.48);

\path[draw=drawColor,line width= 0.6pt,line join=round] (177.15, 30.48) --
	(177.15, 33.48);
\end{scope}
\begin{scope}
\path[clip] (  0.00,  0.00) rectangle (202.36,202.36);
\definecolor{drawColor}{gray}{0.30}

\node[text=drawColor,anchor=base,inner sep=0pt, outer sep=0pt, scale=  0.96] at ( 55.17, 21.46) {-0.10};

\node[text=drawColor,anchor=base,inner sep=0pt, outer sep=0pt, scale=  0.96] at ( 85.67, 21.46) {-0.05};

\node[text=drawColor,anchor=base,inner sep=0pt, outer sep=0pt, scale=  0.96] at (116.16, 21.46) {0.00};

\node[text=drawColor,anchor=base,inner sep=0pt, outer sep=0pt, scale=  0.96] at (146.65, 21.46) {0.05};

\node[text=drawColor,anchor=base,inner sep=0pt, outer sep=0pt, scale=  0.96] at (177.15, 21.46) {0.10};
\end{scope}
\begin{scope}
\path[clip] (  0.00,  0.00) rectangle (202.36,202.36);
\definecolor{drawColor}{RGB}{0,0,0}

\node[text=drawColor,anchor=base,inner sep=0pt, outer sep=0pt, scale=  1.20] at (119.21,  8.33) {$a-a_{0}$};
\end{scope}
\end{tikzpicture}\tabularnewline
(b) Size: $\lambda_{\lu}=0.98$. &  & (e) Power: $\lambda_{\lu}=0.98$ and $\omega_{\sl}=-0.1$.\tabularnewline\addlinespace
\begin{tikzpicture}[x=1pt,y=1pt]
\definecolor{fillColor}{RGB}{255,255,255}
\path[use as bounding box,fill=fillColor,fill opacity=0.00] (0,0) rectangle (202.36,202.36);
\begin{scope}
\path[clip] (  0.00,  0.00) rectangle (202.36,202.36);
\definecolor{drawColor}{RGB}{255,255,255}
\definecolor{fillColor}{RGB}{255,255,255}

\path[draw=drawColor,line width= 0.6pt,line join=round,line cap=round,fill=fillColor] (  0.00,  0.00) rectangle (202.36,202.36);
\end{scope}
\begin{scope}
\path[clip] ( 46.86, 33.48) rectangle (196.36,196.36);
\definecolor{fillColor}{RGB}{255,255,255}

\path[fill=fillColor] ( 46.86, 33.48) rectangle (196.36,196.36);
\definecolor{drawColor}{RGB}{248,118,109}

\path[draw=drawColor,line width= 0.6pt,line join=round] ( 53.65,108.03) --
	( 68.75,107.66) --
	( 83.86,107.59) --
	( 98.96,107.73) --
	(114.06,108.47) --
	(129.16,108.03) --
	(144.26,107.66) --
	(159.36,108.10) --
	(174.46,107.88) --
	(189.56,108.10);
\definecolor{drawColor}{RGB}{0,186,56}

\path[draw=drawColor,line width= 0.6pt,dash pattern=on 1pt off 3pt ,line join=round] ( 53.65,118.03) --
	( 68.75,122.99) --
	( 83.86,120.32) --
	( 98.96,121.88) --
	(114.06,120.76) --
	(129.16,121.28) --
	(144.26,122.39) --
	(159.36,126.39) --
	(174.46,117.06) --
	(189.56,123.13);
\definecolor{drawColor}{RGB}{97,156,255}

\path[draw=drawColor,line width= 0.6pt,dash pattern=on 7pt off 3pt ,line join=round] ( 53.65,117.43) --
	( 68.75,117.88) --
	( 83.86,115.43) --
	( 98.96,114.18) --
	(114.06,112.62) --
	(129.16,114.62) --
	(144.26,118.25) --
	(159.36,116.32) --
	(174.46,110.40) --
	(189.56,113.73);
\definecolor{drawColor}{gray}{0.30}

\path[draw=drawColor,line width= 0.6pt,line join=round,line cap=round] ( 46.86, 33.48) rectangle (196.36,196.36);
\end{scope}
\begin{scope}
\path[clip] (  0.00,  0.00) rectangle (202.36,202.36);
\definecolor{drawColor}{gray}{0.30}

\node[text=drawColor,anchor=base east,inner sep=0pt, outer sep=0pt, scale=  0.96] at ( 41.46, 37.57) {0.000};

\node[text=drawColor,anchor=base east,inner sep=0pt, outer sep=0pt, scale=  0.96] at ( 41.46, 74.59) {0.025};

\node[text=drawColor,anchor=base east,inner sep=0pt, outer sep=0pt, scale=  0.96] at ( 41.46,111.61) {0.050};

\node[text=drawColor,anchor=base east,inner sep=0pt, outer sep=0pt, scale=  0.96] at ( 41.46,148.63) {0.075};

\node[text=drawColor,anchor=base east,inner sep=0pt, outer sep=0pt, scale=  0.96] at ( 41.46,185.65) {0.100};
\end{scope}
\begin{scope}
\path[clip] (  0.00,  0.00) rectangle (202.36,202.36);
\definecolor{drawColor}{gray}{0.30}

\path[draw=drawColor,line width= 0.6pt,line join=round] ( 43.86, 40.88) --
	( 46.86, 40.88);

\path[draw=drawColor,line width= 0.6pt,line join=round] ( 43.86, 77.90) --
	( 46.86, 77.90);

\path[draw=drawColor,line width= 0.6pt,line join=round] ( 43.86,114.92) --
	( 46.86,114.92);

\path[draw=drawColor,line width= 0.6pt,line join=round] ( 43.86,151.93) --
	( 46.86,151.93);

\path[draw=drawColor,line width= 0.6pt,line join=round] ( 43.86,188.95) --
	( 46.86,188.95);
\end{scope}
\begin{scope}
\path[clip] (  0.00,  0.00) rectangle (202.36,202.36);
\definecolor{drawColor}{gray}{0.30}

\path[draw=drawColor,line width= 0.6pt,line join=round] ( 83.86, 30.48) --
	( 83.86, 33.48);

\path[draw=drawColor,line width= 0.6pt,line join=round] (121.61, 30.48) --
	(121.61, 33.48);

\path[draw=drawColor,line width= 0.6pt,line join=round] (159.36, 30.48) --
	(159.36, 33.48);
\end{scope}
\begin{scope}
\path[clip] (  0.00,  0.00) rectangle (202.36,202.36);
\definecolor{drawColor}{gray}{0.30}

\node[text=drawColor,anchor=base,inner sep=0pt, outer sep=0pt, scale=  0.96] at ( 83.86, 21.46) {-0.5};

\node[text=drawColor,anchor=base,inner sep=0pt, outer sep=0pt, scale=  0.96] at (121.61, 21.46) {0.0};

\node[text=drawColor,anchor=base,inner sep=0pt, outer sep=0pt, scale=  0.96] at (159.36, 21.46) {0.5};
\end{scope}
\begin{scope}
\path[clip] (  0.00,  0.00) rectangle (202.36,202.36);
\definecolor{drawColor}{RGB}{0,0,0}

\node[text=drawColor,anchor=base,inner sep=0pt, outer sep=0pt, scale=  1.20] at (121.61,  8.33) {$\omega_{\sl}$};
\end{scope}
\end{tikzpicture} &  & 
\begin{tikzpicture}[x=1pt,y=1pt]
\definecolor{fillColor}{RGB}{255,255,255}
\path[use as bounding box,fill=fillColor,fill opacity=0.00] (0,0) rectangle (202.36,202.36);
\begin{scope}
\path[clip] (  0.00,  0.00) rectangle (202.36,202.36);
\definecolor{drawColor}{RGB}{255,255,255}
\definecolor{fillColor}{RGB}{255,255,255}

\path[draw=drawColor,line width= 0.6pt,line join=round,line cap=round,fill=fillColor] (  0.00,  0.00) rectangle (202.36,202.36);
\end{scope}
\begin{scope}
\path[clip] ( 42.06, 33.48) rectangle (196.36,196.36);
\definecolor{fillColor}{RGB}{255,255,255}

\path[fill=fillColor] ( 42.06, 33.48) rectangle (196.36,196.36);
\definecolor{drawColor}{RGB}{248,118,109}

\path[draw=drawColor,line width= 0.6pt,line join=round] ( 70.65,188.51) --
	( 75.94,188.02) --
	( 81.35,187.20) --
	( 86.86,185.14) --
	( 92.48,180.81) --
	( 98.22,169.88) --
	(104.08,138.61) --
	(110.06, 73.15) --
	(116.16, 40.88) --
	(122.38, 54.98) --
	(128.73,109.11) --
	(135.21,150.73) --
	(141.82,171.29) --
	(148.57,181.39) --
	(155.45,185.93) --
	(162.48,187.80) --
	(169.65,188.60) --
	(176.96,188.87) --
	(184.43,188.94);
\definecolor{drawColor}{RGB}{0,186,56}

\path[draw=drawColor,line width= 0.6pt,dash pattern=on 1pt off 3pt ,line join=round] ( 70.65,188.87) --
	( 75.94,188.66) --
	( 81.35,188.08) --
	( 86.86,186.40) --
	( 92.48,181.69) --
	( 98.22,168.53) --
	(104.08,136.39) --
	(110.06, 78.36) --
	(116.16, 42.79) --
	(122.38, 80.89) --
	(128.73,140.25) --
	(135.21,172.10) --
	(141.82,183.83) --
	(148.57,187.38) --
	(155.45,188.48) --
	(162.48,188.80) --
	(169.65,188.93) --
	(176.96,188.94) --
	(184.43,188.95);
\definecolor{drawColor}{RGB}{97,156,255}

\path[draw=drawColor,line width= 0.6pt,dash pattern=on 7pt off 3pt ,line join=round] ( 70.65,154.34) --
	( 75.94,150.91) --
	( 81.35,146.67) --
	( 86.86,139.67) --
	( 92.48,129.12) --
	( 98.22,113.29) --
	(104.08, 88.67) --
	(110.06, 56.57) --
	(116.16, 41.74) --
	(122.38, 58.22) --
	(128.73, 90.69) --
	(135.21,116.07) --
	(141.82,132.67) --
	(148.57,143.32) --
	(155.45,150.20) --
	(162.48,154.42) --
	(169.65,157.50) --
	(176.96,159.70) --
	(184.43,161.44);
\definecolor{drawColor}{gray}{0.30}

\path[draw=drawColor,line width= 0.6pt,line join=round,line cap=round] ( 42.06, 33.48) rectangle (196.36,196.36);
\end{scope}
\begin{scope}
\path[clip] (  0.00,  0.00) rectangle (202.36,202.36);
\definecolor{drawColor}{gray}{0.30}

\node[text=drawColor,anchor=base east,inner sep=0pt, outer sep=0pt, scale=  0.96] at ( 36.66, 30.53) {0.00};

\node[text=drawColor,anchor=base east,inner sep=0pt, outer sep=0pt, scale=  0.96] at ( 36.66, 69.31) {0.25};

\node[text=drawColor,anchor=base east,inner sep=0pt, outer sep=0pt, scale=  0.96] at ( 36.66,108.09) {0.50};

\node[text=drawColor,anchor=base east,inner sep=0pt, outer sep=0pt, scale=  0.96] at ( 36.66,146.87) {0.75};

\node[text=drawColor,anchor=base east,inner sep=0pt, outer sep=0pt, scale=  0.96] at ( 36.66,185.65) {1.00};
\end{scope}
\begin{scope}
\path[clip] (  0.00,  0.00) rectangle (202.36,202.36);
\definecolor{drawColor}{gray}{0.30}

\path[draw=drawColor,line width= 0.6pt,line join=round] ( 39.06, 33.84) --
	( 42.06, 33.84);

\path[draw=drawColor,line width= 0.6pt,line join=round] ( 39.06, 72.62) --
	( 42.06, 72.62);

\path[draw=drawColor,line width= 0.6pt,line join=round] ( 39.06,111.39) --
	( 42.06,111.39);

\path[draw=drawColor,line width= 0.6pt,line join=round] ( 39.06,150.17) --
	( 42.06,150.17);

\path[draw=drawColor,line width= 0.6pt,line join=round] ( 39.06,188.95) --
	( 42.06,188.95);
\end{scope}
\begin{scope}
\path[clip] (  0.00,  0.00) rectangle (202.36,202.36);
\definecolor{drawColor}{gray}{0.30}

\path[draw=drawColor,line width= 0.6pt,line join=round] ( 55.17, 30.48) --
	( 55.17, 33.48);

\path[draw=drawColor,line width= 0.6pt,line join=round] ( 85.67, 30.48) --
	( 85.67, 33.48);

\path[draw=drawColor,line width= 0.6pt,line join=round] (116.16, 30.48) --
	(116.16, 33.48);

\path[draw=drawColor,line width= 0.6pt,line join=round] (146.65, 30.48) --
	(146.65, 33.48);

\path[draw=drawColor,line width= 0.6pt,line join=round] (177.15, 30.48) --
	(177.15, 33.48);
\end{scope}
\begin{scope}
\path[clip] (  0.00,  0.00) rectangle (202.36,202.36);
\definecolor{drawColor}{gray}{0.30}

\node[text=drawColor,anchor=base,inner sep=0pt, outer sep=0pt, scale=  0.96] at ( 55.17, 21.46) {-0.10};

\node[text=drawColor,anchor=base,inner sep=0pt, outer sep=0pt, scale=  0.96] at ( 85.67, 21.46) {-0.05};

\node[text=drawColor,anchor=base,inner sep=0pt, outer sep=0pt, scale=  0.96] at (116.16, 21.46) {0.00};

\node[text=drawColor,anchor=base,inner sep=0pt, outer sep=0pt, scale=  0.96] at (146.65, 21.46) {0.05};

\node[text=drawColor,anchor=base,inner sep=0pt, outer sep=0pt, scale=  0.96] at (177.15, 21.46) {0.10};
\end{scope}
\begin{scope}
\path[clip] (  0.00,  0.00) rectangle (202.36,202.36);
\definecolor{drawColor}{RGB}{0,0,0}

\node[text=drawColor,anchor=base,inner sep=0pt, outer sep=0pt, scale=  1.20] at (119.21,  8.33) {$a-a_{0}$};
\end{scope}
\end{tikzpicture}\tabularnewline
(c) Size: $\lambda_{\lu}=1$. &  & (f) Power: $\lambda_{\lu}=1$ and $\omega_{\sl}=0.5$.\tabularnewline
\end{tabular}

\end{adjustwidth}

\caption{Rejection probabilities under the simulation design of \secref{simulations},
for \\ $\protect\np_{n}(a)$ (solid line), $\protect\lfst_{n}(a)$
(dashed line) and $\protect\lr_{n}(a;1)$ (dotted line)}

\label{fig:siml}
\end{figure}

The role of $\omega_{\sl}$ in regulating the extent of the size distortion
in $\lr_{n}(a;1)$ is clearly evident in panels~(a)--(c) of \figref{siml}.
That the size distortion disappears at $\omega_{\sl}=0$ is entirely
consistent with \thmref{emw}, which implies that in this case the
Hessian of the limiting loglikelihood is diagonal, and that the part
of the LR process that refers to $a$ is LAMN. Relative to the test
$\lr_{n}(a;I_{q})$ that is efficient in the presence of a unit root,
our test entirely avoids the size distortions that this test is prone
to, while giving up very little power, as is evident in panels (d)--(f).
$\lfst_{n}(a)$ also exhibits perfect size control, but the power
sacrificed for the sake of greater robustness -- with respect to
other departures from the VAR with unit roots -- is clearly evidenced
by its flatter power curves.

\subsection{Application to the expectations theory of the term structure}

\label{sec:empirical}

In \subsecref{long-run-pred} above, we discussed the expectations
theory of the term structure: in particular, how the predictions made
about cointegrating relations in a VAR with exact unit roots, could
be generalised to the quasi-cointegrating relations implied by a VAR
with some roots near unity. To evaluate these predictions empirically,
we estimated a bivariate VAR (with an intercept only) using quarterly
data on 1- and 10-year US Treasury bond yields 1953:Q2--2011:Q3,
as plotted in \figref{bond-rates-quarterly}(a). The Akaike, Bayesian,
and Hannan-Quinn information criteria  agreed on a lag length of
8, and the dominant characteristic root is estimated to be $0.983$.
Tests for cointegrating rank (using the trace test of \citealp{Joh95})
comfortably reject the null of a cointegrating rank of zero but do
not reject a cointegrating rank of one, at conventional significance
levels, which we take as evidence in favour of the presence of one
root in the vicinity of unity.

Consistent with the discussion in \subsecref{qcs}, when specifying
a parameter space $\set L=[\radius,1]$ for the dominant root, we
take $\blw h=8\times4=32$ quarters to be consistent with a lower
estimate of the average duration of the US business cycle, which yields
$\radius=2^{-1/\blw h}=0.979$. The upper panel of \figref{bond-rates-quarterly}(b)
plots the value of the maximised likelihood conditional on a range
of values for the dominant root, $\lambda_{\lu}$, over this interval.
The lower panel plots values of the coefficient on the $1$-year bond
in the normalised quasi-cointegrating vector $\beta$ with the coefficient
on the $10$-year bond being normalised to unity. The dashed line
reports the model-implied value $a_{10}(\lambda_{\lu})$ of this parameter
(see \eqref{am} above). For each value of the root in the upper panel,
the green error bars in the lower panel give the corresponding 95
per cent conditional confidence intervals $\ci_{a\mid\lambda_{\lu}}$
(see \eqref{condci} above), which are based on the efficient test
when the imposed value of $\lambda_{\lu}$ is correct; their midpoints
are given by the dotted line. Both dashed and dotted lines lie remarkably
close to each other, suggesting that the data accord well with the
predictions of the expectations theory, irrespective of the actual
value of $\lambda_{\lu}$. Finally, using the nearly optimal test
$\np_{n}$ gives a 95 per cent confidence interval of $[0.87,1.07]$,
as compared with that of $[0.92,1.13]$ when an exact unit root is
imposed.

\begin{figure}
\begin{adjustwidth}{-2cm}{-2cm}\centering\vspace*{-0.8cm}

\begin{tabular}{c}
\input{figures/term_spread_raw.tex}\tabularnewline
(a) US treasury bond yields: 1-year (solid line) and 10-year (dotted
line)\tabularnewline
\input{figures/conditional_estimates_w_marg_intval.tex}\tabularnewline
(b) Upper panel: concentrated loglikelihood at imposed characteristic
root $\lambda_{\lu}$. \\Bottom panel: 95\% conditional intervals
$\ci_{a\mid\lambda_{\lu}}$ (green error bars); $\np_{n}$ interval
(horizontal lines); \\ midpoints of $\ci_{a\mid\lambda_{\lu}}$ (dotted
green line); theory-implied value of $a_{10}(\lambda_{\lu})$ (dashed
blue line)\tabularnewline
\end{tabular}

\end{adjustwidth}

\caption{\label{fig:bond-rates-quarterly}Expectations theory of term spread.}
\end{figure}

\section{Conclusion}

This paper was motivated by \citeauthor{Ell98Ecta}'s \citeyearpar{Ell98Ecta}
finding that inference on cointegrating relationships are highly sensitive
to departures from the assumption of exact unit roots. We have argued
that this problem is as much one of (non-)identification as it is
of inference, because of the manner in which the conventional definitions
of cointegration break down in the absence of exact unit roots. We
have therefore developed an alternative characterisation of cointegration
in an SVAR, in which the long-run equilibrium relationships between
the series are identified by those directions for which the impulse
responses decay (relatively) most rapidly. With $q$ roots at unity,
this exactly recovers the $r$-dimensional cointegrating space --
and when these roots merely near unity, there remains a well-defined
$q$-dimensional quasi-cointegrating space. While this is not the
only possible way of extending `cointegration' to a wider domain,
a conceptual advantage of the approach taken here is that it maintains
the duality that exists, in an SVAR with exact unit roots, between
the identification of the long-run equilibrium relationships between
the series, and of the subvector of structural shocks whose common
persistent effects underpin those relationships.

Likelihood-based inference on the (quasi-)cointegrating relationships
is affected by nuisance parameters corresponding to the proximity
of the dominant $q$ roots to unity. We have shown that this problem
is not merely reminiscent of inference in a predictive regression,
but in fact asymptotically equivalent in the sense of sharing a common
limiting experiment. Our problem also falls within the class of problems
studied by \citet{EMW15Ecta}, and we have found that, in practice,
tests with excellent size and power properties can be developed by
adapting their approach to the present setting.

\section{References}

{\singlespacing

\bibliographystyle{ecta}
\bibliography{time-series,asymptotics}
}

\newpage{}

\appendix

\part*{Appendices}

\addcontentsline{toc}{part}{Appendices}
\begin{notation*}
For $x\in\reals^{p}$ and $A\in\reals^{p\times p}$, $\smlnorm x$
denotes the Euclidean norm and $\smlnorm A\defeq\sup_{\smlnorm x=1}\smlnorm{Ax}$
the induced matrix norm.
\end{notation*}

\section{Representation theory}

\label{app:qcs-results}

This section provides results that support some of the assertions
made in the course of Sections~\ref{sec:cointegration} and \ref{sec:estimation},
and which are auxiliary to results proved in the following appendices.
Some are well known, but are collected here for ease of reference.
Proofs follow at the end of this appendix. \assref{DGP} is maintained
throughout.

For VAR coefficients $\PHI\defeq(\Phi_{1},\ldots,\Phi_{k})\in\reals^{p\times kp}$,
let
\begin{equation}
F\defeq F(\PHI)\defeq\begin{bmatrix}\Phi_{1} & \Phi_{2} & \cdots & \Phi_{k-1} & \Phi_{k}\\
I & 0 & \cdots & 0 & 0\\
0 & I & \cdots & 0 & 0\\
\vdots & \vdots & \ddots & \vdots & \vdots\\
0 & 0 & \cdots & I & 0
\end{bmatrix}\label{eq:companion}
\end{equation}
denote the associated companion form matrix. For a collection of $m\times n$
matrices $Z_{1},\ldots,Z_{k}$, let 
\[
\col\{Z_{i}\}_{i=1}^{k}\defeq\begin{bmatrix}Z_{1}\\
\vdots\\
Z_{k}
\end{bmatrix},
\]
so that taking $\x_{t}\defeq\col\{x_{t-i}\}_{i=0}^{k-1}$, we may
write \eqref{dgp} as
\begin{equation}
\x_{t}=F\x_{t-1}+\begin{bmatrix}\err_{t}\\
0_{(k-1)p\times1}
\end{bmatrix}\eqdef F\x_{t-1}+\Err_{t}\label{eq:compVAR}
\end{equation}
Let $\lambda_{i}(\PHI)$ denote the $i$th root of the characteristic
polynomial associated to $\PHI$, when these are placed in descending
order of modulus.
\begin{lem}
\label{lem:GLR}Suppose that $\smlabs{\lambda_{q}(\PHI)}>\smlabs{\lambda_{q+1}(\PHI)}$
for some $q\in\{1,\ldots,p\}$. Then there exist there matrices $R\in\reals^{p\times kp}$,
$\Lambda\in\reals^{kp\times kp}$ and $L\in\reals^{p\times kp}$ such
that:
\begin{enumerate}
\item \label{enu:GLR:split}$\Lambda=\diag\{\Lambda_{\lu},\Lambda_{\st}\}$,
where the eigenvalues of $\Lambda_{\lu}\in\reals^{q\times q}$ and
$\Lambda_{\st}$ are $\{\lambda_{i}(\PHI)\}_{i=1}^{q}$ and $\{\lambda_{i}(\PHI)\}_{i=q+1}^{kp}$
respectively;
\item \label{enu:GLR:subspc}the following hold:
\begin{align}
R\Lambda^{k}-\sum_{i=1}^{k}\Phi_{i}R\Lambda^{k-i} & =0 & \Lambda^{k}L^{\trans}-\sum_{i=1}^{k}\Lambda^{k-i}L^{\trans}\Phi_{i} & =0.\label{eq:GLR:RL}
\end{align}
\item \label{enu:GLR:L}$\R\defeq\col\{R\Lambda^{k-i}\}_{i=1}^{k}$ is invertible,
and $L$ equals the first $p$ rows of $\L\defeq(\R^{-1})^{\trans}$;
\item \label{enu:GLR:jordanF}$F(\PHI)=\R\Lambda\L^{\trans}$; and
\item \label{enu:GLR:IRF} in the model \eqref{dgp}, $\irf_{s}^{\err}\defeq\partial y_{t+s}/\partial\err_{t}=R\Lambda^{k-1+s}L^{\trans}$
for $s\geq1$.
\end{enumerate}
Further, the matrices $R^{\ast}\in\reals^{p\times kp},$ $\Lambda^{\ast}\in\reals^{kp\times kp}$
and $L^{\ast}\in\reals^{p\times kp}$ satisfy conditions \enuref{GLR:split}--\enuref{GLR:IRF}
if and only if there exists an invertible $kp\times kp$ matrix $Q=\diag\{Q_{\lu},Q_{\st}\}$,
where $Q_{\lu}\in\reals^{q\times q}$, such that $R^{\ast}=RQ$, $\Lambda^{\ast}=Q^{-1}\Lambda Q$
and $L^{\ast}=L(Q^{\trans})^{-1}$.
\end{lem}
For a given $\PHI$, and its associated companion form $F=F(\PHI)$,
we shall routinely partition the matrices appearing in \lemref{GLR}
as
\begin{align}
R & \defeq[R_{\lu},R_{\st}] & \R & \defeq[\R_{\lu},\R_{\st}] & L & \defeq[L_{\lu},L_{\st}] & \L & \defeq[\L_{\lu},\L_{\st}]\label{eq:partition}
\end{align}
where each of $R_{\lu}$, $\R_{\lu}$, $L_{\lu}$ and $\L_{\lu}$
have $q$ columns, i.e.\ the partitioning is conformable with that
of $\Lambda=\diag\{\Lambda_{\lu},\Lambda_{\st}\}$. This partitioning,
in conjunction with parts~\enuref{GLR:subspc} and \enuref{GLR:IRF}
of the preceding lemma, yields \eqref{eig-eig} and \eqref{irfdecomp}
above. Moreover, we may write part~\enuref{GLR:jordanF} as
\begin{equation}
F=\R\Lambda\L^{\trans}=\R_{\lu}\Lambda_{\lu}\L_{\lu}^{\trans}+\R_{\st}\Lambda_{\st}\L_{\st}^{\trans}\label{eq:invsubdecomp}
\end{equation}
which decomposes $F$ with respect to the invariant subspaces associated
to the eigenvalues of $\Lambda_{\lu}$ and $\Lambda_{\st}$.
\begin{lem}
\label{lem:GLR-converse} Suppose that $\smlabs{\lambda_{q}(\PHI)}>\smlabs{\lambda_{q+1}(\PHI)}$
for some $q\in\{1,\ldots,p\}$, the eigenvalues of $\Lambda_{0}\in\reals^{q\times q}$
are all greater than $\smlabs{\lambda_{q+1}(\PHI)}$ in modulus, and
$R_{0}\in\reals^{p\times q}$ is a full column rank matrix such that
\begin{equation}
R_{0}\Lambda_{0}^{k}-\sum_{i=1}^{k}\Phi_{i}R_{0}\Lambda_{0}^{k-i}=0.\label{eq:R0cond}
\end{equation}
Then there exist matrices $R=[R_{\lu},R_{\st}]$, $\Lambda=\diag\{\Lambda_{\lu},\Lambda_{\st}\}$
and $L$ satisfying the conditions of \lemref{GLR}, with $R_{\lu}=R_{0}$
and $\Lambda_{\lu}=\Lambda_{0}$.
\end{lem}
For the next result, recall the definition of $S_{r}$ given in the
context of \eqref{relative-irf} above. By showing that $S_{r}=(\spn R_{\lu})^{\perp}$,
we substantiate the claim made in the \subsecref{coint-ur} that $S_{r}$
is invariant to the identification of the structural shocks $w_{t}$.
$S_{r}$ may thus be equivalently defined with $\irf_{s}^{\err}$
taking the place of $\irf_{s}^{w}$ in \eqref{relative-irf}.
\begin{lem}
\label{lem:qcs}~
\begin{enumerate}
\item If \assref{QC} holds for some $\radius\in(0,1]$, then $S_{r}=(\spn R_{\lu})^{\perp}$.
\item \label{enu:qcs:cvar}If \assref{J} holds, then $\cs=S_{r}=(\spn R_{\lu})^{\perp}$,
and \assref{QC} holds with $\radius=1$.
\end{enumerate}
\end{lem}
\begin{lem}
\label{lem:repasAR}Suppose \assref{QC} holds. Let $\Lambda=\diag\{\Lambda_{\lu},\Lambda_{\st}\}$,
$R=[R_{\lu},R_{\st}]$ and $\L=[\L_{\lu},\L_{\st}]$ be as in \lemref{GLR}
and \eqref{partition}. Then \eqref{GJ-type-decomp}--\eqref{zproc}
hold with $\Phi_{\lu}=R_{\lu}\Lambda_{\lu}^{k}$, $\Phi_{\st}=R_{\st}\Lambda_{\st}^{k}$,
$z_{\lu,t}\defeq\L_{\lu}^{\trans}\x_{t}$ and $z_{\st,t}\defeq\L_{\st}^{\trans}\x_{t}$.
\end{lem}
One aspect of the cointegrated VAR with exact unit roots, which does
not readily translate to the quasi-cointegrated VAR, is the `error
correction' representation of the adjustments towards equilibrium.
The only situation in which an analogue of that representation is
available is when the $q$ largest roots are identical, so that $\Lambda_{\lu}=\lambda_{0}I_{q}$
for some $\lambda_{0}\in\eigs_{\lu}^{\radius}$, as arises automatically
if $q=1$. (Note that $\lambda_{0}$ must be real in this case.) In
that case, letting $\Delta_{\lambda}x_{t}\defeq x_{t}-\lambda x_{t-1}$,
we have the following.
\begin{lem}
\label{lem:quasi-VECM}Suppose \assref{QC} holds with $\Lambda_{\lu}=\lambda_{0}I_{q}$.
Then $\{x_{t}\}$ satisfies
\begin{equation}
\Delta_{\lambda_{0}}x_{t}=\Pi x_{t-1}+\sum_{i=1}^{k-1}\Psi_{i}\Delta_{\lambda_{0}}x_{t-i}+\err_{t},\label{eq:VECM}
\end{equation}
where $\Pi\defeq-\lambda_{0}^{-k+1}\Phi(\lambda_{0})=\alpha\beta^{\trans}$,
for some $\alpha\in\reals^{p\times r}$ having full column rank. Moreover,
$(\Delta_{\lambda_{0}}x_{t},\beta^{\trans}x_{t})$ follow a VAR($k-1$),
all of whose characteristic roots lie in $\eigs_{\st}^{\radius}$.
\end{lem}
\begin{proof}[Proof of \lemref{GLR}]
Let $J$ denote a $(kp\times kp)$ real Jordan matrix similar to
$F$, each of whose diagonal blocks correspond to roots of $\Phi(\cdot)$,
so that $P^{-1}FP=J$ for some $P\in\reals^{kp\times kp}$. We may
take the diagonal blocks of $J$ to be ordered such that $J=\diag\{J_{\lu},J_{\st}\}$,
where $J_{\lu}\in\reals^{q\times q}$ has all its eigenvalues in $\eigs_{\lu}^{\radius}$.
Letting
\[
X\defeq[\begin{matrix}0_{p} & \cdots & 0_{p} & I_{p}\end{matrix}]P
\]
we have by \citet[Thm.~1.24 and 1.25]{GLR82} that the matrices $(X,J)$
form a standard pair for $\Phi(\cdot)$.\footnote{Note that the `first companion form' matrix defined by these authors
($C_{1}$ on p.~13 of that work) equals $F$ with the ordering of
its rows and columns reversed, so our definitions of $X$ (and below,
$Y$) differ from theirs.} Therefore,
\[
XJ^{k}-\sum_{i=1}^{k}\Phi_{i}XJ^{k-i}=0,
\]
and $\col\{XJ^{k-i}\}_{i=1}^{k}=P$ is invertible, so that the matrix
\begin{equation}
Y\defeq[\begin{matrix}I_{p} & \cdots & 0_{p} & 0_{p}\end{matrix}](P^{\trans})^{-1}\label{eq:Y}
\end{equation}
is well defined. By \citet[Prop.~2.1]{GLR82}, $(Y,J)$ satisfy
\[
J^{k}Y^{\trans}-\sum_{i=1}^{k}J^{k-i}Y^{\trans}\Phi_{i}=0.
\]
Parts~\enuref{GLR:split}--\enuref{GLR:jordanF} of the lemma are
thus satisfied with $(R,\Lambda,L^{\trans},\R,\L^{\trans})=(X,J,Y^{\trans},P,P^{-1})$.
It further follows by recursive substitution that
\[
\irf_{s}=[F^{s}]_{11}=[\R\Lambda^{s}\L^{\trans}]_{11}=R\Lambda^{k-1+s}L^{\trans}
\]
where $[A]_{11}$ denotes the upper left $p\times p$ block of the
matrix $A$; thus part~\enuref{GLR:IRF} is proved.

Finally, let $Q=\diag\{Q_{\lu},Q_{\st}\}$ be as in the final part
of the lemma. It is easily verified that
\begin{align*}
\Lambda_{\ast} & \defeq\diag\{Q_{\lu}^{-1}\Lambda_{\lu}Q_{\lu},Q_{\st}^{-1}\Lambda_{\st}Q_{\st}\}=Q^{-1}\Lambda Q,
\end{align*}
$R_{\ast}\defeq RQ$ and $L_{\ast}\defeq L(Q^{\trans})^{-1}$ have
the required properties. Conversely, if both $(R,\Lambda,L)$ and
$(R_{\ast},\Lambda_{\ast},L_{\ast})$ satisfy conditions \enuref{GLR:split}--\enuref{GLR:IRF},
then both $\Lambda$ and $\Lambda_{\ast}$ are block diagonal matrices
similar to $J=\diag\{J_{\lu},J_{\st}\}$, whence there exists $Q=\diag\{Q_{\lu},Q_{\st}\}$
such that $\Lambda^{\ast}=Q^{-1}\Lambda Q$, etc.
\end{proof}
\begin{proof}[Proof of \lemref{GLR-converse}]
 $\R_{0}\defeq\col\{R_{0}\Lambda_{0}^{k-i}\}_{i=1}^{k}\in\reals^{kp\times q}$
has rank $q$, and \eqref{R0cond} implies that $F\R_{0}=\R_{0}\Lambda_{0}$,
for $F\defeq F(\PHI)$. Since the remaining $kp-q$ eigenvalues of
$F$ are distinct from the eigenvalues of $\Lambda_{0}$, $\R_{0}$
is a simple invariant subspace of $F$ (\citealp[Defn~V.1.2]{SS90}).
Hence there exist $\R,\Lambda,\L\in\reals^{kp\times kp}$ such that
$F=\R\Lambda\L^{\trans}$ and $\L^{\trans}\R=I_{kp}$, and $\R$ and
$\Lambda$ can be partitioned as $\R=[\R_{0},\R_{\st}]$ and $\Lambda=\diag\{\Lambda_{0},\Lambda_{\st}\}$
(\citealp[Thm~V.1.5]{SS90}). Since $\Lambda_{0}$ and $\Lambda_{\st}$
must be similar to the blocks $J_{\lu}$ and $J_{\st}$ of the real
Jordan form of $F$, as introduced in the proof of \lemref{GLR},
the result then follows by the same arguments as were given in that
proof.
\end{proof}
\begin{proof}[Proof of \lemref{qcs}]
\proofpart{1} By \lemref{GLR}\enuref{GLR:IRF}, for any $b\in\reals^{p}$,
\begin{equation}
b^{\trans}\irf_{s}^{w}=b^{\trans}\irf_{s}^{\err}\Upsilon=b^{\trans}R\Lambda^{k-1+s}L^{\trans}\Upsilon=b^{\trans}R_{\lu}\Lambda_{\lu}^{k-1+s}L_{\lu}^{\trans}\Upsilon+b^{\trans}R_{\st}\Lambda_{\st}^{k-1+s}L_{\st}^{\trans}\Upsilon.\label{eq:irf-b}
\end{equation}
Since the spectral radius of $\Lambda_{\st}$ is strictly less than
$\radius$, we have by \citet[Cor.~5.6.13]{HJ13book} that
\begin{equation}
\Lambda_{\st}^{t}/\radius^{t}\goesto0\label{eq:Lst_negl}
\end{equation}
as $t\goesto\infty$. Since $\Lambda_{\lu}$ is diagonalisable under
\enuref{Q:rank}, by \lemref{GLR} we may choose $(R_{\lu},\Lambda_{\lu},L_{\lu})$
such that $\Lambda_{\lu}$ is a real Jordan block diagonal matrix
(as in Corollary 3.4.1.10 of \citealp{HJ13book}). The eigenvalues
of $\Lambda_{\lu}^{\trans}\Lambda_{\lu}=\Lambda_{\lu}\Lambda_{\lu}^{\trans}$
are therefore of the form $\smlabs{\lambda}^{2}$, for $\lambda$
an eigenvalue of $\Lambda_{\lu}$, and thus $\spect_{\min}(\Lambda_{\lu}^{\trans}\Lambda_{\lu})\geq\radius^{2}$,
where $\spect_{\min}(M)$ denotes the smallest eigenvalue of a positive-definite
matrix $M$. Therefore letting $x\defeq R_{\lu}^{\trans}b$,
\begin{multline*}
\smlnorm{x^{\trans}\Lambda_{\lu}^{t}L_{\lu}^{\trans}\Upsilon}^{2}\geq\spect_{\min}(L_{\lu}^{\trans}\Upsilon^{\trans}\Upsilon L_{\lu})\smlnorm{x^{\trans}\Lambda_{\lu}^{t}}^{2}\\
\geq\radius\spect_{\min}(L_{\lu}^{\trans}\Upsilon^{\trans}\Upsilon L_{\lu})\smlnorm{\Lambda_{\lu}^{t-1}x}^{2}\geq\cdots\geq\radius^{2t}\spect_{\min}(L_{\lu}^{\trans}\Upsilon^{\trans}\Upsilon L_{\lu})\smlnorm x^{2}.
\end{multline*}
$\spect_{\min}(L_{\lu}^{\trans}\Upsilon^{\trans}\Upsilon L_{\lu})>0$,
since $\Upsilon$ is nonsingular, and $L_{\lu}$ has full column rank
under \enuref{Q:rank}. Deduce that if $b^{\trans}R_{\lu}\neq0$,
then
\begin{equation}
\liminf_{t\goesto\infty}\smlnorm{b^{\trans}R_{\lu}\Lambda_{\lu}^{t}L_{\lu}^{\trans}\Upsilon}/\radius^{t}>0.\label{eq:Lu_bdd}
\end{equation}
It follows from \eqref{irf-b}--\eqref{Lu_bdd}\noeqref{eq:Lst_negl}
that $b^{\trans}\irf_{s}^{w}/\radius^{s}\goesto0$ as $s\goesto\infty$
if and only if $b\perp\spn R_{\lu}$. Thus $(\spn R_{\lu})^{\perp}$
gives the unique $r$-dimensional subspace of $\reals^{p}$ satisfying
the definition of $S_{r}$.

\proofpart{2} Since $\rank\Phi(1)=p-q$ under \enuref{J:rank}, there
exists $R_{\lu}\in\reals^{p\times q}$ having rank $q$ such that
\begin{equation}
0=\Phi(1)R_{\lu}=R_{\lu}-\sum_{i=1}^{k}\Phi_{i}R_{\lu}=_{(1)}R_{\lu}\Lambda_{\lu}^{k}-\sum_{i=1}^{k}\Phi_{i}R_{\lu}\Lambda_{\lu}^{k-i}\label{eq:Phi1}
\end{equation}
where $=_{(1)}$ follows by taking $\Lambda_{\lu}=I_{q}$. By a similar
argument, here exists a $L_{\lu}\in\reals^{p\times q}$ with $\rank L_{\lu}=q$
and $L_{\lu}^{\trans}\Phi(1)=0$. \assref{J} is thus a special case
of \assref{QC} with $\radius=1$. $S_{r}=(\spn R_{\lu})^{\perp}$
therefore follows immediately from part~(i) of the lemma. Finally,
recall from the second characterisation of the cointegrating space
given in \subsecref{coint-ur} that $\cs=\{\ker\Phi(1)\}^{\perp}$.
By \eqref{Phi1} this also coincides with $(\spn R_{\lu})^{\perp}$.
\end{proof}
\begin{proof}[Proof of \lemref{repasAR}]
 By \eqref{compVAR} and \lemref{GLR},
\[
\L^{\trans}\x_{t}=\L^{\trans}F\x_{t-1}+L^{\trans}\err_{t}=\Lambda\L^{\trans}\x_{t-1}+L^{\trans}\err_{t}.
\]
Since $\Lambda=\diag\{\Lambda_{\lu},\Lambda_{\st}\}$, it is clear
that \eqref{zproc} holds for $z_{\lu,t}$ and $z_{\st,t}$ as defined
in the lemma. Further, taking the first $p$ rows of \eqref{compVAR}
and using \lemref{GLR} again yields
\[
x_{t}=R\Lambda^{k}\L^{\trans}\x_{t-1}+\err_{t}=R_{\lu}\Lambda_{\lu}^{k}\L_{\lu}^{\trans}\x_{t-1}+R_{\st}\Lambda_{\st}^{k}\L_{\st}^{\trans}\x_{t-1}+\err_{t}.\qedhere
\]
\end{proof}
\begin{proof}[Proof of \lemref{quasi-VECM}]
 The representation \eqref{VECM} follows directly from Theorem~1
in \citet{JS99JoE}. By the same arguments as given in the proof of
Corollary~4.3 in \citet{Joh95}, since $\Phi$ has $q$ roots at
$\lambda_{0}$, $\Phi(\lambda_{0})$ must have rank at least equal
to $r=p-q$; that it has rank equal to $r$ then follows from \ref{enu:Q:rank},
which implies that $\Phi(\lambda_{0})R_{\lu}=0$. For the final claim,
we note that under the maintained assumption that $\Lambda_{\lu}=\lambda_{0}I_{q}$,
we have from \eqref{invsubdecomp} that
\begin{equation}
\Delta_{\lambda_{0}}\x_{t}=\x_{t}-\lambda_{0}\x_{t-1}=(F-\lambda_{0}I_{kp})\x_{t-1}+\Err_{t}=\R_{\st}(\Lambda_{\st}-\lambda_{0}I_{kp-q})\L_{\st}^{\trans}\x_{t}+\Err_{t}.\label{eq:quasidiff}
\end{equation}
By \lemref{GLR}, $\L_{\st}\in\reals^{kp\times(kp-q)}$ is a full
column rank matrix such that $\L_{\st}^{\trans}\R_{\lu}=0$, where
$\R_{\lu}=\col\{R_{\lu}\Lambda_{\lu}^{k-i}\}_{i=1}^{k}=\col\{\lambda_{0}^{k-i}R_{\lu}\}_{i=1}^{k}$.
By considering the column span of $\R_{\lu}$, it follows that $\L_{\st}$
must have the same column span as 
\[
\begin{bmatrix}\beta & I_{p}\\
 & -\lambda_{0}I_{p} & \ddots\\
 &  & \ddots & I_{p}\\
 &  &  & -\lambda_{0}I_{p}
\end{bmatrix}.
\]
Hence there exists a full-rank $\Psi\in\reals^{(kp-q)\times(kp-q)}$
such that
\begin{equation}
\L_{\st}^{\trans}\x_{t}=\Psi\begin{bmatrix}\beta^{\trans}x_{t}\\
\Delta_{\lambda_{0}}x_{t}\\
\vdots\\
\Delta_{\lambda_{0}}x_{t-k+2}
\end{bmatrix}.\label{eq:Lst-x}
\end{equation}
We recall from \lemref{repasAR} that the l.h.s.\ is equal to $z_{\st,t}$,
which by that result has a first-order autoregressive representation
with characteristic roots that are the eigenvalues of $\Lambda_{\st}$,
and hence in $\eigs_{\st}^{\radius}$. Indeed, \eqref{zLU} provides
the companion form representation for autoregressive process followed
by $(\Delta_{\lambda_{0}}x_{t},\beta^{\trans}x_{t})$, which accordingly
has the properties claimed.
\end{proof}
\begin{proof}[Proof of \propref{reptheorybody}]
 This is an immediate corollary of \lemref{qcs}.
\end{proof}

\section{Perturbation theory\label{app:perturbation}}

Recall the definition of $\set P\subset\reals^{p\times kp}$ given
in \subsecref{Psetdef}. The normalisation \eqref{Rlurnom} entails
that
\begin{equation}
\begin{bmatrix}A\\
I_{q}
\end{bmatrix}\Lambda_{\lu}^{k}-\sum_{i=1}^{k}\Phi_{i}\begin{bmatrix}A\\
I_{q}
\end{bmatrix}\Lambda_{\lu}^{k-i}=0\label{eq:AlambdaR}
\end{equation}
which by Lemmas~\ref{lem:GLR} and \ref{lem:GLR-converse} uniquely
determines $R_{\lu}=[\begin{smallmatrix}A\\
I_{q}
\end{smallmatrix}]$ and $\Lambda_{\lu}$ as a function of $\PHI\in\set P$. As in \subsecref{functionals},
we shall denote the implied mappings by $R_{\lu}(\PHI)$, $A(\PHI)$,
$\Lambda_{\lu}(\PHI)$, and $\R_{\lu}(\PHI)\defeq\col\{R_{\lu}(\PHI)\Lambda_{\lu}^{k-i}(\PHI)\}_{i=1}^{k}$.
Our first result is that these are smooth (i.e.\ infinitely differentiable);
its proof and those of the subsequent lemmas appear at the end of
this appendix.
\begin{lem}
\label{lem:implicit-maps} $\set P$ is open; and $A(\PHI)$ and $\Lambda_{\lu}(\PHI)$
are smooth on $\set P$.
\end{lem}

Our next result gives the first derivatives of the maps $A(\PHI)$
and $\Lambda_{\lu}(\PHI)$; it is closely related to Theorem~2.1
in \citet{Sun91LinAlg}. To express these derivatives more concisely,
let
\begin{equation}
B(\PHI)\defeq(I_{q}\otimes R_{\st})[(\Lambda_{\lu}^{\trans}\otimes I_{kp-q})-(I_{q}\otimes\Lambda_{\st})]^{-1}(I_{q}\otimes L_{\st}^{\trans}),\label{eq:Bdef}
\end{equation}
where we have suppressed the dependence of each of the r.h.s.\ quantities
on $\PHI$. The matrix in square brackets on the r.h.s.\ has eigenvalues
of the form $\lambda-\mu$, where $\lambda$ and $\mu$ are eigenvalues
of $\Lambda_{\lu}$ and $\Lambda_{\st}$ respectively; it is thus
invertible for all $\PHI\in\set P$. Under the normalisation implied
by \eqref{AlambdaR}, $B$ is uniquely determined by $\PHI\in\set P$,
even though $R_{\st}$, $\Lambda_{\st}$ and $L_{\st}$ individually
are not (as follows from the final part of \lemref{GLR}).
\begin{lem}
\label{lem:derivatives} Let $\PHI_{0}\in\set P$, $A_{0}\defeq A(\PHI_{0})$,
$\Lambda_{0,\lu}\defeq\Lambda_{\lu}(\PHI_{0})$, $R_{0,\lu}\defeq[\begin{smallmatrix}A_{0}\\
I_{q}
\end{smallmatrix}]$ and $\R_{0,\lu}\defeq\col\{R_{0,\lu}\Lambda_{0,\lu}^{k-i}\}_{i=1}^{k}$.
Then
\begin{enumerate}
\item \label{enu:deriv:zero}$A_{0}=A(\PHI)$ and $\Lambda_{0,\lu}=\Lambda_{\lu}(\PHI)$
for all $\PHI\in\set P$ such that $(\PHI-\PHI_{0})\R_{0,\lu}=0$
and $\smlabs{\lambda_{q+1}(\PHI)}<\smlabs{\lambda_{q}(\PHI_{0})}$;
\item \label{enu:deriv:values}the first differentials of $A(\cdot)$ and
$\Lambda_{\lu}(\cdot)$ at $\PHI=\PHI_{0}$ satisfy\footnote{For a more compact notation, here and subsequently we express matrix
derivatives in terms of differentials, in the manner of \citet{MN07}.}
\[
\begin{bmatrix}\vek(\deriv A)\\
\vek(\deriv\Lambda_{\lu})
\end{bmatrix}=\begin{bmatrix}J_{A}(\PHI_{0})\\
J_{\Lambda}(\PHI_{0})
\end{bmatrix}\vek\{(\deriv\PHI)\R_{0,\lu}\}
\]
where
\begin{equation}
J(\PHI)\defeq\begin{bmatrix}J_{A}(\PHI)\\
J_{\Lambda}(\PHI)
\end{bmatrix}\defeq\begin{bmatrix}(I_{q}\otimes\beta^{\trans})B\\{}
[(\Lambda_{\lu}^{\trans}\otimes I_{q})-(I_{q}\otimes\Lambda_{\lu})](I_{q}\otimes G^{\trans})B+(I_{q}\otimes L_{\lu}^{\trans})
\end{bmatrix}\label{eq:Jacobs}
\end{equation}
for $G^{\trans}\defeq[0_{q\times r},I_{q}]$, $\beta^{\trans}=[I_{r},-A]$,
and $\Lambda_{\lu}=\Lambda_{\lu}(\PHI)$, etc.; and
\item \label{enu:deriv:cont}$J(\PHI)$ is continuous.
\end{enumerate}
\end{lem}

When $\Lambda_{\lu}(\PHI)=I_{q}$, the $pq\times pq$ matrix $J(\PHI)$
simplifies as follows.
\begin{lem}
\label{lem:derivatunity}Suppose $\PHI\in\set P$ with $\Lambda_{\lu}(\PHI)=I_{q}$.
Then $J(\PHI)$ is nonsingular, and
\[
\begin{bmatrix}J_{A}(\PHI)\\
J_{\Lambda}(\PHI)
\end{bmatrix}=\begin{bmatrix}I_{q}\otimes\beta^{\trans}R_{\st}(I_{kp-q}-\Lambda_{\st})^{-1}L_{\st}^{\trans}\\
I_{q}\otimes L_{\lu}^{\trans}
\end{bmatrix}.
\]
\end{lem}

\begin{proof}[Proof of \lemref{implicit-maps}]
 We first prove $\set P$ is open. For $F\in\reals^{kp\times kp}$,
let $\lambda_{i}(F)$ denote the $i$th eigenvalue of $F$, when these
are placed in descending order of modulus. Let $\set F$ denote the
set of $kp\times kp$ matrices such that
\begin{enumerate}
\item $\smlabs{\lambda_{q+1}(F)}<\smlabs{\lambda_{q}(F)}$; and
\end{enumerate}
there exist $\Lambda_{\lu}\in\reals^{q\times q}$ and $\R_{\lu}\in\reals^{kp\times q}$
such that
\begin{enumerate}[resume]
\item the eigenvalues of $\Lambda_{\lu}$ are $\{\lambda_{i}(F)\}_{i=1}^{q}$,
$F\R_{\lu}=\R_{\lu}\Lambda_{\lu}$; and
\item $\rank\{\G^{\trans}\R_{\lu}\}=q$, where $\G^{\trans}\defeq[0_{q\times(kp-q)},I_{q}]=[0_{q\times k(p-1)},G^{\trans}]$.
\end{enumerate}
In view of \lemref{GLR}, $\PHI\in\set P$ if and only if the companion
form matrix $F(\PHI)$ is in $\set F$. Since $F(\cdot)$ is trivially
continuous, it suffices to show that $\set F$ is open.

To that end, fix $F_{0}\in\set F$, and let $\R_{0,\lu}$ and $\Lambda_{0,\lu}$
denote matrices satisfying (ii) and (iii) above. By the continuity
of eigenvalues and simple invariant subspaces (Theorems~IV.1.1 and
V.2.8 in \citealp{SS90}), for every $\epsilon>0$ there exists a
$\delta>0$ such that whenever $\smlnorm{F-F_{0}}<\delta$, $F$ satisfies
requirements (i) and (ii) above, with associated $\R_{\lu}$ such
that $\smlnorm{\R_{\lu}-\R_{0,\lu}}<\epsilon$. Since the set of full
rank matrices is open, we may take $\epsilon>0$ sufficiently small
such that (iii) also holds. Thus $F\in\set F$, and so $F_{0}$ is
an interior point of $\set F$; deduce $\set F$ is open.

We turn next to the smoothness of $A(\PHI)$ and $\Lambda_{\lu}(\PHI)$.
For $F_{0}\in\set F$ we have the invariant subspace decomposition
(as per \eqref{invsubdecomp} above)
\begin{equation}
F_{0}=\R_{0,\lu}\Lambda_{0,\lu}\L_{0,\lu}^{\trans}+\R_{0,\st}\Lambda_{0,\st}\L_{0,\st}^{\trans}\label{eq:Fsubs}
\end{equation}
where $\R_{0,\lu}$ and $\Lambda_{0,\lu}$ satisfy (ii)--(iii) above.
Since (iii) holds, we may choose $\R_{0,\lu}$ such that $\G^{\trans}\R_{0,\lu}=I_{q}$;
note that $\L_{0}^{\trans}\R_{0}=I_{kp}$ (as per \lemref{GLR}\enuref{GLR:L})
implies $\L_{0,\lu}^{\trans}\R_{0,\lu}=I_{q}$. Define the maps\begin{subequations}\label{eq:HHast}
\begin{align}
H(\R_{\lu},\Lambda_{\lu};F) & \defeq\begin{bmatrix}\R_{\lu}\Lambda_{\lu}-F\R_{\lu}; & \G^{\trans}\R_{\lu}-I_{q}\end{bmatrix}\\
H^{\ast}(\R_{\lu},\Lambda_{\lu};F) & \defeq\begin{bmatrix}\R_{\lu}\Lambda_{\lu}-F\R_{\lu}; & \L_{0,\lu}^{\trans}\R_{\lu}-I_{q}\end{bmatrix},\label{eq:Hast}
\end{align}
\end{subequations}so that $H(\R_{0,\lu},\Lambda_{0,\lu};F_{0})=H^{\ast}(\R_{0,\lu},\Lambda_{0,\lu};F_{0})=0$;
but note that these maps need not otherwise agree, since they impose
distinct normalisations on $\R_{\lu}$. Once we have shown that the
Jacobian of $H^{\ast}$ with respect to $(\R_{\lu},\Lambda_{\lu})$
is nonsingular at $(\R_{0,\lu},\Lambda_{0,\lu};F_{0})$, it will follow
by the implicit mapping theorem (\citealp[Thm.~XIV.2.1]{Lang93})
that there exists a neighbourhood $N\subset\set F$ of $F_{0}$ and
smooth functions $\R_{\lu}^{\ast}:N\setmap\reals^{kp\times q}$, $\Lambda_{\lu}^{\ast}:N\setmap\reals^{q\times q}$
such that
\[
H^{\ast}[\R_{\lu}^{\ast}(F),\Lambda_{\lu}^{\ast}(F);F]=0
\]
for all $F\in N$; by the continuity of $\R_{\lu}^{\ast}(\cdot)$,
we may choose $N$ such that $\rank\{\G^{\trans}\R_{\lu}^{\ast}(F)\}=q$
for all $F\in N$. Thus
\begin{align}
\R_{\lu}(F) & \defeq\R_{\lu}^{\ast}(F)[\G^{\trans}\R_{\lu}^{\ast}(F)]^{-1}\label{eq:renorm1}\\
\Lambda_{\lu}(F) & \defeq[\G^{\trans}\R_{\lu}^{\ast}(F)]\Lambda_{\lu}^{\ast}(F)[\G^{\trans}\R_{\lu}^{\ast}(F)]^{-1}\label{eq:renorm2}
\end{align}
are well defined for all $F\in N$, and have the property that
\[
H[\R_{\lu}(F),\Lambda_{\lu}(F);F]=0
\]
for all $F\in N$. Since the $(\R_{\lu},\Lambda_{\lu})$ satisfying
$H(\R_{\lu},\Lambda_{\lu};F)=0$ is unique, repeating this construction
for every $F_{0}\in\set F$ allows the smooth maps $\R_{\lu}(F)$
and $\Lambda_{\lu}(F)$ to be extended to the whole of $\set F$.
The smoothness of $\Lambda_{\lu}(\PHI)\defeq\Lambda_{\lu}[F(\PHI)]$
and $\R_{\lu}(\PHI)\defeq\R_{\lu}[F(\PHI)]$ follows immediately,
and that of $A(\PHI)$ by noting that it corresponds to rows $(k-1)p+1$
to $(k-1)p+r$ of $\R_{\lu}(\PHI)$.

It thus remains to verify that the Jacobian of $H^{\ast}$ with respect
to $(\R_{\lu},\Lambda_{\lu})$ is nonsingular at $(\R_{0,\lu},\Lambda_{0,\lu};F_{0})$.
Matrix differentiation gives
\[
\deriv H^{\ast}=\begin{bmatrix}\R_{0,\lu}(\deriv\Lambda_{\lu})+(\deriv\R_{\lu})\Lambda_{0,\lu}-F_{0}(\deriv\R_{\lu}); & \L_{0,\lu}^{\trans}(\deriv\R_{\lu})\end{bmatrix}\eqdef\begin{bmatrix}\deriv H_{1}^{\ast}; & \deriv H_{2}^{\ast}\end{bmatrix}
\]
The Jacobian is nonsingular if $\deriv H^{\ast}=0$ implies $\deriv\R_{\lu}=0$
and $\deriv\Lambda_{\lu}=0$. To that end, suppose $\deriv H^{\ast}=0$.
Then $0=\deriv H_{2}^{\ast}=\L_{0,\lu}^{\trans}(\deriv\R_{\lu})$,
and
\[
\deriv\R_{\lu}=(\R_{0}\L_{0}^{\trans})\deriv\R_{\lu}=(\R_{0,\lu}\L_{0,\lu}^{\trans}+\R_{0,\st}\L_{0,\st}^{\trans})\deriv\R_{\lu}=(\R_{0,\st}\L_{0,\st}^{\trans})\deriv\R_{\lu}
\]
and similarly, by \eqref{Fsubs} above,
\[
F_{0}(\deriv\R_{\lu})=(\R_{0,\lu}\Lambda_{0,\lu}\L_{0,\lu}^{\trans}+\R_{0,\st}\Lambda_{0,\st}\L_{0,\st}^{\trans})\deriv\R_{\lu}=\R_{0,\st}\Lambda_{0,\st}\L_{0,\st}^{\trans}(\deriv\R_{\lu}).
\]
Hence
\begin{align*}
\deriv H_{1}^{\ast} & =\R_{0,\lu}(\deriv\Lambda_{\lu})+\R_{0,\st}[\L_{0,\st}^{\trans}(\deriv\R_{\lu})\Lambda_{0,\lu}-\Lambda_{0,\st}\L_{0,\st}^{\trans}(\deriv\R_{\lu})]\\
 & =\begin{bmatrix}\R_{0,\lu} & \R_{0,\st}\end{bmatrix}\begin{bmatrix}\deriv\Lambda_{\lu}\\
\mathcal{T}[\L_{0,\st}^{\trans}(\deriv\R_{\lu})]
\end{bmatrix},
\end{align*}
where $\mathcal{T}(M)\defeq M\Lambda_{0,\lu}-\Lambda_{0,\st}M.$
Since $\R_{0}$ is nonsingular, $\deriv H_{1}^{\ast}=0$ implies that
$\deriv\Lambda_{\lu}=0$ and $\mathcal{T}[\L_{0,\st}^{\trans}(\deriv\R_{\lu})]=0$;
but since $\Lambda_{0,\lu}$ and $\Lambda_{0,\st}$ have no common
eigenvalues, $\mathcal{T}(M)=0$ if and only if $M=0$ (\citealp{SS90},
Thm~V.1.3). Thus $\L_{0,\st}^{\trans}(\deriv\R_{\lu})=0$, whence
\[
\begin{bmatrix}\L_{0,\lu}^{\trans}\\
\L_{0,\st}^{\trans}
\end{bmatrix}\deriv\R_{\lu}=0
\]
from which it follows that $\deriv\R_{\lu}=0$, since $\L_{0}$ is
nonsingular.
\end{proof}

\begin{proof}[Proof of \lemref{derivatives}]
 \proofpart{1} We have
\[
R_{0,\lu}\Lambda_{0,\lu}^{k}-\sum_{i=1}^{k}\Phi_{i}R_{0,\lu}\Lambda_{0,\lu}^{k-i}=\PHI\R_{0,\lu}=_{(1)}\PHI_{0}\R_{0,\lu}=R_{0,\lu}\Lambda_{0,\lu}^{k}-\sum_{i=1}^{k}\Phi_{0,i}R_{0,\lu}\Lambda_{0,\lu}^{k-i}=_{(2)}0
\]
where $=_{(1)}$ is by hypothesis, and $=_{(2)}$ by \lemref{GLR}.
Since $\smlabs{\lambda_{q+1}(\PHI)}<\smlabs{\lambda_{q}(\PHI_{0})}=\smlabs{\lambda_{q}(\Lambda_{0,\lu})}$
and $\PHI\in\set P$, the result then follows by \lemref{GLR-converse}.

\proofpart{2} Analogously to \eqref{HHast} above, define
\begin{align*}
H(\R_{\lu},\Lambda_{\lu};\PHI) & \defeq\begin{bmatrix}\R_{\lu}\Lambda_{\lu}-F(\PHI)\R_{\lu}; & \G^{\trans}\R_{\lu}-I_{q}\end{bmatrix}\\
H^{\ast}(\R_{\lu},\Lambda_{\lu};\PHI) & \defeq\begin{bmatrix}\R_{\lu}\Lambda_{\lu}-F(\PHI)\R_{\lu}; & \L_{0,\lu}^{\trans}\R_{\lu}-I_{q}\end{bmatrix}.
\end{align*}
By the argument given in the proof of \lemref{implicit-maps}, there
are smooth maps $\R_{\lu}(\PHI)$, $\R_{\lu}^{\ast}(\PHI)$, $\Lambda_{\lu}(\PHI)$
and $\Lambda_{\lu}^{\ast}(\PHI)$ such that $H[\R_{\lu}(\PHI),\Lambda_{\lu}(\PHI);\PHI]=0$
and $H^{\ast}[\R_{\lu}^{\ast}(\PHI),\Lambda_{\lu}^{\ast}(\PHI);\PHI]=0$
for all $\PHI\in\set P$. Since $G^{\trans}R_{0,\lu}=I_{q}$ implies
that $\G^{\trans}\R_{0,\lu}=I_{q}$, we have $\R_{\lu}(\PHI)=\R_{\lu}^{\ast}(\PHI)=\R_{0,\lu}$
and $\Lambda_{\lu}(\PHI)=\Lambda_{\lu}^{\ast}(\PHI)=\Lambda_{0,\lu}$
when $\PHI=\PHI_{0}$, but otherwise these maps need not agree. Since
the maps $\R_{\lu}^{\ast}(\PHI)$ and $\Lambda_{\lu}^{\ast}(\PHI)$
are easier to work with, we first obtain the derivatives of these,
and subsequently those of $A(\PHI)$ and $\Lambda_{\lu}(\PHI)$ via
renormalisation, analogously to \eqref{renorm1}--\eqref{renorm2}.

Setting the total differential of $H^{\ast}$ to zero gives
\begin{equation}
0=\deriv H^{\ast}=\begin{bmatrix}\R_{0,\lu}(\deriv\Lambda_{\lu}^{\ast})+(\deriv\R_{\lu}^{\ast})\Lambda_{0,\lu}-F_{0}(\deriv\R_{\lu}^{\ast})-F(\deriv\PHI)\R_{0,\lu}; & \L_{0,\lu}^{\trans}(\deriv\R_{\lu}^{\ast})\end{bmatrix}\label{eq:dHastzero}
\end{equation}
where $F_{0}\defeq F(\PHI)$, whence by similar arguments as were
given in the proof of \lemref{implicit-maps},
\begin{equation}
F(\deriv\PHI)\R_{0,\lu}=\R_{0,\lu}(\deriv\Lambda_{\lu}^{\ast})+\R_{0,\st}\L_{0,\st}^{\trans}(\deriv\R_{\lu}^{\ast})\Lambda_{0,\lu}-\R_{0,\st}\Lambda_{0,\st}\L_{0,\st}^{\trans}(\deriv\R_{\lu}^{\ast}).\label{eq:totaldiff}
\end{equation}
Vectorising gives
\begin{align}
\vek[F(\deriv\PHI)\R_{0,\lu}] & =(I_{q}\otimes\R_{0,\lu})\vek(\deriv\Lambda_{\lu}^{\ast})+M\vek(\deriv\R_{\lu}^{\ast})\label{eq:vecFdP}
\end{align}
for $M\defeq(I_{q}\otimes\R_{0,\st})[(\Lambda_{0,\lu}^{\trans}\otimes I_{kp-q})-(I_{q}\otimes\Lambda_{0,\st})](I_{q}\otimes\L_{0,\st}^{\trans})$.
Since $\L_{0,\st}^{\trans}\R_{0,\lu}=0$ and $\L_{0,\st}^{\trans}\R_{0,\st}=I_{kp-q}$,
setting
\[
M^{\dagger}\defeq(I_{q}\otimes\R_{0,\st})[(\Lambda_{0,\lu}^{\trans}\otimes I_{kp-q})-(I_{q}\otimes\Lambda_{0,\st})]^{-1}(I_{q}\otimes\L_{0,\st}^{\trans})
\]
we have $M^{\dagger}(I_{q}\otimes\R_{0,\lu})=0$ and $M^{\dagger}M=I_{q}\otimes\R_{0,\st}\L_{0,\st}^{\trans}$.
Since $\L_{0,\lu}^{\trans}(\deriv\R_{\lu}^{\ast})=0$ by \eqref{dHastzero},
it follows that
\[
\deriv\R_{\lu}^{\ast}=(\R_{0,\lu}\L_{0,\lu}^{\trans}+\R_{0,\st}\L_{0,\st}^{\trans})\deriv\R_{\lu}^{\ast}=(\R_{0,\st}\L_{0,\st}^{\trans})\deriv\R_{\lu}^{\ast}
\]
whence $M^{\dagger}M\vek(\deriv\R_{\lu}^{\ast})=\vek(\deriv\R_{\lu}^{\ast})$,
and so premultiplying \eqref{vecFdP} by $M^{\dagger}$ gives
\begin{align*}
\vek(\deriv\R_{\lu}^{\ast}) & =M^{\dagger}\vek[F(\deriv\PHI)\R_{0,\lu}^{\ast}].
\end{align*}
By the structure of the companion form matrix, $\L_{0,\st}^{\trans}F(\deriv\PHI)\R_{0,\lu}=L_{0,\st}^{\trans}(\deriv\PHI)\R_{0,\lu}$.
Since $R$ is given by the final $p$ rows of $\R$, we have
\begin{align}
\vek(\deriv R_{\lu}^{\ast}) & =(I_{q}\otimes R_{0,\st})[(\Lambda_{0,\lu}^{\trans}\otimes I_{kp-q})-(I_{q}\otimes\Lambda_{0,\st})]^{-1}(I_{q}\otimes L_{0,\st}^{\trans})\vek\{(\deriv\PHI)\R_{0,\lu}\}\nonumber \\
 & =B(\PHI_{0})\vek\{(\deriv\PHI)\R_{0,\lu}\}.\label{eq:dRast}
\end{align}

To compute the Jacobian of $A(\PHI)$, note that by partitioning the
$p\times p$ identity matrix as
\[
\begin{bmatrix}G_{\perp} & G\end{bmatrix}\defeq\begin{bmatrix}I_{r} & 0\\
0 & I_{q}
\end{bmatrix}
\]
we have $A(\PHI)=G_{\perp}^{\trans}R_{\lu}(\PHI)=G_{\perp}^{\trans}R_{\lu}^{\ast}(\PHI)[G^{\trans}R_{\lu}^{\ast}(\PHI)]^{-1}$.
From $R_{\lu}^{\ast}(\PHI_{0})=R_{0,\lu}$, $G^{\trans}R_{0,\lu}=\G^{\trans}\R_{0,\lu}=I_{q}$
and $G_{\perp}^{\trans}R_{0,\lu}=A_{0}$, it follows that at $\PHI=\PHI_{0}$
\begin{align}
\deriv A & =G_{\perp}^{\trans}(\deriv R_{\lu}^{\ast})-(G_{\perp}^{\trans}R_{0,\lu})G^{\trans}(\deriv R_{\lu}^{\ast})=(G_{\perp}^{\trans}-A_{0}G^{\trans})\deriv R_{\lu}^{\ast}=\beta_{0}^{\trans}\deriv R_{\lu}^{\ast}\label{eq:dA}
\end{align}
for $\beta_{0}^{\trans}=[I_{r},-A_{0}]$. The first part of \eqref{Jacobs}
follows immediately from \eqref{dRast} and \eqref{dA}. For the Jacobian
of $\Lambda_{\lu}(\PHI)$, note that (as per \eqref{renorm2} above)
\[
\Lambda_{\lu}(\PHI)=[\G^{\trans}\R_{\lu}^{\ast}(\PHI)]\Lambda_{\lu}^{\ast}(\PHI)[\G^{\trans}\R_{\lu}^{\ast}(\PHI)]^{-1}
\]
whence at $\PHI=\PHI_{0}$,
\begin{align*}
\deriv\Lambda_{\lu} & =\G^{\trans}(\deriv\R_{\lu}^{\ast})\Lambda_{0,\lu}+\deriv\Lambda_{\lu}^{\ast}-\Lambda_{0,\lu}\G^{\trans}(\deriv\R_{\lu}^{\ast}).
\end{align*}
Recognising that $\G^{\trans}(\deriv\R_{\lu}^{\ast})=G^{\trans}(\deriv R_{\lu}^{\ast})$
and vectorising, we have
\begin{equation}
\vek(\deriv\Lambda_{\lu})=\{(\Lambda_{0,\lu}^{\trans}\otimes I_{q})-(I_{q}\otimes\Lambda_{0,\lu})\}(I_{q}\otimes G^{\trans})\vek(\deriv R_{\lu}^{\ast})+\vek(\deriv\Lambda_{\lu}^{\ast}).\label{eq:dLam}
\end{equation}
$\deriv R_{\lu}^{\ast}$ is given in \eqref{dRast} above. To obtain
$\deriv\Lambda_{\lu}^{\ast}$, note that premultiplying \eqref{totaldiff}
by $\L_{0,\lu}^{\trans}$ yields 
\begin{equation}
\deriv\Lambda_{\lu}^{\ast}=\L_{0,\lu}^{\trans}F(\deriv\PHI)\R_{0,\lu}=L_{0,\lu}^{\trans}(\deriv\PHI)\R_{0,\lu}.\label{eq:dLamast}
\end{equation}
Thus \eqref{dRast}, \eqref{dLam} and \eqref{dLamast} give the second
part of \eqref{Jacobs}.

\proofpart{3} Continuity of $J(\PHI)$ is immediate from $A(\PHI)$
and $\Lambda_{\lu}(\PHI)$ being smooth.
\end{proof}

\begin{proof}[Proof of \lemref{derivatunity}]
 The stated expression for $J(\PHI)$ is immediate from \eqref{Bdef},
\lemref{derivatives}, and $\Lambda_{\lu}(\PHI)=I_{q}$. That $J(\PHI)$
is nonsingular will follow once we have shown that the $(p\times p)$
matrix
\begin{equation}
K\defeq\begin{bmatrix}\beta^{\trans}R_{\st}(I_{kp-q}-\Lambda_{\st})^{-1}L_{\st}^{\trans}\\
L_{\lu}^{\trans}
\end{bmatrix}\label{eq:K-matrix}
\end{equation}
is nonsingular. We first note the following facts. Since $\PHI\in\set P$
with $\Lambda_{\lu}(\PHI)=I_{q}$, it follows from \eqref{AlambdaR}
that $\rank\Phi(1)\leq p-q$. Since $\Phi(\cdot)$ has exactly $q$
roots at unity, the reverse inequality holds by Corollary~4.3 of
\citet{Joh95}, whence $\rank\Phi(1)=p-q$. Thus \assref{J} holds:
this implies that $\spn\beta=\spn\Phi(1)^{\trans}$ and $\rank L_{\lu}=q$
(see \lemref{qcs} and the characterisation of the CS discussed in
\subsecref{coint-ur}).

Now let $c\in\reals^{p}$ be such that $Kc=0$, so that in particular
$L_{\lu}^{\trans}c=0$. Since $\rank\Phi(1)+\rank L_{\lu}=p$, while
\eqref{eig-eig} with $\Lambda_{\lu}=I_{q}$ implies $L_{\lu}^{\trans}\Phi(1)=0$,
it follows that $c\in\spn\Phi(1)$, i.e.\ $c=\Phi(1)b$ for some
$b\in\reals^{p}$. By \citet[Thm~2.4]{GLR82}, $\Phi(\mu)^{-1}=R(\mu I-\Lambda)^{-1}L^{\trans}$
for any $\mu$ not a root of $\Phi(\cdot)$. Since the columns of
the quasi-cointegrating matrix $\beta$ are orthogonal to $R_{\lu}$,
we have
\begin{equation}
\beta^{\trans}=\beta^{\trans}R_{\st}(\mu I_{kp-q}-\Lambda_{\st})^{-1}L_{\st}^{\trans}\Phi(\mu)\goesto\beta^{\trans}R_{\st}(I_{kp-q}-\Lambda_{\st})^{-1}L_{\st}^{\trans}\Phi(1)\label{eq:betastuff}
\end{equation}
by the continuity of the r.h.s., as $\mu\goesto1$, since $\Lambda_{\st}$
has no eigenvalues at unity. Hence
\[
0=Kc=\begin{bmatrix}\beta^{\trans}R_{\st}(I_{kp-q}-\Lambda_{\st})^{-1}L_{\st}^{\trans}\Phi(1)b\\
0
\end{bmatrix}=\begin{bmatrix}\beta^{\trans}b\\
0
\end{bmatrix}
\]
implying $\beta^{\trans}b=0$. But $\spn\beta=\spn\Phi(1)^{\trans}$,
so we must have $\Phi(1)b=0$. Thus $c=0$, from which it follows
that $K$ is nonsingular.
\end{proof}

\section{Asymptotics}

\label{app:asymptotics}

The assumptions \assref{DGP} and \assref{LOC} are maintained throughout
this appendix. We first recall some notation. Let $\PHI_{0}\defeq\lim_{n\goesto\infty}\PHI_{n}$,
where $\{\PHI_{n}\}$ is the sequence specified by \assref{LOC}.
Let $R_{n}\defeq[R_{\lu}(\PHI_{n}),R_{\st}]$ and $\Lambda_{n}\defeq\diag\{\Lambda_{n,\lu},\Lambda_{\st}\}$
be as in \assref{LOC}. Take $\R_{n}\defeq\col\{R_{n}\Lambda_{n}^{k-i}\}_{i=1}^{k}$
and $\L_{n}\defeq(\R_{n}^{\trans})^{-1}$ as in \lemref{GLR}, and
partition these as $\R_{n}=[\R_{n,\lu},\R_{n,\st}]$\label{app:RnSTdef}
and $\L_{n}=[\L_{n,\lu},\L_{n,\st}]$ (as per \eqref{partition});
note that both these matrices are convergent.

Let $z_{\lu,t}\defeq\L_{n,\lu}^{\trans}\x_{t}$ and $z_{\st,t}=\L_{n,\st}^{\trans}\x_{t}$
be as in \lemref{repasAR} (for $\PHI=\PHI_{n}$); these follow the
autoregressions given in \eqref{zproc}. Recall $E\sim\BM(\Sigma)$
and $Z_{C}(r)\defeq\int_{0}^{r}\e^{C(r-s)}L_{\lu}^{\trans}\deriv E(s)$
from \eqref{Zproc}. For $i\in\{\lu,\st\}$, let $\zdet_{i,t}$ denote
the residual from an OLS regression of $\{\zdet_{\lu,t-1}\}_{t=1}^{n}$
onto a constant and linear trend. Recall that $\Zdet_{C}$ denotes
the residual from an $L^{2}[0,1]$ projection of each sample path
of $Z_{C}$ onto a constant and linear trend. As in \subsecref{likelihood},
let $\hat{\Sigma}_{n}$ denote the unrestricted MLE for $\Sigma$,
i.e.\ the OLS residual variance matrix estimator.

Proofs of the following results appear at the end of this section.
\begin{lem}
\label{lem:wkconv}The following hold jointly:
\begin{enumerate}
\item \label{enu:wkconv:donsker}$n^{-1/2}\sum_{t=1}^{\smlfloor{nr}}\err_{t}\wkc E(r)$
\item $n^{-1/2}z_{\lu,\smlfloor{nr}}\wkc Z_{C}(r)$
\item \label{enu:wkconv:Zc}$n^{-1/2}\zdet_{\lu,\smlfloor{nr}}\wkc\Zdet_{C}(r)$
\end{enumerate}
as weak convergences on the space of right-continuous functions $[0,1]\setmap\reals^{m}$
(with respect to the uniform topology); and
\begin{enumerate}[resume]
\item \label{enu:wkconv:si}$n^{-1}\sum_{t=1}^{n}(\zdet_{\lu,t-1}\otimes\err_{t})\wkc\int_{0}^{1}[\Zdet_{C}(r)\otimes\deriv E(r)]\diff r$
\item \label{enu:wkconv:mgclt}$n^{-1/2}\sum_{t=1}^{n}(\zdet_{\st,t-1}\otimes\err_{t})\wkc\xi\sim\normdist[0,\Omega\otimes\Sigma]$
\item \label{enu:wkconv:varmat}$\hat{\Sigma}_{n}\inprob\Sigma$,
\end{enumerate}
where $\Omega\defeq\lim_{n\goesto\infty}\var(z_{\st,n})$ is positive
definite, and $\xi$ is independent of $E$.
\end{lem}
Now define the reparametrisation $\PHI\elmap\f$ by
\begin{equation}
\f\defeq\begin{bmatrix}\f_{\lu}\\
\f_{\st}
\end{bmatrix}=\begin{bmatrix}\vek\{(\PHI-\PHI_{n})\R_{n,\lu}\}\\
\vek\{(\PHI-\PHI_{n})\R_{n,\st}\}
\end{bmatrix}=\vek\{(\PHI-\PHI_{n})\R_{n}\},\label{eq:reparam}
\end{equation}
which is reversed by setting $\PHI=\PHI_{n}+\vek^{-1}(\f)\L_{n}^{\trans}$,
where $\vek^{-1}(x)$ maps $x\in\reals^{kp^{2}}$ to the matrix $X\in\reals^{p\times kp}$
for which $\vek(X)=x$. The parameter space for $\f$ is the open
set
\begin{equation}
\spcf_{n}\defeq\{\vek[(\PHI-\PHI_{n})\R_{n}]\mid\PHI\in\set P\},\label{eq:Psetn}
\end{equation}
and the true parameters correspond to $\f=0$. Although $\spcf_{n}$
depends on $n$, since $\PHI_{n}\goesto\PHI_{0}\in\set P$ and $\set P$
is open (\lemref{implicit-maps}), there is an $\epsilon>0$ such
that $\spcf_{n}$ contains a ball of radius $\epsilon$ centred at
the origin, for all $n$ sufficiently large. Let
\[
\likens_{n}(\f)\defeq\like_{n}[\PHI_{n}+\vek^{-1}(\f)\L_{n}^{\trans},\hat{\Sigma}_{n}].
\]
Define $D_{n}\defeq\diag\{nI_{\snum\lu},n^{1/2}I_{\snum\st}\}$, where
$\snum\lu\defeq pq$ and $\snum\st\defeq p(kp-q)$ correspond to the
dimensions of the vectors $\f_{\lu}$ and $\f_{\st}$ respectively.
\begin{lem}
\label{lem:lhoodexp}There exist $S_{n}$ and $H_{n}$ such that for
all $\f\in\spcf_{n}$,
\[
\likens_{n}(\f)-\likens(0)=S_{n}^{\trans}(D_{n}\f)-\tfrac{1}{2}(D_{n}\f)^{\trans}H_{n}(D_{n}\f)
\]
where
\begin{gather*}
S_{n}\wkc\begin{bmatrix}\int_{0}^{1}[\Zdet_{C}(r)\otimes\Sigma^{-1}\deriv E(r)]\\
\xi
\end{bmatrix}\eqdef\begin{bmatrix}S_{\lu}\\
S_{\st}
\end{bmatrix}\eqdef S\\
H_{n}\wkc\begin{bmatrix}\int\Zdet_{C}\Zdet_{C}^{\trans} & 0\\
0 & \Omega
\end{bmatrix}\otimes\Sigma^{-1}\eqdef\begin{bmatrix}H_{\lu} & 0\\
0 & H_{\st}
\end{bmatrix}\eqdef H,
\end{gather*}
for $\xi$ as in \lemref{wkconv}.
\end{lem}
Define the constraint maps
\begin{align}
\theta_{n}(\f) & \defeq\vek\{\Lambda_{\lu}[\PHI_{n}+\vek^{-1}(\f)\L_{n}^{\trans}]-(I_{q}+C/n)\}\label{eq:theta-n}\\
\gamma_{n}(\f) & \defeq a_{ij}[\PHI_{n}+\vek^{-1}(\f)\L_{n}^{\trans}]-a_{ij}(\PHI_{n}),\nonumber 
\end{align}
and the associated restricted parameter spaces
\begin{align*}
\spcf_{n\mid\theta} & \defeq\{\f\in\spcf_{n}\mid\theta_{n}(\f)=0\}\\
\spcf_{n\mid\theta,\gamma} & \defeq\{\f\in\spcf_{n}\mid\theta_{n}(\f)=0\text{ and }\gamma_{n}(\f)=0\}.
\end{align*}
Let $\hat{\f}_{n}$, $\hat{\f}_{n\mid\theta}$ and $\hat{\f}_{n\mid\theta,\gamma}$
denote exact maximisers of $\likens_{n}(\f)$ over the sets $\spcf_{n}$,
$\spcf_{n\mid\theta}$ and $\spcf_{n\mid\theta,\gamma}$ respectively:
which may be shown to exist at least with with probability approaching
one (w.p.a.1), and may be arbitrarily defined otherwise.
\begin{lem}
\label{lem:consistency} Each of $D_{n}\hat{\f}_{n}$, $D_{n}\hat{\f}_{n\mid\theta}$
and $D_{n}\hat{\f}_{n\mid\theta,\gamma}$ are $O_{p}(1)$.
\end{lem}
Let $\grad_{\f}g(\f_{0})$ denote the gradient of $g:\spcf\setmap\reals^{d_{g}}$
at $\f=\f_{0}$. The derivatives of the maps $\theta_{n}$ and $\gamma_{n}$
can be inferred from \lemref{derivatives}. Part~\enuref{deriv:values}
of that result gives the derivatives with respect to $\f_{\lu}$,
and part~\enuref{deriv:zero} implies that when $\f_{\lu}=0$, the
first (and higher order) derivatives with respect to $\f_{\st}$ are
identically zero. Now letting $e_{d,i}\in\reals^{d}$ denote a vector
with zero everywhere except for a $1$ in the $i$th position, define
\[
\Pi\defeq[\begin{matrix}\Theta; & \Gamma\end{matrix}]\defeq[\begin{matrix}I_{q}\otimes L_{\lu}; & e_{q,j}\otimes L_{\st}(I_{kp-q}-\Lambda_{\st}^{\trans})^{-1}R_{\st}^{\trans}\beta e_{r,i}\end{matrix}],
\]
which by \lemref{derivatunity} has full column rank, and
\begin{align*}
\THETA & \defeq\begin{bmatrix}\Theta\\
0_{\#\st\times q^{2}}
\end{bmatrix} & \PI & \defeq\begin{bmatrix}\Pi\\
0_{\#\st\times(q^{2}+1)}
\end{bmatrix}.
\end{align*}

\begin{lem}
\label{lem:deriv-limits}~
\begin{enumerate}
\item \label{enu:grad-lim}Let $\{\tilde{\f}_{n}\}$ denote a random sequence
in $\spcf_{n}$ with $\tilde{\f}_{n}\inprob0$. Then
\begin{align*}
\grad_{\f}\theta_{n}(\tilde{\f}_{n}) & \inprob\THETA & \grad_{\f}\gamma_{n}(\tilde{\f}_{n}) & \inprob\GAMMA.
\end{align*}
\item \label{enu:proj-lim}Let $\proj_{\THETA,\perp}$ and $\proj_{\PI,\perp}$
denote orthogonal projections from $\reals^{kp^{2}}$ onto the subspaces
orthogonal to the the columns of $\THETA$ and $\PI$. Then
\begin{align*}
D_{n}\hat{\f}_{n\mid\theta} & =\proj_{\THETA,\perp}D_{n}\hat{\f}_{n\mid\theta}+o_{p}(1)\\
D_{n}\hat{\f}_{n\mid\theta,\gamma} & =\proj_{\PI,\perp}D_{n}\hat{\f}_{n\mid\theta,\gamma}+o_{p}(1).
\end{align*}
\end{enumerate}
\end{lem}
Let $\Theta_{\perp}\in\reals^{pq\times qr}$ and $\Pi_{\perp}\in\reals^{pq\times(qr-1)}$
denote matrices having full column rank, such that $\Theta_{\perp}^{\trans}\Theta=0$
and $\Pi_{\perp}^{\trans}\Pi=0$. We may take $\Theta_{\perp}=I_{q}\otimes L_{\lu,\perp}$,
for $L_{\lu,\perp}$ a $p\times r$ matrix having rank $r$ and for
which $L_{\lu,\perp}^{\trans}L_{\lu}=0$. Since $\Pi=[\Theta,\Gamma]$
there exists a full column rank matrix $\Xi\in\reals^{qr\times(qr-1)}$
for which $\Pi_{\perp}\defeq\Theta_{\perp}\Xi$.
\begin{prop}
\label{prop:andrews}~
\begin{enumerate}
\item \label{enu:andrews:unres}$D_{n}\hat{\f}_{n}=\begin{bmatrix}n\hat{\f}_{n,\lu}\\
n^{1/2}\hat{\f}_{n,\st}
\end{bmatrix}\wkc\begin{bmatrix}H_{\lu}^{-1}S_{\lu}\\
H_{\st}^{-1}S_{\st}
\end{bmatrix}$,
\item \label{enu:andrews:res}$D_{n}\hat{\f}_{n\mid\theta}=\begin{bmatrix}n\hat{\f}_{n,\lu\mid\theta}\\
n^{1/2}\hat{\f}_{n,\st\mid\theta}
\end{bmatrix}\wkc\begin{bmatrix}\Theta_{\perp}(\Theta_{\perp}^{\trans}H_{\lu}\Theta_{\perp})^{-1}\Theta_{\perp}^{\trans}S_{\lu}\\
H_{\st}^{-1}S_{\st}
\end{bmatrix}$,
\item \label{enu:andrews:lrroot}$2[\likens_{n}(\hat{\f}_{n})-\likens_{n}(\hat{\f}_{n\mid\theta})]\wkc S_{\lu}^{\trans}H_{\lu}^{-1}\Theta(\Theta^{\trans}H_{\lu}^{-1}\Theta)^{-1}\Theta^{\trans}H_{\lu}^{-1}S_{\lu}$.
\end{enumerate}
Let $H_{\Theta,\perp}\defeq\Theta_{\perp}^{\trans}H_{\lu}\Theta_{\perp}$,
and $\mathcal{Q}\in\reals^{qr\times qr}$ denote the orthogonal projection
onto $\spn H_{\Theta,\perp}^{1/2}\Xi$. Then
\begin{enumerate}[resume]
\item \label{enu:andrews:coef}$2[\likens_{n}(\hat{\f}_{n\mid\theta})-\likens_{n}(\hat{\f}_{n\mid\theta,\gamma})]\wkc(H_{\Theta,\perp}^{-1/2}\Theta_{\perp}^{\trans}S_{\lu})^{\trans}[I_{qr}-\mathcal{Q}](H_{\Theta,\perp}^{-1/2}\Theta_{\perp}^{\trans}S_{\lu}).$
\end{enumerate}
\end{prop}
The preceding gives the limiting distribution of $\hat{\PHI}_{n}$
under the reparametrisation \eqref{reparam}; the limiting distributions
of estimators of $A$ and $\Lambda_{\lu}$ will then follow by an
application of the delta method, as per
\begin{prop}
\label{prop:deltamethod}Let $\{\PHI_{n}\}$ be as in \assref{LOC},
$\PHI_{0}\defeq\lim_{n\goesto\infty}\PHI_{n}\in\set P$, and $\{\tilde{\PHI}_{n}\}$
a random sequence in $\set P$ with $\tilde{\PHI}_{n}=\PHI_{n}+o_{p}(1)$.
Then
\begin{equation}
\begin{bmatrix}\vek\{A(\tilde{\PHI}_{n})-A(\PHI_{n})\}\\
\vek\{\Lambda_{\lu}(\tilde{\PHI}_{n})-\Lambda_{\lu}(\PHI_{n})\}
\end{bmatrix}=\left(\begin{bmatrix}J_{A}(\PHI_{0})\\
J_{\Lambda}(\PHI_{0})
\end{bmatrix}+o_{p}(1)\right)\vek\{(\tilde{\PHI}_{n}-\PHI_{n})\R_{n,\lu}\}\label{eq:expansion}
\end{equation}
where $\R_{n,\lu}\defeq\R_{\lu}(\PHI_{n})$.
\end{prop}
\begin{proof}[Proof of \lemref{wkconv}]
 \enuref{wkconv:donsker}--\enuref{wkconv:si} follow by Donsker's
theorem for partial sums, Lemma~3.1 in \citet{Phi88Ecta} and the
continuous mapping theorem; \enuref{wkconv:mgclt} by the martingale
central limit theorem (\citealp[Thm.~3.2]{HH80}); and \enuref{wkconv:varmat}
by arguments similar to those given in Section~3.2.2 of \citet{Lut07}.
\end{proof}

\begin{proof}[Proof of \lemref{lhoodexp}]
 Let $\Phi_{i}\defeq\PHI\R_{n,i}$ and $\Phi_{n,i}\defeq\PHI_{n}\R_{n,i}$
for $i\in\{\lu,\st\}$. Then
\begin{align*}
\like_{n}(\PHI,\Sigma) & =-\frac{n}{2}\log\det\Sigma-\min_{m,d}\frac{1}{2}\sum_{t=1}^{n}\norm{y_{t}-m-dt-\PHI\y_{t-1}}_{\Sigma^{-1}}^{2}\\
 & =-\frac{n}{2}\log\det\Sigma-\min_{m,d}\frac{1}{2}\sum_{t=1}^{n}\norm{x_{t}-m-dt-\PHI\x_{t-1}}_{\Sigma^{-1}}^{2}\\
 & =-\frac{n}{2}\log\det\Sigma-\min_{m,d}\frac{1}{2}\sum_{t=1}^{n}\norm{x_{t}-m-dt-\Phi_{\lu}z_{\lu,t-1}-\Phi_{\st}z_{\st,t-1}}_{\Sigma^{-1}}^{2}\\
 & =-\frac{n}{2}\log\det\Sigma-\frac{1}{2}\sum_{t=1}^{n}\norm{\xdet_{t}-\Phi_{\lu}\zdet_{\lu,t-1}-\Phi_{\st}\zdet_{\st,t-1}}_{\Sigma^{-1}}^{2}
\end{align*}
Twice differentiating the r.h.s.\ (as in \citealt[Sec.~3.4]{Lut07})
with respect to $\Phi_{\lu}$ and $\Phi_{\st}$, and noting that $\f_{i}=\vek(\Phi_{i}-\Phi_{n,i})$,
we thus have
\begin{align*}
\likens_{n}(\f)-\likens_{n}(0)=\like_{n}(\PHI,\hat{\Sigma}_{n})-\like_{n}(\PHI_{n},\hat{\Sigma}_{n}) & =S_{n}^{\trans}(D_{n}\f)-\tfrac{1}{2}(D_{n}\f)^{\trans}H_{n}(D_{n}\f)
\end{align*}
where
\begin{gather*}
S_{n}\defeq\begin{bmatrix}n^{-1}\sum_{t=1}^{n}(\zdet_{\lu,t-1}\otimes\hat{\Sigma}_{n}^{-1}\errdet_{t})\\
n^{-1/2}\sum_{t=1}^{n}(\zdet_{\st,t-1}\otimes\hat{\Sigma}_{n}^{-1}\errdet_{t})
\end{bmatrix}=_{(1)}\begin{bmatrix}\frac{1}{n}\sum_{t=1}^{n}(\zdet_{\lu,t-1}\otimes\hat{\Sigma}_{n}^{-1}\err_{t})\\
\frac{1}{n^{1/2}}\sum_{t=1}^{n}(\zdet_{\st,t-1}\otimes\hat{\Sigma}_{n}^{-1}\err_{t})
\end{bmatrix}\\
H_{n}\defeq\begin{bmatrix}n^{-2}\sum_{t=1}^{n}\zdet_{\lu,t-1}\zdet_{\lu,t-1}^{\trans} & n^{-3/2}\sum_{t=1}^{n}\zdet_{\lu,t-1}\zdet_{\st,t-1}^{\trans}\\
n^{-3/2}\sum_{t=1}^{n}\zdet_{\st,t-1}\zdet_{\lu,t-1}^{\trans} & n^{-1}\sum_{t=1}^{n}\zdet_{\st,t-1}\zdet_{\st,t-1}^{\trans}
\end{bmatrix}\otimes\hat{\Sigma}_{n}^{-1},
\end{gather*}
and $\errdet_{t}$ denotes the residual from an OLS regression of
$\{\err_{t}\}_{t=1}^{n}$ on a constant and a linear trend; $=_{(1)}$
holds because each element of $\zdet_{\lu,t-1}$ and $\zdet_{\st,t-1}$
is orthogonal to a constant and linear trend. The stated convergences
of $S_{n}$ and $H_{n}$ then follow by \lemref{wkconv} and the continuous
mapping theorem.
\end{proof}
\begin{proof}[Proof of \lemref{consistency}]
 By \lemref{lhoodexp}, we have
\begin{align*}
\likens_{n}(\f)-\likens_{n}(0) & \leq\smlnorm{D_{n}\f}[\smlnorm{S_{n}}-\tfrac{1}{2}\lambda_{\min}(H_{n})\smlnorm{D_{n}\f}].
\end{align*}
Let $M<\infty$ and $\epsilon>0$. Since $D_{n}=\diag\{nI_{\snum\lu},n^{1/2}I_{\snum\st}\}$,
$S_{n}=O_{p}(1)$ and $H_{n}\wkc H$ is positive definite w.p.a.1,
it is evident that
\begin{align*}
\Prob\left\{ \sup_{\{\f\in\spcf_{n}\mid\smlnorm{D_{n}\f}\geq M\}}[\likens_{n}(\f)-\likens_{n}(0)]<-\epsilon\right\}  & \ge\Prob\left\{ M[\smlnorm{S_{n}}-\tfrac{1}{2}\lambda_{\min}(H_{n})M]<-\epsilon\right\} 
\end{align*}
and
\begin{align*}
\limsup_{n\goesto\infty}\Prob\left\{ M[\smlnorm{S_{n}}-\tfrac{1}{2}\lambda_{\min}(H_{n})M]<-\epsilon\right\}  & \geq\Prob\left\{ M[\smlnorm S-\tfrac{1}{2}\lambda_{\min}(H)M]<-\epsilon\right\} \goesto1
\end{align*}
as $M\goesto\infty$. Deduce that $D_{n}\hat{\f}_{n}=O_{p}(1)$. Since
$\spcf_{n\mid\theta,\gamma}\subset\spcf_{n\mid\theta}\subset\spcf_{n}$
and $0\in\spcf_{n\mid\theta,\gamma}$, that $D_{n}\hat{\f}_{n\mid\theta}$
and $D_{n}\hat{\f}_{n\mid\theta,\gamma}$ are stochastically bounded
follows by the same argument.
\end{proof}
\begin{proof}[Proof of \lemref{deriv-limits}]
 \proofpart{1} Since $\PHI_{n}\goesto\PHI_{0}$, $\L_{n}\goesto\L_{0}$
and $\Lambda_{\lu}(\cdot)$ is continuously differentiable (\lemref{implicit-maps}),
\[
\grad_{\f}\theta_{n}(\tilde{\f}_{n})\inprob\grad_{\f}\vek\{\Lambda_{\lu}[\PHI_{0}+\vek^{-1}(\f)\L_{0}^{\trans}]\}|_{\f=0}=_{(1)}\begin{bmatrix}I_{q}\otimes L_{\lu}\\
0_{\#\st\times q^{2}}
\end{bmatrix}=\THETA
\]
where $=_{(1)}$ follows by Lemmas~\ref{lem:derivatives} and \ref{lem:derivatunity}.
The probability limit of $\grad_{\f}\gamma_{n}(\tilde{\f}_{n})$ follows
similarly.

\proofpart{2} By \lemref{consistency} and the remarks following
\eqref{Psetn}, there exists a ball $B(0,\epsilon)$ of radius $\epsilon>0$,
centred on the origin, such that $B(0,\epsilon)\subset\spcf_{n}$
for all $n$ sufficiently large, and $\Prob\{\hat{\f}_{n\mid\theta}\in B(0,\epsilon)\}\goesto1$.
We may take $\epsilon$ sufficiently small that $\PHI_{\f}\defeq\PHI_{n}+\vek^{-1}(\f)\L_{n}^{\trans}$
has $\smlabs{\lambda_{q+1}(\PHI_{\f})}<\smlabs{\lambda_{q}(\PHI_{n})}$
for all $n$ sufficiently large, for all $\f\in B(0,\epsilon)$. In
particular, suppose $\f_{\lu}=0$; then $(\PHI_{\f}-\PHI_{n})\R_{n,\lu}=0$
and we have by \lemref{derivatives}\enuref{deriv:zero} that $\Lambda_{\lu}(\PHI_{\f})=\Lambda_{\lu}(\PHI_{n})=C/n$.
It follows that $\theta_{n}(0,\hat{\f}_{n,\st\mid\theta})=0$ w.p.a.1.,
whence
\begin{multline*}
0=\theta_{n}(\hat{\f}_{n,\lu\mid\theta},\hat{\f}_{n,\st\mid\theta})=\theta_{n}(\hat{\f}_{n,\lu\mid\theta},\hat{\f}_{n,\st\mid\theta})-\theta_{n}(0,\hat{\f}_{n,\st\mid\theta})\\
=[\Theta+o_{p}(1)]^{\trans}\hat{\f}_{n,\lu\mid\theta}=\Theta^{\trans}\hat{\f}_{n,\lu\mid\theta}+o_{p}(\smlnorm{\hat{\f}_{n,\lu\mid\theta}})
\end{multline*}
by part~(i) of the lemma and a mean value expansion. Hence, letting
$\proj_{\Theta}$ and $\proj_{\Theta,\perp}$ denote the matrices
that orthogonally project from $\reals^{\#\lu}$ onto $\spn\Theta$
and $(\spn\Theta)^{\perp}$ respectively, we have
\begin{align*}
D_{n}\hat{\f}_{n\mid\theta} & =\begin{bmatrix}nI_{\#\lu} & 0\\
0 & n^{1/2}I_{\#\st}
\end{bmatrix}\begin{bmatrix}\proj_{\Theta}+\proj_{\Theta,\perp} & 0\\
0 & I_{\#\st}
\end{bmatrix}\begin{bmatrix}\hat{\f}_{n,\lu\mid\theta}\\
\hat{\f}_{n,\st\mid\theta}
\end{bmatrix}\\
 & =\begin{bmatrix}\proj_{\Theta,\perp} & 0\\
0 & I_{\#\st}
\end{bmatrix}\begin{bmatrix}n\hat{\f}_{n,\lu\mid\theta}\\
n^{1/2}\hat{\f}_{n,\st\mid\theta}
\end{bmatrix}+o_{p}(n\smlnorm{\hat{\f}_{n,\lu\mid\theta}})=\proj_{\THETA,\perp}D_{n}\hat{\f}_{n\mid\theta}+o_{p}(\smlnorm{D_{n}\hat{\f}_{n\mid\theta}}).\qedhere
\end{align*}
\end{proof}

\begin{proof}[Proof of \propref{andrews}]
 \proofpart{1} Immediate from \lemref{lhoodexp}.

\proofpart{2} As in the proof of \lemref{deriv-limits}\enuref{proj-lim},
we may take $\epsilon>0$ such that $B(0,\epsilon)\subset\spcf_{n}$
for all $n$ sufficiently large, and $\Prob\{\hat{\f}_{n\mid\theta}\in B(0,\epsilon)\}\goesto1$.
Hence w.p.a.1., $\hat{\f}_{n\mid\theta}$ satisfies the first-order
conditions for a constrained interior maximum,
\[
\grad_{\f}\like_{n}^{\ast}(\hat{\f}_{n\mid\theta})=D_{n}S_{n}-D_{n}H_{n}(D_{n}\hat{\f}_{n\mid\theta})=\grad_{\f}\theta_{n}(\hat{\f}_{n\mid\theta})\mu_{n},
\]
where $\mu_{n}\in\reals^{q^{2}}$ is a vector of Lagrange multipliers;
whence
\begin{equation}
S_{n}-H_{n}(D_{n}\hat{\f}_{n\mid\theta})=(nD_{n}^{-1})\grad_{\f}\theta_{n}(\hat{\f}_{n\mid\theta})(n^{-1}\mu_{n})\eqdef\THETA_{n}(n^{-1}\mu_{n})\label{eq:FOCconst}
\end{equation}
w.p.a.1. By a similar argument as given in the proof of \lemref{deriv-limits}\enuref{proj-lim},
it follows from \lemref{derivatives}\enuref{deriv:zero} that $\grad_{\f_{\st}}\theta_{n}(0,\hat{\f}_{n,\st\mid\theta})=0$
w.p.a.1, and so by a a mean value expansion and \lemref{consistency},
\[
\grad_{\f_{\st}}\theta_{n}(\hat{\f}_{n\mid\theta})=\grad_{\f_{\st}}\theta_{n}(\hat{\f}_{n,\lu\mid\theta},\hat{\f}_{n,\st\mid\theta})-\grad_{\f_{\st}}\theta_{n}(0,\hat{\f}_{n,\st\mid\theta})=O_{p}(\smlnorm{\hat{\f}_{n,\lu\mid\theta}})=O_{p}(n^{-1}).
\]

Deduce from the preceding and \lemref{deriv-limits}\enuref{grad-lim}
that
\[
\THETA_{n}=(nD_{n}^{-1})\grad_{\f}\theta_{n}(\hat{\f}_{n\mid\theta})=\begin{bmatrix}\grad_{\f_{\lu}}\theta_{n}(\hat{\f}_{n\mid\theta})\\
n^{1/2}\grad_{\f_{\st}}\theta_{n}(\hat{\f}_{n\mid\theta})
\end{bmatrix}\inprob\THETA,
\]
which has full column rank. Let $\THETA_{\perp}\defeq\diag\{\Theta_{\perp},I_{\#\st}\}$,
a full column rank matrix for which $\THETA_{\perp}^{\trans}\THETA=0$;
then $\THETA_{n,\perp}\defeq[I_{kp^{2}}-\THETA_{n}(\THETA_{n}^{\trans}\THETA_{n})^{-1}\THETA_{n}^{\trans}]\THETA_{\perp}\inprob\THETA_{\perp}$
and $\THETA_{n,\perp}^{\trans}\THETA_{n}=0$ for all $n$. Hence w.p.a.1
\begin{align*}
0 & =_{(1)}\THETA_{n,\perp}^{\trans}S_{n}-\THETA_{n,\perp}^{\trans}H_{n}(D_{n}\hat{\f}_{n\mid\theta})\\
 & =_{(2)}\THETA_{n,\perp}^{\trans}S_{n}-\THETA_{n,\perp}^{\trans}H_{n}[\THETA_{\perp}(\THETA_{\perp}^{\trans}\THETA_{\perp})^{-1}\THETA_{\perp}^{\trans}(D_{n}\hat{\f}_{n\mid\theta})+o_{p}(\smlnorm{D_{n}\hat{\f}_{n\mid\theta}})]
\end{align*}
where $=_{(1)}$ follows from premultiplying \eqref{FOCconst} by
$\THETA_{n,\perp}^{\trans}$, and $=_{(2)}$ from \lemref{deriv-limits}\enuref{proj-lim}.
A further appeal to that result and rearranging the preceding yields
\[
D_{n}\hat{\f}_{n\mid\theta}=\proj_{\THETA,\perp}D_{n}\hat{\f}_{n\mid\theta}+o_{p}(\smlnorm{D_{n}\hat{\f}_{n\mid\theta}})=\THETA_{\perp}(\THETA_{n,\perp}^{\trans}H_{n}\THETA_{\perp})^{-1}\THETA_{n,\perp}^{\trans}S_{n}+o_{p}(1+\smlnorm{D_{n}\hat{\f}_{n\mid\theta}}).
\]
The result then follows by Lemmas \ref{lem:lhoodexp} and \ref{lem:consistency}.

\proofpart{3} From parts~(i) and (ii) and \lemref{lhoodexp} we
have
\begin{gather}
2[\likens_{n}(\hat{\f}_{n})-\like_{n}^{\ast}(0)]\wkc S_{\lu}^{\trans}H_{\lu}^{-1}S_{\lu}+S_{\st}^{\trans}H_{\st}^{-1}S_{\st}\label{eq:lr-uncon}\\
2[\likens_{n}(\hat{\f}_{n\mid\theta})-\like_{n}^{\ast}(0)]\wkc S_{\lu}^{\trans}\Theta_{\perp}(\Theta_{\perp}^{\trans}H_{\lu}\Theta_{\perp})^{-1}\Theta_{\perp}^{\trans}S_{\lu}+S_{\st}^{\trans}H_{\st}^{-1}S_{\st}\label{eq:lr-con1}
\end{gather}
whence the result follows by subtracting \eqref{lr-con1} from \eqref{lr-uncon}
and noting that
\[
H_{\lu}^{-1/2}\Theta(\Theta^{\trans}H_{\lu}^{-1}\Theta)^{-1}\Theta^{\trans}H_{\lu}^{-1/2}+H_{\lu}^{1/2}\Theta_{\perp}(\Theta_{\perp}^{\trans}H_{\lu}\Theta_{\perp})^{-1}\Theta_{\perp}^{\trans}H_{\lu}^{1/2}=I_{pq}
\]
since the columns of $H_{\lu}^{-1/2}\Theta$ and $H_{\lu}^{1/2}\Theta_{\perp}$
are mutually orthogonal, and collectively span the whole of $\reals^{pq}$.

\proofpart{4} The same argument as which yielded \eqref{lr-con1}
also gives
\begin{equation}
2[\likens_{n}(\hat{\f}_{n\mid\theta,\gamma})-\like_{n}^{\ast}(0)]\wkc S_{\lu}^{\trans}\Pi_{\perp}(\Pi_{\perp}^{\trans}H_{\lu}\Pi_{\perp})^{-1}\Pi_{\perp}^{\trans}S_{\lu}+S_{\st}^{\trans}H_{\st}^{-1}S_{\st}\label{eq:lr-con2}
\end{equation}
so that subtracting \eqref{lr-con2} from \eqref{lr-con1}, and recalling
$\Pi_{\perp}=\Theta_{\perp}\Xi$, yields
\begin{align*}
2[\likens_{n}(\hat{\f}_{n\mid\theta})-\likens_{n}(\hat{\f}_{n\mid\theta,\gamma})] & \wkc S_{\lu}^{\trans}\Theta_{\perp}(\Theta_{\perp}^{\trans}H_{\lu}\Theta_{\perp})^{-1}\Theta_{\perp}^{\trans}S_{\lu}^{\trans}-S_{\lu}^{\trans}\Pi_{\perp}(\Pi_{\perp}^{\trans}H_{\lu}\Pi_{\perp})^{-1}\Pi_{\perp}^{\trans})S_{\lu}\\
 & =(\Theta_{\perp}^{\trans}S_{\lu})^{\trans}[H_{\Theta,\perp}^{-1}-\Xi(\Xi^{\trans}H_{\Theta,\perp}\Xi)^{-1}\Xi^{\trans}](\Theta_{\perp}^{\trans}S_{\lu})\\
 & =(H_{\Theta,\perp}^{-1/2}\Theta_{\perp}^{\trans}S_{\lu})^{\trans}[I_{qr}-H_{\Theta,\perp}^{1/2}\Xi(\Xi^{\trans}H_{\Theta,\perp}\Xi)^{-1}\Xi^{\trans}H_{\Theta,\perp}^{1/2}](H_{\Theta,\perp}^{-1/2}\Theta_{\perp}^{\trans}S_{\lu}).
\end{align*}
\end{proof}

\begin{proof}[Proof of \propref{deltamethod}]
 Recall the definitions of $\R_{n}=[\R_{n,\lu},\R_{n,\st}]$ and
$\L_{n}=[\L_{n,\lu},\L_{n,\st}]$ given at the beginning of this appendix.
Since $I_{kp}=\R_{n,\lu}\L_{n,\lu}^{\trans}+\R_{n,\st}\L_{n,\st}^{\trans}$,
we may write
\[
\tilde{\PHI}_{n}=\PHI_{n}+[(\tilde{\PHI}_{n}-\PHI_{n})\R_{n,\lu}]\L_{n,\lu}^{\trans}+[(\tilde{\PHI}_{n}-\PHI_{n})\R_{n,\st}]\L_{n,\st}^{\trans}\eqdef\PHI_{n}+\tilde{\Delta}_{n,\lu}+\tilde{\Delta}_{n,\st}.
\]
Since $\tilde{\Delta}_{n,\lu}=o_{p}(1)$ and $\PHI_{n}\goesto\PHI_{0}$,
we have $\smlabs{\lambda_{q+1}(\PHI_{n}+\tilde{\Delta}_{n,\st})}<\smlabs{\lambda_{q}(\PHI_{n})}$
w.p.a.1, and so by \lemref{derivatives}\enuref{deriv:zero}
\begin{align}
A(\tilde{\PHI}_{n})-A(\PHI_{n}) & =A(\PHI_{n}+\tilde{\Delta}_{n,\st}+\tilde{\Delta}_{n,\lu})-A(\PHI_{n}+\tilde{\Delta}_{n,\st})\label{eq:Adifference}
\end{align}
w.p.a.1. Since $A(\cdot)$ is smooth, a second-order Taylor series
expansion and \lemref{derivatives}\enuref{deriv:values} yield
\begin{align}
 & \vek\{A(\PHI_{n}+\tilde{\Delta}_{n,\st}+\tilde{\Delta}_{n,\lu})-A(\PHI_{n}+\tilde{\Delta}_{n,\st})\}\nonumber \\
 & \qquad\qquad\qquad\qquad=[J_{A}(\PHI_{n}+\tilde{\Delta}_{n,\st})+o_{p}(1)]\vek\{\tilde{\Delta}_{n,\lu}\R_{\lu}(\PHI_{n}+\tilde{\Delta}_{n,\st})\}\nonumber \\
 & \qquad\qquad\qquad\qquad=[J_{A}(\PHI_{0})+o_{p}(1)]\vek\{\tilde{\Delta}_{n,\lu}\R_{n,\lu}\}\label{eq:Aexp}
\end{align}
where the second equality holds w.p.a.1, and follows from the continuity
of $J_{A}$ (\lemref{derivatives}\enuref{deriv:cont}), $\PHI_{n}+\tilde{\Delta}_{n,\st}=\PHI_{0}+o_{p}(1)$,
and $\R_{\lu}(\PHI_{n}+\tilde{\Delta}_{n,\st})=\R_{\lu}(\PHI_{n})=\R_{n,\lu}$
(w.p.a.1, as implied by \lemref{derivatives}\enuref{deriv:zero}).
Finally, since 
\begin{equation}
\tilde{\Delta}_{n,\lu}\R_{n,\lu}=[(\tilde{\PHI}_{n}-\PHI_{n})\R_{n,\lu}]\L_{n,\lu}^{\trans}\R_{n,\lu}=(\tilde{\PHI}_{n}-\PHI_{n})\R_{n,\lu},\label{eq:DeltaRlu}
\end{equation}
the first part of \eqref{expansion} follows from \eqref{Adifference}--\eqref{DeltaRlu}.
The proof of the second part is analogous.
\end{proof}

\section{Limiting experiments}

The assumptions \assref{DGP} and \assref{LOC} are maintained throughout
this appendix. Recall the re-parametrisation given in \eqref{reparm-thm}
above, which in view of \eqref{reparam} we can equivalently write
as\begin{subequations}\label{eq:reparm}
\begin{align}
\vall & \defeq n\vek\begin{bmatrix}A[\PHI_{n}+\vek^{-1}(\f)\L_{n}^{\trans}]-A(\PHI_{n})\\
\Lambda_{\lu}[\PHI_{n}+\vek^{-1}(\f)\L_{n}^{\trans}]-\Lambda_{\lu}(\PHI_{n})
\end{bmatrix}\label{eq:vall}\\
f & \defeq n^{1/2}\f_{\st}.\label{eq:fST}
\end{align}
\end{subequations}Under \assref{LOC}, $R_{n,\st}$ and $\Lambda_{n,\st}$
associated with $\{\PHI_{n}\}\subset\set P$ are constant (see \enuref{LOC:stat}),
so $\R_{n,\st}=\R_{0,\st}$ for all $n\in\naturals$, so that in particular
$\f_{\st}=\vek\{(\PHI-\PHI_{n})\R_{n,\st}\}$. Let $\psi_{n}(\f)$
denote the smooth mapping $\f\elmap(\vall,f)$ implied by \eqref{reparm},
which has domain $\pset_{n}$ (defined in \eqref{Psetn} above) and
$\psi_{n}(0)=0$ for all $n\in\naturals$.
\begin{lem}
\label{lem:invpsi}~
\begin{enumerate}
\item \label{enu:invpsi:forward}There exists an $n_{0}\in\naturals$ and
an open neighbourhood $N\subset\reals^{kp^{2}}$ of the origin, such
that $\psi_{n}$ is a smooth diffeomorphism on $N$, for all $n\geq n_{0}$
\item \label{enu:invpsi:back}Let ${\cal K}\subset\reals^{kp^{2}}$ be any
compact neighbourhood of zero. Then there exists an $n_{1}\geq n_{0}$
such that $\psi_{n}^{-1}$ is well defined (and smooth) on ${\cal K}$,
for all $n\geq n_{1}$. Moreover, for any $(\vall,f)\in{\cal K}$,
$\f_{n}\defeq\psi_{n}^{-1}(\vall,f)$ is such that $D_{n}\f_{n}=O(1)$. 
\end{enumerate}
\end{lem}
Thus so long as we restrict attention to $\f\in N$, we may equivalently
parametrise the model in terms of $(\vall,f)$. For a given $(\vall,f)\in\reals^{kp^{2}}$,
$\psi_{n}^{-1}$ is well-defined (and smooth) at $(\vall,f)$ for
all $n$ sufficiently large, in which case we shall define (with a
slight abuse of notation) $\like_{n}(\vall,f)\defeq\like_{n}(\f,\Sigma)$,
where $\f=\psi_{n}^{-1}(\vall,f)$; and set $\like_{n}(\vall,f)\defeq-\infty$
otherwise (to simplify arguments, we treat $\Sigma$ as known here.)
To state our next result, recall the definitions of $S_{\vall}$ and
$H_{\vall}$ given in \eqref{SHpi}.
\begin{lem}
\label{lem:limexp}Jointly over any finite collection of $(\vall,f)\in\reals^{kp^{2}}$,
\[
\like_{n}(\vall,f)-\like_{n}(0,0)\wkc[S_{\vall}^{\trans}\vall-\tfrac{1}{2}\vall^{\trans}H_{\vall}\vall]+[S_{\st}^{\trans}f-\tfrac{1}{2}f^{\trans}H_{\st}f].
\]
\end{lem}
We next show that, up to the term depending on $f$, the preceding
is also the limit of the loglikelihood ratio process in a multivariate
predictive regression with a known covariance matrix; recall \assref{PR}
given in \subsecref{likelihoodasymp}.
\begin{lem}
\label{lem:predreg}Suppose that $\{y_{\pr,t}\}$ and $\{z_{\pr t}\}$
are generated under \assref{PR}, and that $\xi_{t}=[\begin{smallmatrix}\xi_{yt}\\
\xi_{zt}
\end{smallmatrix}]\distiid N[0,\Omega]$ with $\Omega=K\Sigma K^{\trans}$. Then for
\[
\vall=n\vek\begin{bmatrix}A-A(\PHI_{n})\\
\Lambda-\Lambda_{\lu}(\PHI_{n})
\end{bmatrix},
\]
we have, jointly over any finite collection of $\vall\in\reals^{pq}$
\[
\like_{n,\pr}(\vall)-\like_{n,\pr}(0)\wkc S_{\vall}^{\trans}\vall-\tfrac{1}{2}\vall^{\trans}H_{\vall}\vall,
\]
where $\like_{n,\pr}(\vall)$ is the loglikelihood defined in \thmref{emw}.
\end{lem}
Finally, we show that when (the entirety of) $\PHI$ is in unknown,
and the model is estimated subject to the constraint \eqref{vall},
then the limit of the \emph{concentrated} loglikelihood ratio process
is asymptotically identical to that of the predictive regression,
up to (random) terms that do not depend on $\vall$. Let $\hat{f}_{n\mid\vall}\defeq n^{1/2}\hat{\f}_{n,\st\mid\vall}$,
where $\hat{\f}_{n\mid\vall}$ denotes the maximiser of $\likens_{n}(\f)$
subject to $\f$ satisfying \eqref{vall}.
\begin{lem}
\label{lem:like-cons}Jointly over every finite collection of $\vall\in\reals^{pq}$,
\[
\like_{n}(\vall,\hat{f}_{n\mid\vall})-\like_{n}(0)\wkc S_{\vall}^{\trans}\vall-\tfrac{1}{2}\vall^{\trans}H_{\vall}\vall+S_{\st}^{\trans}H_{\st}^{-1}S_{\st}
\]
\end{lem}
\begin{proof}[Proof of \lemref{invpsi}]
 Consider the mapping $\Psi$ and the permutation matrix $M\in\reals^{pq\times pq}$
such that
\begin{align}
\Psi(\PHI) & \defeq\begin{bmatrix}\vek A(\PHI)\\
\vek\Lambda_{\lu}(\PHI)\\
\vek\PHI\R_{0,\st}
\end{bmatrix} & {\cal M}\Psi(\PHI)\defeq\begin{bmatrix}M & 0\\
0 & I_{\#\st}
\end{bmatrix}\Psi(\PHI) & =\begin{bmatrix}\vek\begin{bmatrix}A(\PHI)\\
\Lambda_{\lu}(\PHI)
\end{bmatrix}\\
\vek\PHI\R_{0,\st}
\end{bmatrix}\label{eq:Psi-and-M}
\end{align}
By Lemmas~\ref{lem:implicit-maps} and \ref{lem:derivatives}, $\Psi$
is smooth and at $\PHI=\PHI_{0}$ has first differential
\begin{equation}
\deriv\Psi=\begin{bmatrix}J_{A}(\PHI_{0}) & 0\\
J_{\Lambda}(\PHI_{0}) & 0\\
0 & I_{\#\st}
\end{bmatrix}\left(\begin{bmatrix}\R_{0,\lu}^{\trans}\\
\R_{0,\st}^{\trans}
\end{bmatrix}\otimes I_{p}\right)\vek(\deriv\PHI).\label{eq:dPsi}
\end{equation}
The Jacobian on the r.h.s.\ is invertible by \lemref{GLR}\ref{enu:GLR:L}
and \lemref{derivatunity} (for the latter, since $\Lambda_{\lu}(\PHI_{0})=I_{q}$).
Thus by the inverse mapping theorem, there is an open neighbourhood
$N_{\set P}\subset\set P$ of $\PHI_{0}$ on which $\Psi$ has a smooth
inverse.

Now let $\tau_{n}(\f)\defeq\PHI_{n}+\vek^{-1}(\f)\L_{n}^{\trans}$,
which converges (uniformly on compacta) to a linear and invertible
mapping $\tau_{0}(\f)$ for which $\tau_{0}(0)=\PHI_{0}$. Hence there
exists an $n_{0}\in\naturals$ and a (fixed) open neighbourhood $N\subset\pset_{n}$
of zero such that $\tau_{n}(N)\subset N_{\set P}$, for all $n\geq n_{0}$,
with $\tau_{n}$ being invertible on $N$. By composition, the sequence
of maps defined by
\begin{equation}
D_{n}^{-1}\psi_{n}(\f)={\cal M}\{\Psi[\tau_{n}(\f)]-\Psi(\PHI_{n})\}\label{eq:almost-psi}
\end{equation}
is smooth and invertible on $N$, for all $n\geq0$, and has a smooth
inverse there; hence part~\enuref{invpsi:forward} holds. Finally,
since the image of $N$ under the r.h.s.\ must itself be an open
neighbourhood of zero, and
\begin{equation}
\begin{bmatrix}\vall\\
f
\end{bmatrix}=\psi_{n}(\f)=D_{n}{\cal M}\{\Psi[\tau_{n}(\f)]-\Psi(\PHI_{n})\},\label{eq:repmap}
\end{equation}
we may deduce that for any compact neighbourhood ${\cal K}$ of zero,
there is an $n_{1}\geq n_{0}$ such that the inverse $\psi_{n}^{-1}$
is well-defined and smooth for all $(\vall,f)\in{\cal K}$, for all
$n\geq n_{1}$. Finally, to show that the $\f_{n}\defeq\psi_{n}^{-1}(\vall,f)$
has $D_{n}\f_{n}=O(1)$, we note that since the r.h.s.\ of \eqref{almost-psi}
is (locally to zero) a diffeomorphism, which itself equals zero at
$\f=0$, the fact that
\[
{\cal M}\{\Psi[\tau_{n}(\f_{n})]-\Psi(\PHI_{n})\}=D_{n}^{-1}\begin{bmatrix}\vall\\
f
\end{bmatrix}\goesto0
\]
must imply that $\f_{n}\goesto0$. Hence follows from \eqref{dPsi}
and \eqref{repmap} that, by a Taylor expansion of \eqref{repmap}
around $\f=0$,
\[
\begin{bmatrix}MJ+o_{p}(1) & 0\\
0 & I_{\#\st}
\end{bmatrix}\begin{bmatrix}\f_{n,\lu}\\
\f_{n,\st}
\end{bmatrix}=D_{n}^{-1}\begin{bmatrix}\vall\\
f
\end{bmatrix}
\]
whence $D_{n}\f_{n}=O(1)$ as claimed. Thus part~\enuref{invpsi:back}
holds.
\end{proof}
\begin{proof}[Proof of \lemref{limexp}]
In view of \lemref{invpsi}, we may take $n$ sufficiently large
such that $\psi_{n}^{-1}$ is well defined at $(\vall,f)$. Let $\f_{n}^{\ast}\defeq\psi_{n}^{-1}(\vall,f)=o(1)$,
$\PHI_{n}^{\ast}\defeq\PHI_{n}+\vek^{-1}(\f_{n}^{\ast})\L_{n}^{\trans}$,
and $M\in\reals^{pq\times pq}$ be as in \eqref{Psi-and-M}. Then
by \propref{deltamethod}, for $J\defeq J(\PHI_{0})$
\begin{equation}
\vall=n\vek\begin{bmatrix}A(\PHI_{n}^{\ast})-A(\PHI_{n})\\
\Lambda_{\lu}(\PHI_{n}^{\ast})-\Lambda_{\lu}(\PHI_{n})
\end{bmatrix}=[MJ+o(1)]n\f_{n,\lu}^{\ast}\label{eq:pi-deriv}
\end{equation}
where by \lemref{derivatunity} and the definition of $M$,
\[
MJ=M\begin{bmatrix}I_{q}\otimes{\cal J}\\
I_{q}\otimes L_{\lu}^{\trans}
\end{bmatrix}=I_{q}\otimes\begin{bmatrix}{\cal J}\\
L_{\lu}^{\trans}
\end{bmatrix}\eqdef I_{q}\otimes K
\]
for ${\cal J}$ and $K$ as defined in \eqref{Kdef}. Noting also
that $n^{1/2}\f_{n,\st}^{\ast}=f$, it follows from \lemref{lhoodexp}
and \eqref{pi-deriv} that
\begin{align*}
\like_{n}(\vall,f)-\like_{n}(0,0) & =\like_{n}(\f_{n,\lu}^{\ast},\f_{n,\st}^{\ast})-\like_{n}(0,0)\\
 & =S_{n}^{\trans}(D_{n}\f_{n}^{\ast})-\tfrac{1}{2}(D_{n}\f_{n}^{\ast})^{\trans}H_{n}(D_{n}\f_{n}^{\ast})\\
 & \wkc[S_{\lu}^{\trans}(MJ)^{-1}\vall-\tfrac{1}{2}\vall^{\trans}[(MJ)^{-1}]^{\trans}H_{\lu}(MJ)^{-1}\vall]+[S_{\st}^{\trans}f-\tfrac{1}{2}f^{\trans}H_{\st}f].
\end{align*}
To complete the proof, we note that
\begin{align*}
[(MJ)^{-1}]^{\trans}S_{\lu} & =(I_{q}\otimes K^{-1})^{\trans}\int_{0}^{1}[\Zdet_{C}(r)\otimes\Sigma^{-1}\deriv E(r)]\\
 & =\int_{0}^{1}[\Zdet_{C}(r)\otimes(K\Sigma K^{\trans})^{-1/2}\deriv W(r)]=S_{\vall}
\end{align*}
where we have used that $E(s)=\Sigma^{-1/2}W(s)$ and,
\begin{align*}
[(MJ)^{-1}]^{\trans}H_{\lu}(MJ)^{-1} & =(I_{q}\otimes K^{-1})^{\trans}\left(\int\Zdet_{C}\Zdet_{C}^{\trans}\otimes\Sigma^{-1}\right)(I_{q}\otimes K^{-1})\\
 & =\int\Zdet_{C}\Zdet_{C}^{\trans}\otimes(K\Sigma K^{\trans})^{-1}=H_{\vall}.\qedhere
\end{align*}
\end{proof}
\begin{proof}[Proof of \lemref{predreg}]
 Letting $\vAll=[\begin{smallmatrix}A_{\pr}\\
\Lambda_{\pr}
\end{smallmatrix}]$ and noting that $\vall=n\vek(\vAll-[\begin{smallmatrix}A(\PHI_{n})\\
\Lambda_{\lu}(\PHI_{n})
\end{smallmatrix}])$, we have
\[
\like_{n}^{\pr}(\vall)=K_{n}-\frac{1}{2}\sum_{t=1}^{n}\smlnorm{x_{t}-\vAll z_{t-1}}_{\Omega^{-1}},
\]
where $K_{n}\defeq-\frac{n}{2}\log(2\pi\log\det\Omega)$. It then
follows by exactly the same arguments as were used in the proof of
\lemref{lhoodexp} that
\[
\like_{n,\pr}(\vall)-\like_{n,\pr}(0)=S_{n,\pr}^{\trans}\vall-\tfrac{1}{2}\vall^{\trans}H_{n,\pr}\vall
\]
where
\begin{align*}
S_{n,\pr} & =\frac{1}{n}\sum_{t=1}^{n}(\zdet_{t-1}\otimes\Omega^{-1/2}\eta_{t}) & H_{n,\pr} & =\frac{1}{n}\sum_{t=1}^{n}(\zdet_{t-1}\zdet_{t-1}^{\trans}\otimes\Omega^{-1}).
\end{align*}
Under \assref{PR}, it follows by \lemref{wkconv}\ref{enu:wkconv:Zc}
that $n^{-1/2}\zdet_{\smlfloor{nr}}\wkc\Zdet_{C,\pr}(r)$ on $D[0,1]$,
where $\Zdet_{C,\pr}(r)$ denotes the residual from the projection
of 
\begin{equation}
Z_{C,\pr}(r)\defeq\int_{0}^{r}\e^{C(r-s)}\Omega_{zz}^{1/2}\deriv W(s)\label{eq:Zpr}
\end{equation}
on a constant and a linear trend, and we have partitioned $\Omega=[\begin{smallmatrix}\Omega_{yy} & \Omega_{yz}\\
\Omega_{zy} & \Omega_{zz}
\end{smallmatrix}]$ conformably with $\xi_{t}=[\begin{smallmatrix}\xi_{yt}\\
\xi_{zt}
\end{smallmatrix}]$. Then by the continuous mapping theorem and the same arguments as
used in the proof of \lemref{wkconv}\enuref{wkconv:si},
\begin{align}
S_{n,\pr} & \wkc\int_{0}^{1}[\Zdet_{C,\pr}(r)\otimes\Omega^{-1/2}W(r)]\diff r & H_{n,\pr} & \wkc\int\Zdet_{C,\pr}\Zdet_{C,\pr}^{\trans}\otimes\Omega^{-1}.\label{eq:SHpr}
\end{align}
Thus we can bring \eqref{Zpr} into agreement with \eqref{Zproc},
and the limits on the r.h.s.\ of \eqref{SHpr} with \eqref{SHpi},
by setting
\[
\Omega=\begin{bmatrix}\Omega_{yy} & \Omega_{yz}\\
\Omega_{zy} & \Omega_{zz}
\end{bmatrix}=\begin{bmatrix}{\cal J}\Sigma{\cal J}^{\trans} & {\cal J}\Sigma L_{\lu}\\
L_{\lu}^{\trans}\Sigma{\cal J}^{\trans} & L_{\lu}^{\trans}\Sigma L_{\lu}
\end{bmatrix}=K\Sigma K^{\trans}.\qedhere
\]
\end{proof}
\begin{proof}[Proof of \lemref{like-cons}]
 We first show that $D_{n}\hat{\f}_{n\mid\vall}=O_{p}(1)$. By \lemref{invpsi},
for all $n$ sufficiently large, there exists a (deterministic) sequence
$\f_{n\mid\vall}\in\pset_{n}$ with $D_{n}\f_{n\mid\vall}=O(1)$,
such that \eqref{vall} holds at $\f=\f_{n\mid\vall}$. It follows
from \lemref{lhoodexp} that for each $\epsilon>0$, there exists
an $N<\infty$ such that
\[
\limsup_{n\goesto\infty}\Prob\{\like_{n}^{\ast}(\f_{n\mid\vall})-\like_{n}^{\ast}(0)<-N\}<\epsilon/2.
\]
On the other hand, adapting the argument given in the proof of \lemref{consistency},
we may also choose $M<\infty$ sufficiently large such that
\begin{align*}
\Prob\left\{ \sup_{\{\f\in\spcf_{n}\mid\smlnorm{D_{n}\f}\geq M\}}[\likens_{n}(\f)-\likens_{n}(0)]<-2N\right\}  & \ge\Prob\left\{ M[\smlnorm{S_{n}}-\tfrac{1}{2}\lambda_{\min}(H_{n})M]\leq-2N\right\} \\
 & >1-\epsilon/2
\end{align*}
for all $n$ sufficiently large. Deduce that with probability at least
$1-\epsilon$, $\like_{n}^{\ast}(\f_{n\mid\vall})$ must strictly
exceed $\like_{n}(\f)$ over all $\f\in\pset_{n}$ with $\smlnorm{D_{n}\f}\geq M$;
it follows that the constrained maximiser $\hat{\f}_{n\mid\vall}$
must have $\smlnorm{D_{n}\hat{\f}_{n\mid\vall}}<M$. Deduce that $D_{n}\hat{\f}_{n\mid\vall}=O_{p}(1)$
as claimed.

Now it follows from \eqref{dPsi} and \eqref{repmap} that, at $\f=\hat{\f}_{n\mid\vall}$,
\[
\begin{bmatrix}\deriv\vall\\
\deriv f
\end{bmatrix}=\begin{bmatrix}MJ+o_{p}(1) & 0\\
0 & I_{\#\st}
\end{bmatrix}D_{n}\deriv\f
\]
and hence
\[
D_{n}\deriv\f=\begin{bmatrix}(MJ)^{-1}+o_{p}(1) & 0\\
0 & I_{\#\st}
\end{bmatrix}\begin{bmatrix}\deriv\vall\\
\deriv f
\end{bmatrix}
\]
at $(\vall,\hat{f}_{n\mid\vall})$. Thus
\begin{equation}
n\hat{\f}_{n,\lu}=[(MJ)^{-1}+o_{p}(1)]\vall,\label{eq:consLUlim}
\end{equation}
and since $\hat{f}_{n\mid\vall}$ must satisfy the first-order conditions
for a maximum, we have from \lemref{lhoodexp} that 
\begin{align*}
\grad_{f}\like_{n}(\vall,\hat{f}_{n\mid\vall}) & =\grad_{f}[S_{n}^{\trans}(D_{n}\f)-\tfrac{1}{2}(D_{n}\f)^{\trans}H_{n}(D_{n}\f)]_{f=\hat{f}_{n\mid\vall}}\\
 & =S_{n,\st}-(n\hat{\f}_{n,\lu\mid\vall})^{\trans}H_{n,\ls}-H_{n,\st}n^{1/2}\hat{\f}_{n,\st\mid\vall}
\end{align*}
whence 
\begin{equation}
n^{1/2}\hat{\f}_{n,\st\mid\vall}=H_{n,\st}^{-1}S_{n,\st}+o_{p}(1)\wkc H_{\st}^{-1}S_{\st}.\label{eq:consSTlim}
\end{equation}
Thus, in view of \eqref{consLUlim} and \eqref{consSTlim}, the weak
limit of
\begin{align*}
\like_{n}(\vall,\hat{f}_{n\mid\vall})-\like_{n}(0) & =\like_{n}^{\ast}(\hat{\f}_{n\mid\vall})-\like_{n}^{\ast}(0)=S_{n}^{\trans}(D_{n}\hat{\f}_{n\mid\pi})-\tfrac{1}{2}(D_{n}\hat{\f}_{n\mid\pi})^{\trans}H_{n}(D_{n}\hat{\f}_{n\mid\pi})
\end{align*}
is as claimed.
\end{proof}

\section{Proofs of theorems}

\label{app:theoremproofs}
\begin{proof}[Proof of \thmref{emw}]
 This follows directly from Lemmas~\ref{lem:limexp}--\ref{lem:like-cons},
noting in particular that $\likens_{n}(\vall)=\like_{n}(\vall,\hat{f}_{n\mid\vall})$,
where the latter is as appears in \lemref{like-cons}.
\end{proof}
\begin{proof}[Proof of \thmref{estimators}]
 \proofpart{1} In the notation of \appref{asymptotics}, $\hat{\f}_{n,\lu}=\vek\{(\hat{\PHI}_{n}-\PHI_{n})\R_{n,\lu}\}$.
By \propref{andrews}\enuref{andrews:unres}
\begin{align*}
n\vek\{(\hat{\PHI}_{n}-\PHI_{n})\R_{n,\lu}\} & \wkc\left[\left(\int\Zdet_{C}\Zdet_{C}^{\trans}\right)^{-1}\otimes I_{p}\right]\int_{0}^{1}[\Zdet_{C}(r)\otimes\deriv E(r)]\\
 & =\vek\left\{ \int(\deriv E)\Zdet_{C}^{\trans}\left(\int\Zdet_{C}\Zdet_{C}^{\trans}\right)^{-1}\right\} ,
\end{align*}
and so by \propref{deltamethod} 
\begin{equation}
\begin{bmatrix}\vek\{A(\hat{\PHI}_{n})-A(\PHI_{n})\}\\
\vek\{\Lambda_{\lu}(\hat{\PHI}_{n})-\Lambda_{\lu}(\PHI_{n})\}
\end{bmatrix}\wkc\begin{bmatrix}J_{A}(\PHI_{0})\\
J_{\Lambda}(\PHI_{0})
\end{bmatrix}\vek\left\{ \int(\deriv E)\Zdet_{C}^{\trans}\left(\int\Zdet_{C}\Zdet_{C}^{\trans}\right)^{-1}\right\} .\label{eq:limdeltameth}
\end{equation}
Since $\PHI_{n}\goesto\PHI_{0}$ with $\Lambda_{\lu}(\PHI_{0})=I_{q}$
under \assref{LOC}, we have by \lemref{derivatunity} that
\begin{equation}
\begin{bmatrix}J_{A}(\PHI_{0})\\
J_{\Lambda}(\PHI_{0})
\end{bmatrix}=\begin{bmatrix}I_{q}\otimes\beta^{\trans}R_{\st}(I_{kp-q}-\Lambda_{\st})^{-1}L_{\st}^{\trans}\\
I_{q}\otimes L_{\lu}^{\trans}
\end{bmatrix}.\label{eq:limjacob}
\end{equation}
The result then follows from \eqref{limdeltameth} and \eqref{limjacob},
by reversing the vectorisation.

\proofpart{2} In the notation of \appref{asymptotics}, maximising
$\like_{n}^{\ast}(\PHI)$ subject to $\Lambda_{\lu}(\PHI)=\Lambda_{n,\lu}=I_{q}+C/n$
corresponds to maximising $\like_{n}(\f)$ subject to $\theta_{n}(\f)=0$.
Thus $\hat{\f}_{n,\lu\mid\theta}=\vek\{(\hat{\PHI}_{n\mid\Lambda_{n,\lu}}-\PHI_{n})\R_{n,\lu}\}$,
and so by \propref{andrews}\enuref{andrews:res}
\[
n\vek\{(\hat{\PHI}_{n\mid\Lambda_{n,\lu}}-\PHI_{n})\R_{n,\lu}\}\wkc\Theta_{\perp}(\Theta_{\perp}^{\trans}H_{\lu}\Theta_{\perp})^{-1}\Theta_{\perp}^{\trans}S_{\lu}
\]
where $\Theta_{\perp}=I_{q}\otimes L_{\lu,\perp}$. Hence by \propref{deltamethod},
\[
\vek\{A(\hat{\PHI}_{n\mid\Lambda_{n,\lu}})-A(\PHI_{n})\}\wkc J_{A}(\PHI_{0})\Theta_{\perp}(\Theta_{\perp}^{\trans}H_{\lu}\Theta_{\perp})^{-1}\Theta_{\perp}^{\trans}S_{\lu}.
\]

To determine the distribution of the r.h.s., we note that
\begin{equation}
\Theta_{\perp}^{\trans}S_{\lu}=\int_{0}^{1}[\Zdet_{C}(r)\otimes L_{\lu,\perp}^{\trans}\Sigma^{-1}\deriv E(r)]\eqdef\int_{0}^{1}[\Zdet_{C}(r)\otimes\deriv U(r)].\label{eq:ThetaSlu}
\end{equation}
Recall that $\Zdet_{C}$ is a function only of $Z_{C}$, which from
\eqref{Zproc} is given by
\begin{equation}
Z_{C}(r)=\int_{0}^{r}\e^{C(r-s)}L_{\lu}^{\trans}\deriv E(s)\eqdef\int_{0}^{r}\e^{C(r-s)}\deriv V(s).\label{eq:Zprocagain}
\end{equation}
$(U,V)=(L_{\lu,\perp}^{\trans}\Sigma^{-1}E,L_{\lu}^{\trans}E)$ is
a pair of vector Brownian motions, with covariance
\[
\expect U(1)V(1)^{\trans}=L_{\lu,\perp}^{\trans}\Sigma^{-1}\expect[E(1)E(1)^{\trans}]L_{\lu}=L_{\lu,\perp}^{\trans}L_{\lu}=0;
\]
whence $U$ and $V$ are independent. In particular, we have from
\eqref{Zprocagain} that $U$ is independent of $\Zdet_{C}$. This,
combined with the fact that
\begin{align*}
J_{A}(\PHI_{0})\Theta_{\perp}(\Theta_{\perp}^{\trans}H_{\lu}\Theta_{\perp})^{-1} & =\left(\int\Zdet_{C}\Zdet_{C}^{\trans}\right)^{-1}\otimes\mathcal{J}L_{\lu,\perp}(L_{\lu,\perp}^{\trans}\Sigma^{-1}L_{\lu,\perp})^{-1}
\end{align*}
depends only on $\Zdet_{C}$, implies $J_{A}(\PHI_{0})\Theta_{\perp}(\Theta_{\perp}^{\trans}H_{\lu}\Theta_{\perp})^{-1}\Theta_{\perp}^{\trans}S_{\lu}$
is mixed normal with variance
\[
\left(\int\Zdet_{C}\Zdet_{C}^{\trans}\right)^{-1}\otimes\mathcal{J}L_{\lu,\perp}(L_{\lu,\perp}^{\trans}\Sigma^{-1}L_{\lu,\perp})^{-1}L_{\lu,\perp}^{\trans}\mathcal{J}^{\trans},
\]
which proves \eqref{Arstr}.

Finally, note that the preceding holds for any choice of $L_{\lu,\perp}\in\reals^{p\times r}$
having full column rank and $L_{\lu,\perp}^{\trans}L_{\lu}=0$. Let
$\alpha\defeq\Phi_{0}(1)\beta(\beta^{\trans}\beta)^{-1}\in\reals^{p\times r}$,
where $\Phi_{0}(1)\defeq\lim_{n\goesto\infty}\Phi_{n}(1)$; then
\[
L_{\lu}^{\trans}\alpha=L_{\lu}^{\trans}\Phi_{0}(1)\beta(\beta^{\trans}\beta)^{-1}=0
\]
by \eqref{eig-eig} with $\Lambda_{\lu}=\Lambda_{\lu}(\PHI_{0})=I_{q}$.
Further, $\rank\alpha=r$ since $\spn\Phi_{0}(1)=\spn\beta$, and
thus we may indeed choose $L_{\lu,\perp}=\alpha$. In this case,
\begin{align*}
\mathcal{J}L_{\lu,\perp} & =\beta^{\trans}R_{\st}(I_{kp-q}-\Lambda_{\st})^{-1}L_{\st}^{\trans}\Phi_{0}(1)\beta(\beta^{\trans}\beta)^{-1}=_{(1)}\beta^{\trans}\beta(\beta^{\trans}\beta)^{-1}=I_{r},
\end{align*}
where $=_{(1)}$ follows from \eqref{betastuff} above. Thus \eqref{johvar}
is proved.
\end{proof}
\begin{proof}[Proof of \thmref{lrstats}]
 We first prove \eqref{multDF}. In the notation of \appref{asymptotics},
$\lr_{n}(\Lambda_{n,\lu})=2[\like_{n}^{\ast}(\hat{\f}_{n})-\like_{n}^{\ast}(\hat{\f}_{n\mid\theta})]$.
By \propref{andrews}\enuref{andrews:lrroot},
\[
\lr_{n}(\Lambda_{n,\lu})\wkc S_{\lu}^{\trans}H_{\lu}^{-1}\Theta(\Theta^{\trans}H_{\lu}^{-1}\Theta)^{-1}\Theta^{\trans}H_{\lu}^{-1}S_{\lu}\eqdef\lr,
\]
where $\Theta=I_{q}\otimes L_{\lu}$, $S_{\lu}=\int[\Zdet_{C}(r)\otimes\Sigma^{-1}\deriv E]$,
and $H_{\lu}=\int\Zdet_{C}\Zdet_{C}^{\trans}\otimes\Sigma^{-1}$.
To obtain the claimed expression for $\lr$, note that
\[
S_{\lu}=\int[\Zdet_{C}(r)\otimes\Sigma^{-1}\deriv E]=\vek\left\{ \Sigma^{-1}\int(\deriv E)\Zdet_{C}^{\trans}\right\} 
\]
and
\[
H_{\lu}^{-1}\Theta(\Theta^{\trans}H_{\lu}^{-1}\Theta)^{-1}\Theta^{\trans}H_{\lu}^{-1}=\left(\int\Zdet_{C}\Zdet_{C}^{\trans}\right)^{-1}\otimes\Sigma L_{\lu}(L_{\lu}^{\trans}\Sigma L_{\lu})^{-1}L_{\lu}^{\trans}\Sigma
\]
whence, using $\vek(A)^{\trans}\vek(B)=\tr(A^{\trans}B)$,
\begin{align}
\lr & =\tr\left\{ \Delta^{-1/2}L_{\lu}^{\trans}\int(\deriv E)\Zdet_{C}^{\trans}\left(\int\Zdet_{C}\Zdet_{C}^{\trans}\right)^{-1}\int\Zdet_{C}(\deriv E)^{\trans}L_{\lu}\Delta^{-1/2}\right\} \label{eq:LRmed}
\end{align}
where $\Delta\defeq L_{\lu}^{\trans}\Sigma L_{\lu}$. To simplify
this further, note that $L_{\lu}^{\trans}E$ is a $q$-dimensional
Brownian motion with variance $\Delta$, and so for $W_{\ast}(r)\defeq\Delta^{-1/2}L_{\lu}^{\trans}E(r)\sim\BM(I_{q})$,
we have
\begin{multline*}
Z_{C}(r)=\int_{0}^{r}\e^{C(r-s)}L_{\lu}^{\trans}\deriv E(s)=\int_{0}^{r}\e^{C(r-s)}\Delta^{1/2}\deriv W_{\ast}(s)\\
=_{(1)}\Delta^{1/2}\int_{0}^{r}\e^{C_{\ast}(r-s)}\deriv W_{\ast}(s)\eqdef\Delta^{1/2}Z_{C_{\ast}}(r)
\end{multline*}
where $C_{\ast}\defeq\Delta^{-1/2}C\Delta^{1/2}$ is as in the statement
of the theorem, and $=_{(1)}$ follows from $\e^{C}D=D\e^{D^{-1}CD}$
for any nonsingular $D$. Hence $\Zdet_{C}(r)=\Delta^{1/2}\Zdet_{C_{\ast}}(r)$,
whereupon \eqref{multDF} follows from \eqref{LRmed} and the definition
of $W_{\ast}$.

We next prove \eqref{chisqlim}. Maximisation of $\like_{n}^{\ast}(\PHI)$
subject to $\Lambda_{\lu}(\PHI)=I_{q}+C/n$ and $a_{ij}(\PHI)=a_{0}$
corresponds, in the notation of \appref{asymptotics}, to maximisation
of $\like_{n}(\f)$ subject to $\theta_{n}(\f)=0$ and $\gamma_{n}(\f)=0$.
Therefore by \propref{andrews}\enuref{andrews:coef},
\begin{align*}
\lr_{n}[a_{ij}(\PHI_{n});\Lambda_{n,\lu}] & =2[\like_{n}(\hat{\f}_{n\mid\theta})-\like_{n}(\hat{\f}_{n\mid\theta,\gamma})]\wkc(H_{\Theta,\perp}^{-1/2}\Theta_{\perp}^{\trans}S_{\lu})^{\trans}[I_{qr}-\mathcal{Q}](H_{\Theta,\perp}^{-1/2}\Theta_{\perp}^{\trans}S_{\lu}).
\end{align*}
Recall from \eqref{ThetaSlu} and the subsequent arguments that
\[
\vek\{\Theta_{\perp}^{\trans}S_{\lu}\}\eqdist\left(\int\Zdet_{C}\Zdet_{C}^{\trans}\otimes L_{\lu,\perp}^{\trans}\Sigma^{-1}L_{\lu,\perp}\right)^{1/2}\eta
\]
for $\eta\sim\normdist[0,I_{qr}]$ independent of $\Zdet_{C}$, and
therefore also of
\[
H_{\Theta,\perp}=\Theta_{\perp}^{\trans}H_{\lu}\Theta_{\perp}=\int\Zdet_{C}\Zdet_{C}^{\trans}\otimes L_{\lu,\perp}^{\trans}\Sigma^{-1}L_{\lu,\perp}.
\]
Thus $\vek\{H_{\Theta,\perp}^{-1/2}\Theta_{\perp}^{\trans}S_{\lu}\}\sim\normdist[0,I_{qr}]$
is independent of $H_{\lu}$, and therefore also of $\mathcal{Q}$.
The result follows by noting that $H_{\Theta,\perp}^{1/2}\Xi$ has
rank $qr-1$ a.s., whence $I_{qr}-\mathcal{Q}$ projects orthogonally
onto a subspace of dimension 1, a.s.
\end{proof}

\section{Computational appendix}

\label{app:numerical}

\subsection{Test statistics}

Computation of $\ci_{a_{ij}\mid\Lambda_{0}}$ and $\cinp$ involves
maximising $\likens_{n}(\PHI)$ subject to the restrictions that $\Lambda_{\lu}(\PHI)=\Lambda_{0}$
for some specified $\Lambda_{0}\in\set L$, and possibly also that
for and $a_{ij}(\PHI)=a_{0}$ for some $a_{0}\in\reals$. To implement
this estimator numerically, we introduce the constraint
\begin{equation}
\begin{bmatrix}A\\
I_{q}
\end{bmatrix}\Lambda_{0}^{k}-\sum_{i=1}^{k}\Phi_{i}\begin{bmatrix}A\\
I_{q}
\end{bmatrix}\Lambda_{0}^{k-i}=0,\label{eq:estconstraint}
\end{equation}
which incorporates \eqref{eig-eig} and \eqref{Rlurnom} above: it
forces $\Phi(\lambda)$ to have roots at the eigenvalues of $\Lambda_{0}$,
and the associated $R_{\lu}$ matrix to respect the normalisation
\eqref{Rlurnom}. We then proceed as follows:
\begin{enumerate}
\item Given $A\in\reals^{r\times q}$ and $\Lambda_{0}\in\set L$, maximise
$\likens_{n}(\PHI)$ over $\PHI\in\reals^{p\times kp}$, subject to
\eqref{estconstraint}, to obtain the maximum likelihood estimate
$\hat{\PHI}_{n\mid A,\Lambda_{0}}$ using two-step, restricted least-squares
estimation. This is straightforward, since \eqref{estconstraint}
is a linear restriction on $\PHI$ (see \citealp{Lut07}, Ch.\ 7).
\item Using a general purpose optimiser, compute 
\begin{equation}
\max_{A\in\reals^{r\times q}}\likens_{n}(\hat{\PHI}_{n\mid A,\Lambda_{0}}).\label{eq:qcc-estimator}
\end{equation}
\end{enumerate}
The maximum of $\likens_{n}(\PHI)$ subject to $\Lambda_{\lu}(\PHI)=\Lambda_{0}$
and $a_{ij}(\PHI)=a_{0}$ obtains by holding restricting $a_{ij}$
when computing the maximum in \eqref{qcc-estimator}. Point estimates
of $\PHI$ that merely impose the requirement that $\Lambda_{0}\in\set L$
(as appears e.g.\ in \eqref{LRroot}) obtain by maximising $\likens_{n}(\hat{\PHI}_{n\mid A,\Lambda_{0}})$
over both $A$ and $\Lambda_{0}\in\set L$.

When $q=1$, and in the special case where $\Lambda_{0}=\lambda_{0}I_{q}$,
\eqref{qcc-estimator} simplifies so that model \eqref{redform} becomes
\begin{equation}
\Delta_{\lambda_{0}}y_{t}=m+dt-\lambda_{0}^{-k+1}\Phi(\lambda_{0})y_{t-1}+\sum_{i=1}^{p-1}\Psi_{i}\Delta_{\lambda_{0}}y_{t-i}\label{eq:rrr-model}
\end{equation}
where $\Delta_{\lambda_{0}}y_{t}\defeq y_{t}-\lambda_{0}y_{t-1}$
denotes a quasi-difference (see \lemref{quasi-VECM}). Since $\Phi(\lambda_{0})$
has rank $p-q=r$, $\likens_{n}(\PHI)$ can then be efficiently maximised,
subject to \textbf{$\Lambda_{\lu}(\PHI)=\lambda_{0}I_{q}$, }via a
reduced rank regression, exactly as in \citet[Ch.~6]{Joh95}.

When $q\geq2$, some care needs to be taken with the parametrisation
of $\set L$. If we take this to be either the set of real normal
($\Ln$) or symmetric ($\Ls$) matrices, then each $\Lambda_{\lu}\in\set L$
can be expressed as $\Lambda_{\lu}=QD_{\lu}Q^{\trans}$, where $Q\in\reals^{q\times q}$
is an orthogonal matrix ($Q^{\trans}Q=I_{q})$ and $D_{\lu}$ is a
block diagonal, with blocks that are either: $1\times1$ and equal
to each of the real eigenvalues of $\Lambda_{\lu}$, or $(2\times2)$
and of the form $[\begin{smallmatrix}a & b\\
-b & a
\end{smallmatrix}]$, if $\Lambda_{\lu}$ has a pair of complex eigenvalues at $\lambda=a\pm\i b$
(\citealp{HJ13book}, Thm.~2.5.6 and 2.5.8). Since $Q$ can be constructed
from $q(q-1)/2$ plane rotations (\citealp[Prob.~2.1.P29]{HJ13book}),
both $\Ln$ and $\Ls$ can thus be expressed in terms of of $q(q+1)/2$
free parameters lying in a compact set.

\subsection{\label{app:nearly-opt-tests}Nearly optimal tests}

For simplicity of exposition, suppose that $p=2$ and $q=r=1$, so
that $A=a$ and $\Lambda_{\lu}=\lambda_{\lu}$. We want to test
\[
H_{0}:a=a_{0}\sep\lambda_{\lu}\in\set L\qquad\text{against}\qquad H_{1}:a\neq a_{0}\sep\lambda_{\lu}\in\set L,
\]
where $\set L=[\radius,1]$ for some user-chosen $\radius<1$. Let
$F_{0}$ and $F_{1}$ denote discrete distributions on $\reals\times\set L$,
that respectively concentrate on those subsets of the parameter space
consistent with the null and the alternative. Consider a test of the
form
\[
\np_{n}(a)\defeq\indic\left\{ \int_{\reals\times\set L}\e^{\likens_{n}(a,\lambda)}F_{1}(\deriv a,\deriv\lambda)>\crit_{\alpha}\int_{\set L}\e^{\likens_{n}(a_{0},\lambda)}F_{0}(a_{0},\deriv\lambda)\right\} 
\]
To implement the procedure of \citet{EMW15Ecta}, we require estimates
of the following probabilities:
\begin{enumerate}
\item the null rejection probability $\Prob_{(a_{0},\lambda)}\{\np_{n}(a_{0})=1\}$
for $\lambda$ in (a discretisation of) $\set L$;
\item the weighted null rejection probability $\int\Prob_{(a_{0},\lambda)}\{\np_{n}(a_{0})=1\}F_{0}(a_{0},\deriv\lambda)$;
and
\item the power under the weighted alternative, $\int_{\reals\times\set L}\Prob_{(a,\lambda)}\{\np_{n}(a_{0})=1\}F_{1}(\deriv a,\deriv\lambda)$.
\end{enumerate}
In view of \thmref{emw}, aside from $(a,\lambda)$ the only parameter
that these probabilities depend on, in large samples, is the long-run
covariance matrix $\Omega\defeq K\Sigma K^{\trans}$.

Suppose for the moment that $\Omega=\Omega_{0}$ is known. To estimate
the probabilities in (i)--(iii), we only need to simulate data from
a VAR with the same implied values of $a$, $\lambda$ and $\Omega_{0}$.
Fixing $(a,\lambda)$, consider the bivariate VAR(1) with autoregressive
coefficient matrix $\Phi(a,\lambda)\defeq R(a)\Lambda(\lambda)L(a)^{\trans}$,
where
\begin{align*}
R(a)=\begin{bmatrix}R_{\lu}(a) & R_{\st}\end{bmatrix} & =\begin{bmatrix}a & 1\\
1 & 0
\end{bmatrix} & \Lambda(\lambda) & =\diag\{\lambda,\lambda_{\st}\} & L(a)^{\trans} & =\begin{bmatrix}L_{\lu}(a)^{\trans}\\
L_{\st}(a)^{\trans}
\end{bmatrix}\defeq R(a)^{-1},
\end{align*}
where we may take $\lambda_{\st}=0$. The implied quasi-cointegrating
relation is $\beta(a)^{\trans}=(1,-a)$. Then taking
\[
K(a)\defeq\begin{bmatrix}\beta(a)^{\trans}R_{\st}(1-\lambda_{\st})^{-1}L_{\st}(a)^{\trans}\\
L_{\lu}(a)^{\trans}
\end{bmatrix}
\]
we can ensure that the implied $\Omega$ in this model agrees with
$\Omega_{0}$, at the null value of $a$, by setting the variance
matrix $\Sigma$ of the reduced-form errors to
\[
\Sigma_{0}\defeq K(a_{0})^{-1}\Omega_{0}[K(a_{0})^{\trans}]^{-1}.
\]

Now for $(a,\lambda)\in\reals\times\set L$, let $Y_{n}^{(b)}(a,\lambda)\defeq\{y_{t}^{(b)}\}_{t=1}^{n}$
denote a sample of length $n$ (i.e.\ of the same length as the observed
sample) generated as
\[
y_{t}^{(b)}=\Phi(a,\lambda)y_{t-1}^{(b)}+\Sigma_{0}^{1/2}w_{t}
\]
where $w_{t}\distiid N[0,I_{p}]$, with $y_{0}^{(b)}=0$. For each
$(a,\lambda)$ lying on a discrete grid that contains the supports
of $F_{0}$ and $F_{1}$, we generate a total of $B=20,000$ such
samples. Let $\like_{n}^{\ast}[\PHI;Y_{n}^{(b)}(a,\lambda)]$ denote
the model loglikelihood (with the deterministic terms concentrated
out), computed on the basis of the data $Y_{n}^{(b)}(a,\lambda)$.
For another (or possibly the same) $(a^{\prime},\lambda^{\prime})\in\reals\times\set L$,
define
\[
\hat{\like}_{n}^{(b)}(a^{\prime},\lambda^{\prime}\mid a,\lambda)\defeq\max_{\{\PHI\in\set P\mid\Lambda_{\lu}(\PHI)=\lambda^{\prime},A(\PHI)=a^{\prime}\}}\like_{n}^{\ast}[\PHI;Y_{n_{0}}^{(b)}(a,\lambda)]
\]
to be the concentrated loglikelihood at $A(\PHI)=a^{\prime}$ and
$\Lambda_{\lu}(\PHI)=\lambda^{\prime}$. Then we can compute
\begin{align*}
\np_{n}^{(b)}(a_{0}\mid a,\lambda) & \defeq\indic\biggl\{\int_{\reals\times\set L}\exp\{\likens_{n}(a^{\prime},\lambda^{\prime}\mid a,\lambda)\}F_{1}(\deriv a^{\prime},\deriv\lambda^{\prime})\\
 & \qquad\qquad\qquad\qquad\qquad>\crit_{\alpha}\int_{\set L}\exp\{\likens_{n}(a_{0},\lambda^{\prime}\mid a,\lambda)\}F_{0}(a_{0},\deriv\lambda^{\prime})\biggr\}
\end{align*}
as the realisation of the nearly optimal test on the dataset $Y_{n}^{(b)}(a,\lambda)$.
Hence we can estimate the probabilities in (i)--(iii) above by replacing
each instance of $\Prob_{(a,\lambda)}\{\np_{n}(a_{0})=1\}$ with
\[
\frac{1}{B}\sum_{b=1}^{B}\np_{n}^{(b)}(a_{0}\mid a,\lambda)
\]
for each value of $(a,\lambda)\in\reals\times\set L$ (in practice,
for a discrete subset thereof).

Finally, since $\Omega_{0}$ is unknown, it needs to be consistently
estimated. To that end, we recognise that for $\err_{\lu,t}=L_{\lu}^{\trans}\err_{t}$
\[
\Omega=K\Sigma K^{\trans}=\lrvar\left(\begin{bmatrix}\beta^{\trans}x_{t}\\
\err_{\lu,t}
\end{bmatrix}\right),
\]
where the final equality follows in particular by observing that
\begin{align*}
\sum_{l=-\infty}^{\infty}\expect(\beta^{\trans}x_{t})\err_{\lu,t-l}^{\trans} & =\beta^{\trans}\err_{t}\err_{\lu,t}+\beta^{\trans}R_{\st}\Lambda_{\st}^{k}\sum_{l=1}^{\infty}\expect z_{\st,t-1}\err_{\lu,t-l}^{\trans}\\
 & =\beta^{\trans}\left[I+R_{\st}\Lambda_{\st}^{k}\sum_{l=0}^{\infty}\Lambda_{\st}^{l}L_{\st}^{\trans}\right]\Sigma L_{\lu}\\
 & =_{(1)}\beta^{\trans}R_{\st}\left[\Lambda_{\st}^{k-1}+\Lambda_{\st}^{k}(I-\Lambda_{\st})^{-1}\right]L_{\st}^{\trans}\Sigma L_{\lu}\\
 & =\beta^{\trans}R_{\st}\Lambda_{\st}^{k-1}(I-\Lambda_{\st})^{-1}L_{\st}^{\trans}\Sigma L_{\lu}\\
 & =_{(2)}\beta^{\trans}R_{\st}(I-\Lambda_{\st})^{-1}L_{\st}^{\trans}\Sigma L_{\lu},
\end{align*}
where $=_{(1)}$ follows by $R\Lambda^{k-1}L^{\trans}=I_{p}$ and
$\beta^{\trans}R_{\lu}=0$, and $=_{(2)}$ by $R\Lambda^{i}L^{\trans}=0$
for $i\in\{0,\ldots,k-2\}$, which are themselves implied by $\R\L^{\trans}=I_{kp}$
(see \lemref{GLR}). Hence we can estimate $\Omega$ by computing
ML estimates of $\beta$ and $L_{\lu}$ (under only the restriction
that $\Lambda_{\lu}(\PHI)\in\set L$), and then computed an estimate
(after demeaning and detrending) of the long-run covariance matrix
of $\hat{\beta}^{\trans}y_{t}$ and $\hat{L}_{\lu}^{\trans}\hat{\err}_{t}$,
where $\hat{\err}_{t}$ denote the VAR residuals.

\end{document}